\documentclass{article}
\usepackage[ascii]{inputenc}

\usepackage[english]{babel}
\usepackage[tbtags]{amsmath}
\usepackage{graphicx}
\usepackage[colorlinks=true, allcolors=blue]{hyperref}
\usepackage{amsthm}
\usepackage{cite}
\usepackage{comment}
\usepackage{enumitem}
\usepackage{mathtools}
\usepackage{amsfonts}
\usepackage{amssymb}
\usepackage {tikz}
\usetikzlibrary {positioning}
\tikzset{main node/.style={circle,fill=blue!40,draw,minimum size=1.5cm,inner sep=0pt},}
\tikzset{small dot/.style={fill=blue,circle,scale=1.5}}
\tikzset{small dot two/.style={fill=red,circle,scale=1.5}}

\tikzset{small dot blue/.style={fill=blue!40,circle,scale=1.5}}
\usepackage{mathabx}

\tikzset{small dot 3/.style={fill=blue!40,circle,scale=1.25}}

\usetikzlibrary{arrows}
\usepackage{blkarray}

\usepackage{authblk}

\usepackage {tikz}
\usetikzlibrary {positioning}
\tikzset{small dot new/.style={fill=blue!40,circle,scale=1.2}}
\tikzset{small dot new 2/.style={fill=black,circle,scale=1.2}}
\tikzset{small dot new 3/.style={fill=red!40,circle,scale=1.2}}

\tikzset{gas/.style={fill=blue!40,circle,scale=0.8}}

\tikzset{smaller dot new 2/.style={fill=black,circle,scale=0.3}}

\usetikzlibrary{patterns}
\usepackage{subcaption}

\usepackage{braket}

\usetikzlibrary{matrix}

\newtheorem{prop}{Proposition}

\newtheorem{coro}{Corollary}
\newtheorem{rem}{Remark}
\newtheorem{defi}{Definition}
\newtheorem{lemma}{Lemma}
\newtheorem{theo}{Theorem}
\newtheorem{conj}{Conjecture}

\usepackage{blindtext}

\usetikzlibrary{decorations.pathreplacing,calligraphy}

\title{Universality and classification of elementary thermal operations}
\author{Pedro Hack$^1$ and Christian B.~Mendl$^{1,2}$}
\date{\today\\[2ex]%
    {\small $^1$Technical University of Munich, School of CIT, Department of Computer Science, Boltzmannstra{\ss}e 3, 85748 Garching, Germany\\%
    $^2$Technical University of Munich, Institute for Advanced Study,\\ Lichtenbergstra{\ss}e 2a, 85748 Garching, Germany\\[2ex]%
    }}

\begin{document}

\maketitle

\begin{abstract}
Elementary thermal operations are thermal operations that act non-trivially on at most two energy levels of a system at the same time. They were recently introduced in order to bring thermal operations closer to experimental feasibility. A key question to address is whether any thermal operation could be realized via elementary ones, that is, whether elementary thermal operations are \emph{universal}. This was shown to be false in general, although the extent to which elementary thermal operations are universal remained unknown. Here, we characterize their universality in both the sense described above and a weaker one, where we do not require them to decompose any thermal operation, but to be able to reproduce any input-output pair connected via thermal operations. Moreover, we do so for the two variants of elementary thermal operations that have been proposed, one where only deterministic protocols are allowed and one where protocols can be conditioned via the realization of a random variable, and provide algorithms to emulate thermal operations whenever their elementary counterparts are (weakly or not) universal. Lastly, we show that non-deterministic protocols reproduce thermal operations better than deterministic ones in most scenarios, even when they are not universal. Along the way, we relate elementary thermal operations to random walks on graphs.
\end{abstract}

\section{Introduction}
The fundamental tools we are interested in here are the \emph{thermal operations} \cite{horodecki2013fundamental,lostaglio2018elementary,brandao2013resource}. Thermal operations are quantum channels that, given a system in state $\rho$ with Hamiltonian $H_S$ that is in contact with a heat bath with Hamiltonian $H_B$, 
take the form
\begin{equation}
\label{quan to}
    \mathcal E (\rho) = \text{Tr}_B \left[U \left(\rho \otimes \frac{e^{-\beta H_B}}{{\text{Tr}\left(e^{-\beta H_B}\right)}}\right) U^\dagger \right],
\end{equation}
where $U$ is an energy-preserving unitary 
($[U,H_S+H_B] = 0$), $\text{Tr}_B$ is the partial trace over the heat bath and $\beta$ is the inverse temperature.\footnote{The restriction to an energy-preserving unitary is justified via the notion of passivity \cite{lostaglio2018elementary,pusz1978passive,lenard1978thermodynamical}.}

The study of thermal operations has been pursued in a field known as \emph{resource theory}. The fundamental ideas in resource theory appeared in the chemical physics literature and are originally due to Ruch and collaborators \cite{ruch1975diagram,ruch1976principle,mead1977mixing,ruch1978mixing}, Alberti and Uhlmann \cite{alberti1981dissipative,alberti1982stochasticity} and Zylka \cite{zylka1985note}. Despite this, some of the basic tools were introduced outside the realm of physics and are still used in areas ranging from information theory to economics \cite{arnold2007majorization,arnold2018majorization}, with standard mathematical references in the field being the books by Marshall et al. \cite{marshall_inequalities:_2011} and Bhatia \cite{bhatia2013matrix}. Regarding physics, an overview of key results and areas of application can be found in the recent review by Lostaglio \cite{lostaglio2019introductory} and the book by Sagawa \cite{sagawa2022entropy}. The fundamental contributions to resource theory include \cite{horodecki2013fundamental,brandao2015second,gour2018quantum}, with some of the recent advances on the topic being \cite{alhambra2016fluctuating,mazurek2019thermal,shiraishi2020two,lostaglio2022continuous,surace2023state,czartowski2023thermal}. Regarding applications, algorithmic cooling should be highlighted \cite{lostaglio2018elementary,alhambra2019heat,korzekwa2022optimizing}.

The conditions under which some $\rho'$ can via achieved from $\rho$ via some thermal operation $\mathcal E$, $\rho'=\mathcal E (\rho)$, constitute a long-standing challenge in general. However, whenever $\rho$ or $\rho'$ are energy incoherent, the challenge simplifies and the transitions between $\rho$ and $\rho'$ are fully characterized in terms of their associated population vectors $p$ and $p'$. In particular, we can translate the question regarding transitions to the \emph{thermo-majorization} relation on probability distributions. More specifically, $\rho'$ can be achieved from $\rho$ provided $p'$ is thermo-majorized by $p$ \cite{horodecki2013fundamental,lostaglio2015quantum}. Moreover, the latter condition is known to be equivalent to the existence of some 
stochastic matrix $G$ such that $p'=Gp$ and $G g = g$, where
\begin{equation*}
    g = \frac{1}{Z} \left( e^{-\beta E_1}, \dots, e^{-\beta E_{|\Omega|}}\right)
\end{equation*}
is the Gibbs distribution associated to $H_S$ and $Z = \sum_{i=1}^{|\Omega|}  e^{-\beta E_i}$ its partition function.
That is, $G$ is a $g$-stochastic matrix that maps $p$ to $p'$.\footnote{In general, for $d \in \mathbb P_\Omega$, we denote $M \in \mathcal M_{|\Omega|, |\Omega|}(\mathbb R)$ as $d$-stochastic provided it is stochastic, i.e., 
\begin{equation*}
    \sum_{i=1}^{|\Omega|} M_{i,j}=1 \text{ for } 1 \leq j \leq |\Omega|,
\end{equation*}
and has $d$ as equilibrium distribution $M d =d$ \cite{marshall_inequalities:_2011}. Here, $\mathbb P_\Omega$ stands for the set of probability distributions over some finite set $\Omega$ and 
$\mathcal M_{|\Omega|, |\Omega|}(\mathbb R)$ for the set of $|\Omega| \times |\Omega|$ matrices with real entries.} Lastly, thermo-majorization relations provide fundamental constraints for general thermodynamic transitions since they characterize the population dynamics induced by arbitrary thermal operations \cite{lostaglio2018elementary,horodecki2013fundamental,korzekwa2016coherence}.

In this context, a natural subset of the $g$-stochastic matrices considered in resource theory is the one where we only allow $G$ to act non-trivially on at most two elements of the state space, the so-called \emph{elementary thermal operations} \cite{lostaglio2018elementary}. Theoretically, this subset of stochastic matrices has already been considered in the study of the fundamental notion of uncertainty known as majorization for its simplicity \cite{marshall_inequalities:_2011,hardy1952inequalities,bhatia2013matrix}. However, from a resource-theoretic point of view, the recent interest in elementary thermal operations comes from experimental considerations. To draw an analogy, the situation is quite similar to the decomposition of unitary matrices in terms of circuits involving at most two-level gates \cite{reck1994experimental}, which was pursued given that experimental implementations of such gates via optical devices are well-known \cite{yurke19862}.\footnote{Regarding the formal relation between both questions, it is easy to see that the only stochastic matrices that are unitary are the permutation matrices. ($M \in \mathcal M_{|\Omega|, |\Omega|}(\mathbb R)$ is a \emph{permutation} matrix if it has a single one in each row and column with the remaining entries being zero \cite{marshall_inequalities:_2011}.)} 
In our resource-theoretic context, the main issue is that universal sets of thermal operations, like the so-called \emph{crude} operations \cite{perry2018sufficient}, often require control over the global system and heat bath. The fact that this issue is considerably improved upon when using elementary thermal operations, which can be well reproduced via the collision or the Jaynes-Cummings model \cite{lostaglio2018elementary,jaynes1963comparison,ziman2005description}, motivated their recent introduction in resource theory \cite{lostaglio2018elementary}. 

Following the analogy with \cite{reck1994experimental}, a fundamental question regarding elementary thermal operations is to what extent they can be used to implement thermal operations, that is, whether they are universal or not. A question analogous to that in \cite{reck1994experimental} would be whether every $g$-stochastic matrix can be decomposed as a product of elementary thermal operations. A more general question \cite{lostaglio2018elementary} would be whether the same holds if we take several sets of products of elementary thermal operations and  condition which one we use on the realization of some random variable (for instance, a coin toss). This is equivalent to asking whether every $g$-stochastic matrix can be decomposed as a convex combination of products of elementary thermal operations. (It should be noted that, while the incorporation of such random variable has the clear drawback that post-processing will be required, its advantages will become clear later on. For the moment, it suffices to say that its incorporation does not seem to add further restrictions regarding experimental applicability.)  


The fundamental questions above have not yet been fully answered, although some results are known. In this regard, \cite{lostaglio2018elementary} showed that elementary thermal operations (in the stronger form and, hence, in the weaker one without convex combinations), are not universal in general, that is, for any Gibbs distribution.
However, whenever the Gibbs distribution is uniform, it is well-known that they are universal \cite{muirhead1902some,birkhoff1946tres,hardy1952inequalities}. Despite these results, no characterization of the universality of elementary thermal operations in terms of $g \in \mathbb P_\Omega$ is known. Moreover, there is no characterization of universality in the weaker sense, where we do not necessarily require elementary thermal operations to decompose any thermal operation, but we ask whether we can find some elementary thermal operations connecting each pair $p,p' \in \mathbb P_\Omega$ that is connected by thermal operations. We will refer to this as \emph{weak universality}. Lastly, there is no characterization of the extent to which the incorporation of random variables (and, hence, post-processing) to elementary thermal operations offers an advantage regarding universality. These constitute our main questions. In summary, we ask:

\begin{enumerate}[label=(Q\arabic*)]
    \item When are elementary thermal operations \emph{weakly} universal?
    \item When are elementary thermal operations universal?
    \item When does the incorporation of  random variables to the elementary thermal operations offer an advantage regarding weak or non-weak universality?
\end{enumerate}
(A schematic representation of (Q2) can be found in Figure \ref{circuit diagram}. In particular, we represent the stronger framework where convex combinations are allowed.)

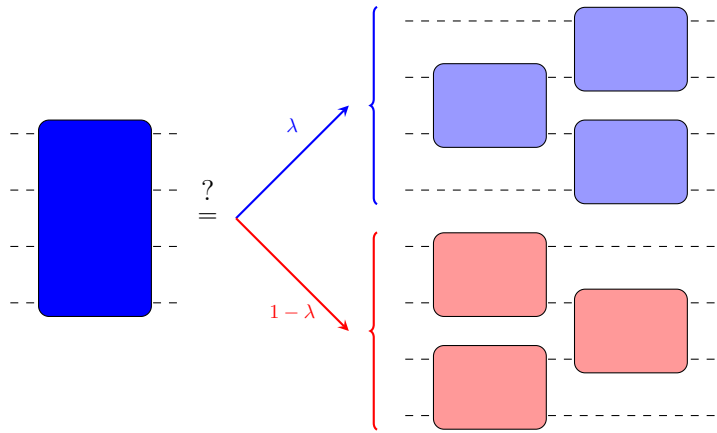
\begin{figure}[t]
\centering
\begin{tikzpicture} [scale=0.75, every node/.style={scale=0.75}, >=stealth]
\draw [dashed] (-0.5,2) -- (2.5,2);
\draw [dashed] (-0.5,1) -- (2.5,1);
\draw [dashed] (-0.5,0) -- (2.5,0);
\draw [dashed] (-0.5,-1) -- (2.5,-1);

\node[rounded corners, xshift=1cm, yshift=.5cm,, draw,fill=blue, text height = 3.25cm, minimum width = 2cm] {};

\node at (3, 0.5) {\Large $=$};
\node at (3, 1) {\Large $?$};

\draw [->, blue, thick] (3.5,0.5) -- node[ above =0.3cm, blue] {\textbf{$\lambda$}} (5.5,2.5);
\draw [->, red, thick] (3.5,0.5) -- node [below =0.3cm, red] {\textbf{$1-\lambda$}} (5.5,-1.5);

\draw [decorate,
    decoration = {brace},thick,blue] (6,0.75) --  (6,4.25);

\draw [decorate,
    decoration = {brace},thick,red] (6,-3.25) --  (6,0.25);

\draw [dashed] (6.5,4) -- (12,4);
\draw [dashed] (6.5,3) -- (12,3);
\draw [dashed] (6.5,2) -- (12,2);
\draw [dashed] (6.5,1) -- (12,1);

\node[rounded corners, xshift=8cm, yshift=2.5cm,, draw,fill=blue!40, text height = 1.25cm, minimum width = 2cm] {};

\node[rounded corners, xshift=10.5cm, yshift=3.5cm,, draw,fill=blue!40, text height = 1.25cm, minimum width = 2cm] {};

\node[rounded corners, xshift=10.5cm, yshift=1.5cm,, draw,fill=blue!40, text height = 1.25cm, minimum width = 2cm] {};

\draw [dashed] (6.5,0) -- (12,0);
\draw [dashed] (6.5,-1) -- (12,-1);
\draw [dashed] (6.5,-2) -- (12,-2);
\draw [dashed] (6.5,-3) -- (12,-3);

\node[rounded corners, xshift=8cm, yshift=-0.5cm,, draw,fill=red!40, text height = 1.25cm, minimum width = 2cm] {};

\node[rounded corners, xshift=8cm, yshift=-2.5cm,, draw,fill=red!40, text height = 1.25cm, minimum width = 2cm] {};

\node[rounded corners, xshift=10.5cm, yshift=-1.5cm,, draw,fill=red!40, text height = 1.25cm, minimum width = 2cm] {};
\end{tikzpicture}
\caption{Schematic representation of (Q2): When can thermal operations be decomposed as convex combinations of sequences of two-level thermal operations? In this representation, we ask whether a thermal operation acting on four energy levels (left of the equality) can be decomposed as a convex combination (consisting of two elements with weights $\lambda$ and $1-\lambda$ in the figure) of two sequences of thermal operations that act only on two energy levels simultaneously.}
\label{circuit diagram}
\end{figure}

As our main results, we answer (Q1) and (Q2) fully and make substantial progress on (Q3). First, we introduce the resource theories and polytopes we consider here in Section~\ref{to eto intro}. We then gather some known results regarding their relationship when the Gibbs distribution is uniform and some basic remarks for the following (Sections~\ref{sec: uniform} and \ref{sec: basics}). To gain some more intuition, we relate these resource theories to random walks on graphs in Section \ref{RW on graphs}. Question (Q1) is resolved in Section~\ref{rt relas}, and (Q2) in Section~\ref{poly relas}. Our progress regarding (Q3) can be found in Section \ref{advantage}. Lastly, we address the convexity of elementary thermal operations as a consequence of our work regarding (Q3), see Section~\ref{sec: convex weto}.

\section{Thermal and elementary thermal operations}
\label{to eto intro}

In this section, we introduce the thermal operations we are concerned with here. Note that, throughout this work, we consider $0 < d \in \mathbb P_\Omega$, the only physically relevant situation since, for $1 \leq i \leq |\Omega|$, $d_i=0$ implies $E_i$ is infinite. (The reader interested in the case $0 \leq d$ may find \cite{hartfiel1974study} useful.)
Moreover, some of the definitions in this section are usually written for the Gibbs distribution and, hence, do not include the partition function \cite{lostaglio2018elementary,lostaglio2019introductory}. We simply take some $0<d \in \mathbb P_\Omega$ and note this results in no meaningful difference. 

\subsection{Thermal operations}

We begin by introducing the well-known resource theory and polytope of thermal operations \cite{lostaglio2019introductory,lostaglio2018elementary}.

\begin{defi}[Resource theory and polytope of thermal operations (TO)]
If $0<d \in \mathbb P_\Omega$, then the resource theory of thermal operations allows all
transformations that can be realised by applying a $d$-stochastic matrix.

Analogously, the polytope of thermal operations $\mathcal P_{\text{TO}}(d)$ is the set of $d$-stochastic matrices $M \in \mathcal M_{|\Omega|, |\Omega|}(\mathbb R)$.
\end{defi}

Strictly speaking, a \emph{convex polytope} or \emph{polytope} for simplicity is the convex hull of a non-empty finite set \cite{brondsted2012introduction}. Although the polytope of thermal operations is indeed a polytope (see Section \ref{poly relas} or \cite{hartfiel1974study,jurkat1967term}), we will sometimes use the word \emph{polytope} loosely. We will return to this point later on.

We say that $q \in \mathbb P_\Omega$ can be \emph{achieved} from $p \in \mathbb P_\Omega$ via thermal operations (and analogously for the operations we will introduce later) if there exists a $d$-stochastic matrix $M \in \mathcal M_{|\Omega|, |\Omega|}(\mathbb R)$ such that $q=Mp$.

An easy way of determining whether a transition from $p$ to $q$ is possible is to use the
\emph{$d$-majorization curve}. For any $p \in \mathbb P_\Omega$, this curve can be constructed as follows:

\begin{enumerate}[label=(\alph*)]
    \item We consider a permutation $\Pi^d_p \in S_{|\Omega|}$ such that, for $1 \leq i < |\Omega|$,
    \begin{equation}
    \label{d-curve}
        \frac{p_{(\Pi^d_p)^{-1}(i)}}{d_{(\Pi^d_p)^{-1}(i)}} \geq \frac{p_{(\Pi^d_p)^{-1}(i+1)}}{d_{(\Pi^d_p)^{-1}(i+1)}}.
    \end{equation} 
    We call $\Pi^d_p$ the $d$-permutation of $p$. (This constitutes an abuse of language, since we should call it \emph{a} $d$-permutation whenever there is some equality in \eqref{d-curve}. However, we assume that $\Pi^d_p$ is unique in general and deal with the problematic scenario when needed, referring to it as \emph{uncertainty} in $\Pi^d_p$. Lastly, note that the most pathological case is $p=d$, since any permutation can be the $d$-permutation of $d$.) Moreover, we refer to $(\Pi^d_p)^{-1}$ as the $d$-order of $p$ and define 
    \begin{equation*}
        p^d \coloneqq \left( p_{(\Pi^d_p)^{-1}(i)} \right)_{i=1}^{|\Omega|}.
    \end{equation*}
    Furthermore, for $1 \leq i \leq |\Omega|$, we say $p_i$ is \emph{associated to $\alpha$} if $d_i=\alpha$ and extend this definition naturally to $p^d_i$.
    \item Plot the pairs of points 
    \begin{equation*}
        \left( \sum_{i=1}^k d_{(\Pi^d_p)^{-1}(i)} , \sum_{i=1}^k p_{(\Pi^d_p)^{-1}(i)} \right)_{k=1}^{|\Omega|}
    \end{equation*}
    together with $(0,0)$ and connect them piecewise linearly to form a concave curve. We call this curve the \emph{$d$-majorization curve} and denote it by $c_p^d$. It should be noted that the curve is also known as \emph{Lorenz $d$-curve} or \emph{thermo-majorization curve}.
\end{enumerate}

We say that $p$ \emph{$d$-majorizes} or \emph{thermo-majorizes} $q$, which we denote by $q \preceq_d p$, if $c_p^d$ is never below $c_q^d$. The reason we include these definitions is that it is well-known \cite{lostaglio2019introductory,lostaglio2018elementary,ruch1978mixing,ruch1980generalization} that, for any pair $p,q \in \mathbb P_\Omega$, there exists a $d$-stochastic matrix $M$ such that $q = M p$ if and only if $p$ $d$-majorizes $q$, $q \preceq_d p$. This will prove to be useful later on.

\subsection{Elementary thermal operations}

To introduce the resource theories and polytopes of elementary thermal operations, we need some more definitions. We will denote by $d^\downarrow$ the distribution that results from permuting the components of some $d \in \mathbb P_\Omega$ until they are arranged in non-increasing order.

It is not hard to see that, if $Q \in \mathcal M_{|\Omega|, |\Omega|}(\mathbb R)$ is a permutation matrix such that $d^\downarrow=Qd$ and $M\in \mathcal M_{|\Omega|, |\Omega|}(\mathbb R)$ is a $d$-stochastic matrix acting non-trivially on at most two levels, then
\begin{equation*}
    M^\downarrow \coloneqq Q M Q^T
\end{equation*}
is a $d^\downarrow$-stochastic matrix acting non-trivially on at most two levels.
Moreover, by definition, there exist $i,j$ with $1 \leq i \leq j \leq |\Omega|$ and $x,y,z,t \in [0,1]$ such that 
\begin{equation}
\label{g eto}
    M^\downarrow = \begin{pmatrix}
x & y\\
z & t
\end{pmatrix} \bigoplus \mathbb{I}_{\setminus (i,j)}
= \begin{pmatrix}
1- \gamma_{i,j} y & y \\
\gamma_{i,j} y & 1-y
\end{pmatrix} \bigoplus \mathbb{I}_{\setminus (i,j)}
= (1-y) \mathbb I + y P^{d^\downarrow}(i,j),
\end{equation}
where $\gamma_{i,j} \coloneqq d^\downarrow_j/d^\downarrow_i$ and 
\begin{equation*}
   P^{d^\downarrow}(i,j) \coloneqq \begin{pmatrix}
1- \gamma_{i,j} & 1 \\
\gamma_{i,j} & 0
\end{pmatrix}
\bigoplus \mathbb{I}_{\setminus (i,j)}.
\end{equation*}
(Note that in \eqref{g eto} we mean that $ M^\downarrow$ is identity except for $ M^\downarrow_{k,\ell}$ with $k,\ell \in \{ i,j\}$, where it takes the values we specify on the left of $\bigoplus$.)

As a result, we have the following decomposition of $M$:
\begin{equation}
\label{decomp d-swap}
    M= (1-y) \mathbb I + y Q^T P^{d^\downarrow}(i,j) Q.
\end{equation}

The decomposition in \eqref{decomp d-swap} motivates the following definition, which will be key in order to introduce the elementary resource theories we consider here.

\begin{defi}[$d$-swap and $T^d$-transform]
\label{d-swap td-trans}
If $0<d \in \mathbb P_\Omega$, $Q \in \mathcal M_{|\Omega|, |\Omega|}(\mathbb R)$ is a permutation matrix such that $d^\downarrow=Q d$ and $1 \leq i \leq j \leq |\Omega|$, then we call $P^d(i,j) \in  \mathcal M_{|\Omega|, |\Omega|}(\mathbb R)$ a \emph{$d$-swap} provided
\begin{equation}
 \label{eto decom}
        P^d(i,j) = Q^T P^{d^\downarrow}(i,j) Q.
    \end{equation}
    In the same scenario, we call 
    $T^d_\lambda(i,j) \in \mathcal M_{|\Omega|, |\Omega|}(\mathbb R)$ a $T^d$-transform provided there exists some $0 \leq \lambda \leq 1$ such that
    \begin{equation}
    \label{td-trans def}
        T^d_\lambda(i,j)= (1-\lambda) \mathbb I + \lambda P^d(i,j).
    \end{equation}
\end{defi}

As shown in \eqref{decomp d-swap}, $T^d$-transforms constitute all $d$-stochastic matrices that act non-trivially on at most two levels. Moreover, \eqref{decomp d-swap} also shows that any such matrix can be decomposed as a convex combination of $d$-swaps. (This is why we include the identity in our definition of $d$-swaps.)
 Moreover, since it uses the Gibbs distribution notation for some inverse temperature $\beta$, $d$-swaps are called $\beta$-swaps in \cite{lostaglio2018elementary}. (In fact, in \cite{lostaglio2018elementary}, $\beta$-swaps were only introduced for $d=d^\downarrow$.) $T^d$ transforms are sometimes referred to as \emph{simple} $d$-stochastic matrices \cite{marshall_inequalities:_2011} or as \emph{elementary thermal operations} \cite{lostaglio2018elementary}. We prefer, however, to follow the notation in \cite{marshall_inequalities:_2011,hardy1952inequalities}, where, when dealing with the uniform case $d=u \coloneqq (1/|\Omega|,\dots,1/|\Omega|)$, they refer to what we call $T^u$-transforms as $T$-transforms. Another name for $T^u$ transforms is \emph{elementary doubly stochastic matrices} \cite{marcus1984products}.

 Definition \ref{d-swap td-trans} allows us to introduce two pairs of elementary thermal resource theories and polytopes, which we call \emph{strong} and \emph{weak}.

\begin{defi}[Resource theory and polytope of weak elementary thermal operations (WETO)]
If $0<d \in \mathbb P_\Omega$, then the resource theory of weak elementary thermal operations $\text{RT}_{\text{WETO}}(d)$ allows all
transformations that can realised by sequential application of $T^d$-transforms.

Analogously, the polytope of weak elementary thermal operations $\mathcal P_{\text{WETO}}(d)$ is the set of matrices $M \in \mathcal M_{|\Omega|, |\Omega|}(\mathbb R)$ that can be decomposed as a sequence of $T^d$-transforms.\footnote{This constitutes an abuse of language since the WETO polytope is not convex in  general. (We give an example after the proof of Theorem \ref{equi to weto poly} and characterize its convexity in Corollary \ref{convex WETO}.)}
\end{defi}

Sometimes the literature refers to the fact that some $q \in \mathbb P_\Omega$ can be achieved from some $p \in \mathbb P_\Omega$ via the resource theory of weak elementary thermal operations by saying that $q$ is \emph{simply $d$-majorized} by $p$ \cite[Definition 14.B.2.a]{marshall_inequalities:_2011}.

\begin{defi}[Resource theory and polytope of strong elementary thermal operations (ETO) {\cite[Definition 1]{lostaglio2018elementary}}]
\label{defi ETO}
If $0<d \in \mathbb P_\Omega$, then the resource theory of strong elementary thermal operations $ \text{RT}_{\text{ETO}}(d)$ allows all
transformations that can realised by sequential application of $d$-swaps and convex combinations thereof.

Analogously, the polytope of strong elementary thermal operations $\mathcal P_{\text{ETO}}(d)$ is the set of matrices $M \in \mathcal M_{|\Omega|, |\Omega|}(\mathbb R)$ that can be decomposed as a convex combination of sequences of $d$-swaps.
\end{defi}

It should be noted that Definition \ref{defi ETO} and \cite[Definition 1]{lostaglio2018elementary} are equivalent. In \cite{lostaglio2018elementary}, what we call the \emph{resource theory of strong elementary thermal operations} is referred to as the \emph{resource theory of elementary thermal operations} and, moreover, $d$-swaps are substituted by $T^d$ transforms (which they call ETOs). In any case, it is easy to see that they are equivalent \cite{lostaglio2018elementary}.\footnote{The fact that our definition is included in theirs holds by definition and, for the converse, one simply ought to develop the products of $T^d$-transforms for each sequence in the convex combination and note that the result is indeed a convex combination of products of $d$-swaps.}

If we only allow the sequences of $d$-swaps to have length one, then we refer to Definition \ref{defi ETO} as the \emph{length one} resource theory (and polytope) of strong elementary thermal operations. We will begin by analyzing this simplified scenario in Section \ref{RW on graphs}, where we also connect it with random walks on graphs.

Since the main aim of this work is to study the relation between the resource theories and polytopes introduced in this section, we begin by considering the known results for the case where $d=u$ is uniform in the following section. (See also \cite[Section 2.3]{lostaglio2018elementary}.)

\section{Universality and weak universality for uniform equilibrium distribution}
\label{sec: uniform}

The first thing we ought to remark is that all resource theories coincide provided $d=u$. This follows directly from the following theorem, which was originnally developed by Muirhead \cite{muirhead1902some} and was generalized by Hardy et al. \cite{hardy1952inequalities} to encompass probability distributions.

\begin{theo}[Equivalence TO and WETO resource theories for uniform $d$ {\cite{muirhead1902some,hardy1952inequalities}}]
\label{equi to weto uniform}
If $d \in \mathbb P_\Omega$ is the uniform distribution $u$, then
    the thermal operations resource theory is equal to the weak elementary thermal operations 
resource theory
\begin{equation*}
    \text{RT}_{\text{TO}}(u) = \text{RT}_{\text{WETO}}(u).
\end{equation*}
\end{theo}

Moreover, since
\begin{equation*}
\text{RT}_{\text{WETO}}(d) \subseteq \text{RT}_{\text{ETO}}(d) \subseteq \text{RT}_{\text{TO}}(d)    
\end{equation*}
 for all $d \in \mathbb P_\Omega$ by definition, all thermal resource theories considered here are equivalent provided the equilibrium distribution is uniform.

When dealing with thermal operations and (strong) elementary thermal operations, the equivalence can be generalized to polytopes as a direct consequence of Birkhoff theorem \cite{birkhoff1946tres} (also sometimes attributed to Von Neumann \cite{von1953certain}) together with the fact that permutation matrices can be decomposed as products of transpositions.

\begin{theo}[Equivalence TO and ETO polytopes for uniform $d$ {\cite{birkhoff1946tres,von1953certain}}]
\label{birkhoff}
  If $d \in \mathbb P_\Omega$ is the uniform distribution $u$, then the polytope of thermal operations is equivalent to the polytope of elementary thermal operations
  \begin{equation*}
    \mathcal P_{\text{TO}}(u) = \mathcal P_{\text{ETO}}(u).
  \end{equation*}
\end{theo}

The equivalence at the polytope level, however, does not include weak elementary thermal operations, as shown by Marcus et al. \cite{marcus1984products}. (See also \cite{marshall_inequalities:_2011}.)

\begin{theo}[Equivalence TO and WETO polytopes for uniform $d$ {\cite[Theorem 1]{marcus1984products}}]
\label{uniform weto = to}
    If $d \in \mathbb P_\Omega$ is the uniform distribution $u$, then the following statements are equivalent:
     \begin{enumerate}[label=(\alph*)]
        \item The weak elementary thermal operations polytope is equivalent to the thermal operations polytope
         \begin{equation*}
    \mathcal P_{\text{TO}}(u) = \mathcal P_{\text{WETO}}(u).
\end{equation*}
        \item $|\Omega| = 2$.
        \end{enumerate}
\end{theo}

In the following, we address the generalization of Theorems \ref{equi to weto uniform}, \ref{birkhoff} and \ref{uniform weto = to} to arbitrary equilibrium distributions $0<d \in \mathbb P_\Omega$.
Before we start showing our main results, we make some basic remarks that will be useful in the future.

\section{Basic observations for the general case}
\label{sec: basics}

Recall that we only deal with the case where $0 < d \in \mathbb P_\Omega$. As a first remark, we note that, given an arbitrary $0 <d \in \mathbb P_\Omega$, it is equivalent to show the relation between the different resource theories and polytopes for $d^\downarrow$ than to do so for $d$. In order to show this, we use the following terminology
 \begin{equation*}
    \begin{split}
    \mathcal P(p) &\coloneqq \{\mathcal P_{\text{TO}}(p), \mathcal P_{\text{ETO}}(p), \mathcal P_{\text{WETO}}(p)\}, \\
     \text{RT}(p) &\coloneqq \{\text{RT}_{\text{TO}}(p), \text{RT}_{\text{ETO}}(p), \text{RT}_{\text{WETO}}(p)\}
    \end{split}
    \end{equation*}
for all $p \in \mathbb P_\Omega$.

\begin{lemma}[Equivalence under permutation of $d$]
\label{ordered d}
    Consider $0<d \in \mathbb P_\Omega$. If we take a pair of polytopes $A(d),B(d) \in \mathcal P(d)$ or resource theories $A(d),B(d) \in \text{RT}(d)$ of $d$ and the corresponding polytopes $A(d^\downarrow),B(d^\downarrow) \in \mathcal P(d^\downarrow)$ or resource theories $A(d^\downarrow),B(d^\downarrow) \in \text{RT}(d^\downarrow)$ of $d^\downarrow$, then the following statements are equivalent:
    \begin{enumerate}[label=(\alph*)]
        \item $A(d) \subseteq B(d)$.
        \item $A(d^\downarrow) \subseteq B(d^\downarrow)$.
        \end{enumerate}
\end{lemma}

\begin{proof}
Take $Q\in \mathcal M_{|\Omega|, |\Omega|}(\mathbb R)$ the permutation matrix such that $d^\downarrow = Q d$ and note that
    it suffices to show that $(b)$ implies $(a)$ since the converse follows analogously. Moreover,  for simplicity, we only consider two cases:
    \begin{enumerate}[label=(\Alph*)]
        \item $A(d^\downarrow)=\mathcal P_{\text{TO}}(d^\downarrow)$, $B(d^\downarrow)=\mathcal P_{\text{ETO}}(d^\downarrow)$ and we take the corresponding $A(d)$ and $B(d)$. If $M\in \mathcal M_{|\Omega|, |\Omega|}(\mathbb R)$ is a $d$-stochastic matrix,  then $M^\downarrow$ is a  $d^\downarrow$-stochastic matrix. Hence, by assumption, 
        \begin{equation*}
        \begin{split}
        M^\downarrow &= Q M Q^T = \sum_{k=0}^{k_0} \lambda_k \prod_{\ell=0}^{\ell_0} P^{d^\downarrow}(i_{k,\ell},j_{k,\ell}) \text{, and} \\
        M &= \sum_{k=0}^{k_0} \lambda_k \prod_{\ell=0}^{\ell_0} Q^T P^{d^\downarrow}(i_{k,\ell},j_{k,\ell}) Q,
        \end{split}
        \end{equation*}
        where $\sum_{k=0}^{k_0} \lambda_k=1$, and $\lambda_k>0$ and $1 \leq i_{k,\ell}<j_{j,\ell} \leq |\Omega|$ for $0 \leq k \leq k_0$ and $0 \leq \ell \leq \ell_0$.
        \item $A(d^\downarrow)=\text{RT}_{\text{TO}}(d^\downarrow)$, $B(d^\downarrow)=\text{RT}_{\text{ETO}}(d^\downarrow)$ and we take the corresponding $A(d)$ and $B(d)$. Given a pair $p, q \in \mathbb P_\Omega$ such that $q= Mp$ for some $d$-stochastic matrix $M\in \mathcal M_{|\Omega|, |\Omega|}(\mathbb R)$, we have that $Q q = Q M Q^T Q p$. By assumption, since $Q M Q^T$ is a $d^\downarrow$-stochastic matrix, we have
        \begin{equation*}
         \begin{split}
            Q q &= \left(\sum_{k=0}^{k_0} \lambda_k \prod_{\ell=0}^{\ell_0}  P^{d^\downarrow}(i_{k,\ell},j_{k,\ell}) \right) Q p \text{, and}\\
            q &= \left(\sum_{k=0}^{k_0} \lambda_k \prod_{\ell=0}^{\ell_0} Q^T  P^{d^\downarrow}(i_{k,\ell},j_{k,\ell}) Q \right) p,
            \end{split}
        \end{equation*}
        where $\sum_{k=0}^{k_0} \lambda_k=1$, and $\lambda_k>0$ and $1 \leq i_{k,\ell}<j_{j,\ell} \leq |\Omega|$ for $0 \leq k \leq k_0$ and $0 \leq \ell \leq \ell_0$.
        \end{enumerate}
        Since the other cases follow analogously, this concludes the proof.
\end{proof}

Note that Lemma \ref{ordered d} works more in general for any permutation of $d$, but we focus on $d^\downarrow$ since it is the one we will use later on.

As a second remark, it is straightforward to relate the different resource theories and polytopes for the smallest dimension, i.e. for $|\Omega|=2$. (When $|\Omega|=1$, then $|\mathbb P_\Omega|=1$.)  In particular, as a consequence of \eqref{decomp d-swap},
we have the following lemma.

\begin{lemma}
[Equivalence TO, ETO and WETO polytopes for $|\Omega| = 2$]
\label{equiv dim 2}
    If $0<d \in \mathbb P_\Omega$ and $|\Omega|=2$, then the polytopes of thermal and (weak and strong) elementary thermal operations are equivalent
    \begin{equation*}
        \mathcal P_{\text{TO}}(d) = \mathcal P_{\text{ETO}}(d)=\mathcal P_{\text{WETO}}(d).
    \end{equation*}
\end{lemma}

As a direct consequence of Lemma \ref{equiv dim 2}, all resource theories considered here coincide $\text{RT}_{\text{TO}}(d) = \text{RT}_{\text{ETO}}(d) = \text{RT}_{\text{WETO}}(d)$ provided $|\Omega|=2$.

As a last remark, the following straightforward property of elementary thermal operations will prove to be useful in the following, specially when justifying no-go results. 

\begin{lemma}[Monotonicity of support under ETO {\cite[Lemma 4]{lostaglio2018elementary}}]
\label{monot eto}
    If $0<d \in \mathbb P_\Omega$ and $p,q \in \mathbb P_\Omega$ such that $q \in \mathcal C_d^{ETO}(p)$, then $|\text{supp}(p)| \leq |\text{supp}(q)|$, where $\text{supp}(p) \coloneqq \{ i \in \Omega | p_i>0 \}$
    for all $p \in \mathbb P_\Omega$.
\end{lemma}

 Before addressing our main questions, and to gain some intuition regarding ETO, we consider a simpler scenario where they appear in the next section. In particular, we relate them to random walks on graphs.

\section{Random walks on complete graphs and the length one ETO polytope} 
\label{RW on graphs}

In the simpler scenario where we only allow sequences of $d$-swaps with length one, the resource theory of strong elementary thermal operations can be interpreted in terms of random walks on complete graphs. In this section, we introduce a new definition of random walk on a complete graph in Section \ref{new rw}, characterize the transition matrices that fulfill this definition (equivalently, we characterize the ETO polytope when we only allow sequences of $d$-swaps with length one) in Section \ref{charac rw} and relate our definition to the one usually used in the literature in Section \ref{usual rw}.  

\subsection{ETO random walks on complete graphs}
\label{new rw}

Random walks on graphs have been extensively studied \cite{levin2017markov,lovasz1993random,motwani1995randomized,xia2019random} and have become of increasing interest given their close connection to quantum walks \cite{szegedy2004quantum,childs2009universal}. Instead of focusing on the general case, we will only consider complete graphs. Take, hence, a complete undirected graph $G=(V_G,E_G)$ without loops, where $V_G$ are the vertices and $E_G$ the edges.\footnote{A graph is \emph{complete} provided there is an edge between any pair of distinct vertices. Moreover, a loop is an edge from a vertex to itself. Lastly, a graph is \emph{undirected} provided its edges have no direction. These definition are taken from \cite{levin2017markov}.}
Furthermore, assume we have some substance that is distributed among the vertices $V_G$ and think of the edges $E_G$ as \emph{channels} through which the substance can be redistributed among the vertices. As an example, when $|V_G|=3$, we can consider a gas in a \emph{toric} box (see Figure \ref{gas graph}). Lastly, we assume that the substance may prefer being in certain vertices over others and model this by incorporating some reference distribution over the vertices $d \in \mathbb P_{V_G}$ that may give preference to some subsets of $V_G$ over others. We can think of these preferences among vertices as signaling the existence of some \emph{marked} vertices, as one usually encounters in the study of decision trees \cite{farhi1998quantum}. The existence of preferences introduces asymmetries in some channels, as one can see in Figure \ref{complete graph}. In the context of the gas in a box, we can think of the introduction of preferences as surrounding it with a heat bath, with the preferences following a Gibbs distribution.

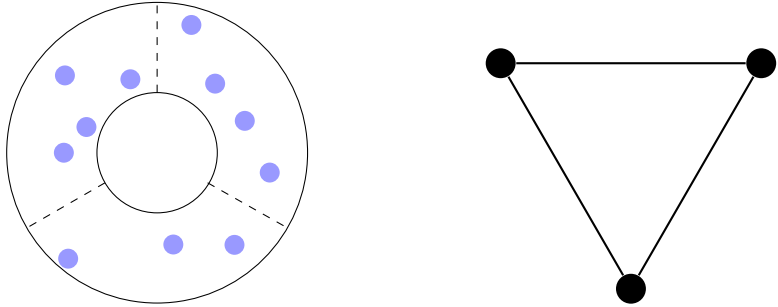
\begin{figure}[t]
\centering
\begin{subfigure}{0.4\textwidth}
  \centering
  \begin{tikzpicture}[scale=0.8]
    \draw (0,0) circle (2.5cm);
    \draw (0,0) circle (1cm);
    \draw [dashed] (0,1) -- (0,2.5);
    \draw [dashed] ({cos(210)},{sin(210)}) -- ({2.5*cos(210)},{2.5*sin(210)});
    \draw [dashed] ({cos( 33)},{sin(330)}) -- ({2.5*cos(330)},{2.5*sin(330)});
    \node[gas] at ({1.3*cos(110)},{1.3*sin(110)}) {};
    \node[gas] at ({2*cos(140)},{2*sin(140)}) {};
    \node[gas] at ({1.55*cos(180)},{1.55*sin(180)}) {};
    \node[gas] at ({1.25*cos(160)},{1.25*sin(160)}) {};
     
    \node[gas] at ({2.2*cos(75)},{2.2*sin(75)}) {};
    \node[gas] at ({1.5*cos(50)},{1.5*sin(50)}) {};
    \node[gas] at ({1.55*cos(20)},{1.55*sin(20)}) {};
    \node[gas] at ({1.9*cos(350)},{1.9*sin(350)}) {};

    \node[gas] at ({2.3*cos(230)},{2.3*sin(230)}) {};
    \node[gas] at ({1.55*cos(280)},{1.55*sin(280)}) {};
    \node[gas] at ({2*cos(310)},{2*sin(310)}) {};
  \end{tikzpicture}
\end{subfigure}
\hspace{0.1\textwidth}
\begin{subfigure}{0.4\textwidth}
  \centering
  \begin{tikzpicture}[scale=0.8]
   \node[small dot new 2] at (0,-2.5) (3) {};
   \node[small dot new 2] at ({2.5*cos(150)},{2.5*sin(150)}) (2) {};
   \node[small dot new 2] at ({2.5*cos(30)},{2.5*sin(30)}) (1) {};

   \path[draw,thick,-]
    (1) edge node {} (2)
    (2) edge node {} (3)
    (3) edge node {} (1);
  \end{tikzpicture}
\end{subfigure}
\caption{Gas in a toric box (left) and its associated complete graph (right).}
\label{gas graph}
\end{figure}

\begin{figure}[!t]
\centering
\begin{tikzpicture}[scale=0.8]
\node[small dot new 3,label={[yshift=0.01cm, thick,red!40, font=\fontsize{11}{11}\selectfont, thick]270:$d_0$}] at (-2,0) (1) {};
\node[small dot new 3,label={[yshift=0.01cm, thick,red!40, font=\fontsize{11}{11}\selectfont, thick]270:$d_0$}] at (2,0) (2) {};
\node[small dot new 2,label={[xshift=0.01cm, thick, font=\fontsize{11}{11}\selectfont, thick]180:$d_1$}] at (-4,2) (3) {};
\node[small dot new 2,label={[xshift=0.01cm, thick, font=\fontsize{11}{11}\selectfont, thick]0:$d_1$}] at (4,2) (4) {};
\node[small dot new 2,label={[yshift=0.01cm, thick, font=\fontsize{11}{11}\selectfont, thick]90:$d_1$}] at (-2,4) (5) {};
\node[small dot new 2,label={[yshift=0.01cm, thick, font=\fontsize{11}{11}\selectfont, thick]90:$d_1$}] at (2,4) (6) {};

\path[draw,thick,-]
    (3) edge node {} (6)
    (3) edge node {} (4)
    (3) edge node {} (5)
    (4) edge node {} (6)
    (5) edge node {} (6)
    (5) edge node {} (4)
    ;
\path[draw,thick,blue!60,-]
    (1) edge node {} (3)
    (2) edge node {} (4)
    (3) edge node {} (2)
    (5) edge node {} (2)
    (5) edge node {} (1)
    (6) edge node {} (2)
    (1) edge node {} (6)
    (1) edge node {} (4)
    ;
    \path[draw,thick,red!40,-]
    (1) edge node {} (2)
    ;
    \node[rounded corners, draw, dashed, red!40, text height = 1cm,minimum  width=6cm,yshift=-0.2cm]  (main) {};
\end{tikzpicture}
\caption{Complete $d$-graph $G$ with $|V_G|=6$ and $d^\downarrow=(d_0,d_0,d_1,\dots,d_1) \in \mathbb P_V$ with $d_1<d_0$. The black and red edges represent symmetric channels, and the blue edges asymmetric ones.}
\label{complete graph}
\end{figure}
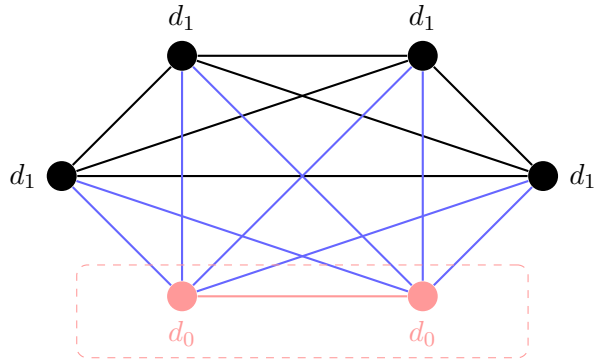

Our aim, and the reason why walks on graphs were introduced \cite{levin2017markov}, is to associate to $G$ a transition matrix $M$ such that, if we assume time to be discrete and consider some initial distribution of the substance over the edges $p_0 \in \mathbb P_{V_G}$, then $p_1 = M p_0 \in \mathbb P_{V_G}$ is a distribution that could have been obtained by allowing the substance to redistribute itself along the edges of $G$ after one time step. 
With this in mind, we define the following random walk on complete graphs.

\begin{defi}[ETO random walk on a complete graph]
\label{eto random walk}
If $G=(V_G,E_G)$ is a complete graph with some preferences among its vertices given by a distribution $d \in \mathbb P_{V_G}$, then an ETO random walk on $G$ is a matrix $M \in \mathcal M_{|\Omega|, |\Omega|}(\mathbb R)$ that can be decomposed as a convex combination of $d$-swaps. 
\end{defi}

Note that Definition \ref{eto random walk} coincides with the length one ETO polytope and, moreover, instead of assuming that the redistribution of substance takes place among all edges, it allows any distribution over the different edges in $E_G$ (including those in which some of them are not used). Furthermore, the identity can also be given a non-zero weight.

Before we continue, let us make a remark regarding the classical example of a gas in a box. In particular, we consider $d$ to be the uniform distribution and $|V_G| \geq 3$, as illustrated in Figure \ref{class gas} via its associated graph. Although this is not covered by ETO random walks (since the associated graph is not complete), one could treat this (or any non-complete graph) in an analogous way by taking into account the topology of the specific problem. (The only drawback being the potential lack of symmetry compared to the case of complete graphs.) It should be noted that one cannot take $M$ to be any $d$-stochastic matrix. In fact, in this scenario, permuting the gas concentration between non-adjacent compartments while leaving the rest unchanged would be allowed by doubly stochastic matrices. Such considerations led to the introduction of the \emph{molecular diffusion ordering} in \cite{hack2022majorization}, which corresponds to the transitions associated to non-homogeneous Markov chains whose constituent parts are precisely ETO random walks.

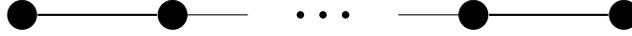
\begin{figure}
  \centering
  \begin{tikzpicture}
    \node[small dot new 2] at (0,0) (1) {};
    \node[small dot new 2] at (2,0) (2) {};
    \node[small dot new 2] at (6,0) (3) {};
    \node[small dot new 2] at (8,0) (4) {};

    \node[smaller dot new 2] at (3.7,0)  {};
    \node[smaller dot new 2] at (4,0)  {};
    \node[smaller dot new 2] at (4.3,0)  {};

    \draw  (2,0) -- (3,0);
    \draw  (5,0) -- (6,0);

    \path[draw,thick,-]
      (1) edge node {} (2)
      (3) edge node {} (4);
  \end{tikzpicture}
  \caption{Graph associated with the classical thermodynamic example of a gas in a box.}
  \label{class gas}
\end{figure}

\subsection{Characterizing the length one ETO polytope}
\label{charac rw}

As a first result, in the following theorem, we characterize the $d$-stochastic matrices $M$ that belong to the length one ETO polytope and, equivalently, the set of ETO random walks on a complete graph. Note that we say a $d$-stochastic matrix $M \in \mathcal M_{|\Omega|, |\Omega|}(\mathbb R)$ satisfies \emph{detailed balance} provided $M_{i,j}d_j=M_{j,i}d_i$ for $1 \leq i,j \leq |\Omega|$ \cite{levin2017markov}.

\begin{theo}[Characterization of the length one ETO polytope]
\label{length one ETO}
If $0<d \in \mathbb P_\Omega$ and $M \in \mathcal M_{|\Omega|, |\Omega|}(\mathbb R)$, 
then the following statements are equivalent:
\begin{enumerate}[label=(\alph*)]
        \item $M$ can be decomposed as a convex combination of $d$-swaps.
        \item $M^\downarrow$ is stochastic, it satisfies detailed balance and there exists some $\lambda \in [0,1]$ such that
\begin{equation}
\label{final}
    M^\downarrow_{i,i} = \lambda + \sum_{i<j \leq |\Omega|} \left(1-\frac{d^\downarrow_j}{d^\downarrow_i}\right) M^\downarrow_{i,j} + \sum_{\substack{1 \leq k<j \leq |\Omega|\\ k,j \neq i}} M^\downarrow_{k,j}
\end{equation} 
for $1 \leq i \leq |\Omega|$.
        \end{enumerate}
\end{theo}

\begin{proof}
We assume for simplicity that $d=d^\downarrow$ and, hence, $M=M^\downarrow$ throughout the proof.

    We begin by showing that $(a)$ implies $(b)$. In order to do so, we start by simply fixing a decomposition of $M \in \mathcal M_{|\Omega|, |\Omega|}(\mathbb R)$ is terms of $d$-swaps
    \begin{equation*}
    M = \lambda \mathbb I + \sum_{s=1}^{s_0} \lambda_s P^d(i_s,j_s),    
    \end{equation*}
    where $\lambda + \sum_{s=1}^{s_0} \lambda_s =1$, and, for $1 \leq s \leq s_0$, $\lambda, \lambda_s \geq 0$ and $1 \leq i_s<j_s \leq |\Omega|$. In this scenario, we immediately have that $M$ is stochastic and, moreover, it fulfills detailed balance since $P^d(i_s,j_s)$ does and it is the only $d$-swap that contributes to $M_{i_s,j_s}$ for all $1 \leq s \leq s_0$. Furthermore, it is easy to see that, for $1 \leq s \leq s_0$ and $1 \leq i \leq |\Omega|$, we have that 
    \begin{equation}
    \label{intermediate}
    P^d(i_s,j_s)_{i,i}= \sum_{i<j \leq |\Omega|} \left(1-\frac{d_j}{d_i}\right) P^d(i_s,j_s)_{i,j} + \sum_{\substack{1 \leq k<j \leq |\Omega|\\ k,j \neq i}} P^d(i_s,j_s)_{k,j}.   
    \end{equation}
 To conclude, it is easy to show that \eqref{final} holds using \eqref{intermediate}. We have
    \begin{equation*}
    \begin{split}
    (M)_{i,i}
    &= \lambda + \sum_{s=1}^{s_0} \lambda_s \left( \sum_{ i<j \leq |\Omega|} \left(1-\frac{d_j}{d_i}\right) P^d(i_s,j_s)_{i,j} + \sum_{\substack{1 \leq k<j \leq |\Omega|\\ k,j \neq i}} P^d(i_s,j_s)_{k,j} \right)\\
    &= \lambda + \sum_{i<j \leq |\Omega|} \left(1-\frac{d_j}{d_i}\right) \sum_{s=1}^{s_0} \lambda_s P^d(i_s,j_s)_{i,j} + \sum_{\substack{1 \leq k<j \leq |\Omega|\\ k,j \neq i}} \sum_{s=1}^{s_0} \lambda_s P^d(i_s,j_s)_{k,j}\\
    &= \lambda + \sum_{i<j \leq |\Omega|} \left(1-\frac{d_j}{d_i}\right) M_{i,j} + \sum_{\substack{1 \leq k<j \leq |\Omega|\\ k,j \neq i}} M_{k,j},
    \end{split}
\end{equation*}
 for $1 \leq i \leq |\Omega|$.
  
  We conclude by showing that $(b)$ implies $(a)$. In order to do so, it suffices to notice that any 
  stochastic 
  matrix $M \in \mathcal M_{|\Omega|, |\Omega|}(\mathbb R)$ fulfilling detailed balance and \eqref{final} can be decomposed in the following way:
  \begin{equation}
    \label{necessity}
        M= \lambda \mathbb I + \sum_{1 \leq k < j \leq |\Omega|} M_{k,j} P^d(k,j).
    \end{equation}
    This constitutes a convex combination of $d$-swaps since
\begin{equation*}
\begin{split}
        \lambda + \sum_{1 \leq k < j \leq |\Omega|} M_{k,j} &= \lambda + \sum_{2 \leq k < j \leq |\Omega|} M_{k,j} + \sum_{1 < j \leq |\Omega|} M_{1,j} \\
        &= M_{1,1} - \sum_{1<j \leq |\Omega|} \left(1-\frac{d_j}{d_1}\right) M_{1,j}
        + \sum_{1 < j \leq |\Omega|} M_{1,j} \\
        &= M_{1,1} + \sum_{1<j \leq |\Omega|} \frac{d_j}{d_1} M_{1,j}\\
        &= M_{1,1} + \sum_{1<j \leq |\Omega|} M_{j,1}\\
        &=1,
        \end{split}
    \end{equation*}
where we applied \eqref{final} in the second equality, detailed balance in the fourth, and the fact $M$ is stochastic in the last.

To conclude, we verify that \eqref{necessity} holds. This is the case since the following relations hold for $1 \leq k < j \leq |\Omega|$: $(A)$ $(M_{k,j} P^d(k,j))_{k,j}=M_{k,j}$ by definition, $(B)$ $(M_{k,j} P^d(k,j))_{j,k}=M_{j,k}$ by detailed balance, and $(C)$
 \begin{equation*}
    (M_{k,j} P^d(k,j))_{\ell,\ell} = \begin{cases}
    M_{k,j} &\text {if } \ell \neq k,j, \\
   \left(1- \frac{d_j}{d_k} \right) M_{k,j} &\text {if } \ell=k,\\
   0 &\text {if } \ell=j,
    \end{cases}
\end{equation*}
by definition. Hence, applying \eqref{final}, we have that \eqref{necessity} holds.
\end{proof}

\begin{rem}
    Note that Theorem \ref{length one ETO} provides a simple algorithm to determine whether some matrix belongs to the length one ETO polytope and, in case it does, it also returns the weights to decompose such a matrix in terms of $d$-swaps. Up to permutations in $d$, the algorithm can be summarized as follows:
    \begin{enumerate}[label=(\alph*)]
        \item Input $0<d \in \mathbb P_\Omega$ and $M \in \mathcal M_{|\Omega|, |\Omega|}(\mathbb R)$.
        \item Check $M$ is stochastic and satisfies detailed balance.
        \item Calculate $\lambda_i$ for $1 \leq i \leq |\Omega|$ following \eqref{final} and check $\lambda_i=\lambda$ for $1 \leq i \leq |\Omega|$ and $0 \leq \lambda \leq 1$.
        \item Output the decomposition of $M$ in terms of $d$-swaps: $\lambda$ times the identity plus $M_{i,j}$ times the $d$-swap acting non-trivially on the components $i$ and $j$ for all $i<j$. 
    \end{enumerate}
\end{rem}

If $|\Omega|=2$, note that any $d$-stochastic matrix $M$ satisfies detailed balance and \eqref{final} always holds with $\lambda=1-M_{1,2}$. Moreover,
in case $d \in \mathbb P_\Omega$ is the uniform distribution and for arbitrary $\Omega$, then detailed balance reduces to $M$ being symmetric and \eqref{final} to 
\begin{equation}
\label{simple thm 3}
    M_{i,i} = \lambda + \sum_{\substack{1 \leq k<j \leq |\Omega|\\ k,j \neq i}} M_{k,j}
\end{equation} 
for $1 \leq i \leq |\Omega|$ and $\lambda \in [0,1]$.

\subsection{Random walks on complete graphs}
\label{usual rw}

Although \emph{weighted random walks} have also been considered in the literature \cite[Chapter 9]{levin2017markov}, we relate here Definition \ref{eto random walk} with the usual definition on random walk only in the case when the reference distribution $d \in \mathbb P_{V_G}$ is uniform.
In this scenario, $M_0$ is the \emph{simple random walk} on (a complete graph) $G$ provided we have, for all $1 \leq i,j \leq |V_G|$, that
\begin{equation}
\label{simple RW}
    (M_0)_{i,j} = \begin{cases}
    \frac{1}{|V_G|-1}, &\text {if } i \neq j,\\
   0,&\text{if } i=j.
    \end{cases}
\end{equation}
Moreover, $M_1$ is the \emph{lazy random walk} on $G$ provided
\begin{equation}
    M_1= \frac{1}{2} \left( \mathbb I + M_0\right)
\end{equation}
for some simple random walk $M_0$.

Hence, the simple random walk assumes that the substance coming from one vertex is redistributed equally among all the vertices it is connected to (in our case, all of them) and the lazy random walk that half of it stays in the original vertex and the rest is redistributed equally among all the vertices it is connected to.


We can use Theorem \ref{length one ETO} to directly establish the relation between simple, lazy and ETO random walks on graphs, as we state in the following corollary.

\begin{coro}[Relation between random walks on complete graphs]
If $G=(V_G,E_G)$ is a complete undirected graph without loops, $d \in \mathbb P_{V_G}$ is the uniform distribution and $M\in \mathcal M_{|V_G|, |V_G|}(\mathbb R)$, then the following statements hold:
 \begin{enumerate}[label=(\alph*)]
\item If $M$ is the simple random walk on $G$, then $M$ is an ETO random walk on $G$ if and only if $|V_G| = 2$. 
\item If $M$ is the lazy random walk on $G$, then $M$ is an ETO random walk on $G$ if and only if $|V_G| \leq 4$.
\end{enumerate}
\end{coro}

\begin{proof}
   Both statements follow directly from \eqref{simple thm 3}. The first follows since we have
    \begin{equation*}
    M_{i,i}=0<\frac{1}{2}|V_G|-1= \sum_{\substack{1 \leq k<j \leq |V_G|\\ k,j \neq i}} M_{k,j} \iff 2 < |V_G|   
    \end{equation*}
    for $1 \leq i \leq |\Omega|$.
    The second follows since we have
    \begin{equation*}
    M_{i,i}=\frac{1}{2}<\frac{1}{2} \left(\frac{1}{2}|V_G|-1 \right)= \sum_{\substack{1 \leq k<j \leq |V_G|\\ k,j \neq i}} M_{k,j} \iff 4 < |V_G|   
    \end{equation*}
    for $1 \leq i \leq |\Omega|$.
\end{proof}

As a last remark, note that weighted random walks consider asymmetries between the different edges in the graph (with the motivation coming from electric networks), while ETO random walks make distinctions at the vertex level.

In the following section, we return to our main question of how the different thermal resource theories are related to each other.

\section{Weak universality of elementary thermal operations}
\label{rt relas}

In this section, we answer (Q1) by considering first non-determinisitic protocols in Section \ref{weak uni ETO} and then determinisitc ones in Section \ref{weak uni WETO}.

\subsection{Weak universality of strong elementary thermal operations}
\label{weak uni ETO}

The first question we address is the relation between the thermal and elementary thermal resource theories. We begin by extending \cite[Corollary 5]{lostaglio2018elementary}, where it was shown that they do not always coincide.\footnote{The impossibility of extending Theorem \ref{birkhoff} in general was already mentioned by Veinott in \cite[p. 2]{veinott1971least}, which points to a manuscript called \emph{On $d$-majorization and $d$-Schur convexity} by the same author. The latter was not published according to \cite{marshall_inequalities:_2011}.}

\begin{prop}[Difference TO and ETO resource theories]
\label{difference to eto}
If $0<d \in \mathbb P_\Omega$, then 
the thermal operations resource theory is equal to the elementary thermal operations 
resource theory
\begin{equation*}
    \text{RT}_{\text{TO}}(d) = \text{RT}_{\text{ETO}}(d)
\end{equation*}
only if 
$d^\downarrow = (d_0,\dots,d_0,d_1)$.
\end{prop}

\begin{proof}
By Lemma \ref{ordered d}, it suffices to assume that $d=d^\downarrow$ throughout the proof.

    We show the result by contrapositive, that is, we assume that there is no pair $d_0,d_1 \in \mathbb R$ with $0<d_1 \leq d_0$ such that $d = (d_0,\dots,d_0,d_1)$.
    We split the proof in two cases where, for each of them, we construct a pair of distributions $p,q \in \mathbb P_\Omega$ such that $q$ can be attained from $p$ by means of thermal operations but not by elementary thermal operations. We consider $d_\alpha$ the largest component in $d$ and notice that, by assumption, there exist two entries of $d$, which we name for simplicity $d_\beta$ and $d_\gamma$, such that $d_\gamma \leq d_\beta < d_\alpha$.  (For simplicity, and w.l.o.g., we will assume that $d_\beta$ and $d_\gamma$ correspond, respectively, to the second and third entries of $d$. If that were not the case, we could follow the argument below taking the appropriate components of both $p$ and $q$.) Consider, hence, the following two cases:
\begin{enumerate}[label=(\Alph*)]
    \item $d_\alpha \geq d_\beta+d_\gamma$. In this case, we can follow the idea in \cite[Corollary 5]{lostaglio2018elementary}. In particular, we can take $q=(1,0,\dots,0)$ and $p=(0,a,b,0,\dots,0)$ with $d_\gamma/d_\alpha \leq b \leq (d_\alpha-d_\beta)/d_\alpha$ (this can be done since $d_\alpha \geq d_\beta+d_\gamma$ by assumption) and $a=1-b$. By Lemma \ref{monot eto}, $q$ cannot be achieved from $p$ via elementary thermal operations since it has a smaller support. However, it is easy to check that $q \preceq_d p$. 
\item $d_\alpha < d_\beta+d_\gamma$. In this case, we can take $q=(0,a,b,0,\dots,0)$ with $(d_\alpha-d_\beta)/d_\alpha \leq b \leq d_\gamma/d_\alpha$ (this can be done since $d_\alpha < d_\beta+d_\gamma$ by assumption) and $a=1-b$, and $p=(1,0,\dots,0)$. $q$ cannot be achieved from $p$ via elementary thermal operations since, given that $d_\alpha > d_i$ for any $i \neq \alpha$ by assumption, the action of any elementary thermal operation on $p$
will result in a non-zero first component.
(In case we had multiple entries in $d$ equal to $d_\alpha$, the action of any elementary thermal operation on $p$
will leave
, at least, one non-zero component in the entries associated with those for which $d$ takes the value $d_\alpha$.) However, it is easy to check that $q \preceq_d p$. 
\end{enumerate}
This concludes the proof.
\end{proof}

\begin{rem}
    As a result of Proposition \ref{difference to eto}, even if we condition our experimental protocols using random variables, elementary thermal operations are not weakly universal (hence not universal) provided $d^\downarrow \neq (d_0,\dots,d_0,d_1)$.
\end{rem}

Note that, in agreement with Lemma \ref{equiv dim 2}, $d^\downarrow=(d_0,d_1)$ for any $d \in \mathbb P_\Omega$ if $|\Omega|=2$. Moreover, note that Proposition \ref{difference to eto} shows discrepancies between thermal and elementary thermal resource theories whenever $d^\downarrow = (d_0,d_1,\dots,d_1)$ with $d_0>d_1$. The counterexamples, however, do not extend to the case where $d^\downarrow = (d_0,\dots,d_0,d_1)$ with $d_0>d_1$. This follows from the asymmetry in the $d$-swaps whenever the the energy levels that are non-trivially affected are different.

Given that $\mathcal C^{ETO}_d(p)$ is a closed set for all $p,d \in \mathcal P_\Omega$ with $0<d$ \cite[Theorem 6]{lostaglio2018elementary}, Proposition \ref{difference to eto} actually shows that there exists some $q \in \mathbb P_\Omega$ such that there is no sequence of distributions $(q_\varepsilon)_\varepsilon \subseteq \mathbb P_\Omega$ which are $\varepsilon$-close to $q$ and achievable from $p$ via elementary thermal operations $q_\varepsilon \in \mathcal C^{ETO}_d(p)$.

Since it seems we cannot extend Proposition \ref{difference to eto} further, let us establish the equivalence between both resource theories for a low dimensional $\Omega$, where we can calculate everything explicitly. In fact, we can directly show the equivalence at the polytope level, as we do in the following proposition.

\begin{prop}[Equivalence TO and ETO polytopes for $|\Omega| = 3$]
\label{eto vs to}
If $0<d \in \mathbb P_\Omega$ and $|\Omega| = 3$, then
the following statements are equivalent:
\begin{enumerate}[label=(\alph*)]
     \item The thermal operations polytope is equal to the elementary thermal operations 
polytope
\begin{equation*}
    \mathcal P_{\text{TO}}(d) = \mathcal P_{\text{ETO}}(d).
\end{equation*}
    \item $d^\downarrow = (d_0,d_0,d_1)$.
 \end{enumerate}
\end{prop}

\begin{proof} 
By Lemma \ref{ordered d}, it suffices to assume that $d=d^\downarrow$ throughout the proof.

Sufficiency follows from Proposition \ref{difference to eto}
by contradiction. Assume both polytopes are equivalent for some $d \neq (d_0,d_0,d_1)$. By Proposition \ref{difference to eto}, there exists a pair $p,q \in \mathbb P_\Omega$ such that $q \preceq_d p$ and $q$ is not achievable from $p$ via elementary thermal operations. Since $q \preceq_d p$, there exists some $d$-stochastic matrix $M$ such that $q=Mp$ and, since the polytopes are equivalent,
\begin{equation*}
    q= \left(\sum_{k=0}^{k_0} \lambda_k \prod_{\ell=0}^{\ell_0} P^d(i_{k,\ell},j_{k,\ell}) \right) p,
\end{equation*}
with $\sum_{k=0}^{k_0} \lambda_k=1$, and $\lambda_k>0$ and $1 \leq i_{k,\ell}<j_{j,\ell} \leq |\Omega|$ for $0 \leq k \leq k_0$ and $0 \leq \ell \leq \ell_0$. This contradicts Proposition \ref{difference to eto}.

To prove necessity, since the polytope of elementary thermal operations is always contained inside that of thermal operations by definition, it suffices to note that all the extremal points of the polytope of thermal operations can be obtained as a product of elementary thermal operations. In particular, taking $\gamma = d_1/d_0$, the set of extremes of the polytope of thermal operations is $\{\mathbb I, A_1, \dots, A_9\}$, where 
      \begin{equation}
    \label{extremes poly dim 3}
  \begin{alignedat}{3}
&A_1=\begin{pmatrix}
 1-\gamma & 0 & 1 \\
\gamma & 1-\gamma & 0 \\
0 & \gamma & 0 
\end{pmatrix}, \text{ }&&A_2= \begin{pmatrix}
 1-\gamma & 0 & 1 \\
0 & 1 & 0 \\
\gamma & 0 & 0 
\end{pmatrix}, \text{ }&& A_3= \begin{pmatrix}
 0 & 1 & 0 \\
1-\gamma & 0 & 1 \\
\gamma & 0 & 0 
\end{pmatrix}, \\ \\
 &A_4=\begin{pmatrix}
 1-\gamma & \gamma & 0 \\
0 & 1-\gamma & 1 \\
\gamma & 0 & 0 
\end{pmatrix}, \text{ }&& A_5= \begin{pmatrix}
 1 & 0 & 0 \\
0 & 1-\gamma & 1 \\
0 & \gamma & 0 
\end{pmatrix}, \text{ }&& A_6=\begin{pmatrix}
 0 & 1 & 0 \\
1 & 0 & 0 \\
0 & 0 & 1 
\end{pmatrix},\\\\
 &A_7=\begin{pmatrix}
 \gamma & 1-\gamma & 0 \\
1-\gamma & 0 & 1 \\
0 & \gamma & 0 
\end{pmatrix}, \text{ }&& A_8=\begin{pmatrix}
 0 & 1-\gamma & 1 \\
1 & 0 & 0 \\
0 & \gamma & 0 
\end{pmatrix}, \text{ } &&  A_9=\begin{pmatrix}
 0 & 1-\gamma & 1 \\
1-\gamma & \gamma & 0 \\
\gamma & 0 & 0 
\end{pmatrix}.
  \end{alignedat}
\end{equation}
(This can be calculated using \cite{jurkat1967term} and \cite{hartfiel1974study}. The reader interested in how this is done can check Section \ref{poly relas}.)

To conclude, we simply notice that, aside from the identity and the elementary thermal operations acting on two levels, we have 
\begin{equation*}
     \begin{split}
     &A_1= P^d(2,3) P^d(1,3), \\
        &A_3= P^d(2,3) P^d(1,2), \\
& A_4= P^d(1,3) P^d(2,3),\\
& A_7=P^d(1,3) P^d(1,2) P^d(1,3),\\
&A_8=  P^d(1,2) P^d(2,3),\\
&A_9= P^d(2,3) P^d(1,2) P^d(2,3). 
 \end{split}
     \end{equation*}
     This concludes the proof.
\end{proof}

\begin{rem}
    As a result of Proposition \ref{eto vs to}, if $|\Omega|=3$, elementary thermal operations with conditional protocols are universal (and, as one can easily see, also weakly universal) if and only if $d^\downarrow = (d_0,\dots,d_0,d_1)$.
\end{rem}

 Proposition \ref{eto vs to} has a straightforward proof since we can explicitly calculate the extremal points in the TO polytope. As we will show in Theorem \ref{equi to eto poly}, we can use the tools in \cite{jurkat1967term,hartfiel1974study} and follow a similar strategy for arbitrary dimensions. However, we can use a simpler approach provided we are only interested in the relation between resource theories. To do so, we rely on the following lemma, where $S_{|\Omega|}$ denotes the set of permutations over $\Omega$.

\begin{lemma}[Extremes thermal cone {\cite[Lemma 12]{lostaglio2018elementary}}]
\label{losta lemma}
If $0<d \in \mathbb P_\Omega$ and $p,q \in \mathbb P_\Omega$ with $q \preceq_d p$, then $q$ can be written as a convex combination of elements in the set
$\{p^\Pi\}_{\Pi \in S_{|\Omega|}}$, where
\begin{enumerate}[label=(\alph*)]
    \item $x_i^\Pi \coloneqq \sum_{j=1}^i d_{\Pi^{-1}(j)}$, and $y_i^\Pi \coloneqq c_p^d(x_i^\Pi)$, and
    \item $p^\Pi_i \coloneqq y^{\Pi}_{\Pi(i)} - y^{\Pi}_{\Pi(i)-1} $, with $y_{0}\coloneqq0$,
\end{enumerate}
for $\Pi \in S_{|\Omega|}$ and $1\leq i \leq |\Omega|$.
\end{lemma}


We are now in position to characterize the equivalence between TO and ETO resource theories in general. To show this, we use the notation
\begin{equation}
\label{def future}
    \mathcal C_d(p) \coloneqq \{ q \in \mathbb P_\Omega | q \preceq_d p\}
\end{equation}
for any $p,d \in \mathbb P_\Omega$, $0<d$. (Later on, we will use the notation $\mathcal C^{ETO}_d(p)$ and $\mathcal C^{WETO}_d(p)$ for the natural extensions of \eqref{def future}.)

\begin{theo}[Relation TO and ETO resource theories]
\label{polytopes relation}
If $0<d \in \mathbb P_\Omega$,
then the following statements are equivalent:
\begin{enumerate}[label=(\alph*)]
     \item The thermal operations resource theory is equal to the elementary thermal operations 
resource theory
\begin{equation*}
    \text{RT}_{\text{TO}}(d) = \text{RT}_{\text{ETO}}(d).
\end{equation*}
    \item $d^\downarrow = (d_0,\dots,d_0,d_1)$.
 \end{enumerate} 
\end{theo}

\begin{proof}
By Lemma \ref{ordered d}, it suffices to assume that $d=d^\downarrow$ throughout the proof. For simplicity, we fix $n = |\Omega|$ and $\gamma = d_1/d_0$ throughout the proof.

Since sufficiency follows from Proposition \ref{difference to eto}, we only ought to prove necessity.
In order to do so, we will show that, for any pair $p,q \in \mathbb P_\Omega$, $q$ can be obtained from $p$ by elementary thermal operations provided the same holds for thermal operations. Instead of showing it directly, we can profit from Lemma \ref{losta lemma}, which provides
a finite set of distributions $\{p^\Pi | \Pi \in S_n\}$ such that, if $q$ can be obtained from $p$ by thermal operations, then $q$ can be decomposed as a convex combination of distributions in that set (which we denote by \emph{conv}),
\begin{equation}
    \mathcal C_d(p) = \text{conv} \{p^\Pi | \Pi \in S_n\}.
\end{equation}
Thus, it suffices to show that each distribution in the set is attainable by performing a sequence of elementary thermal operations on $p$ to conclude the proof.

Before we start the argument, note that we can assume w.l.o.g. that $p_i \geq p_{i+1}$ for $i=1,\dots,n-2$. (This is the case since $p$ and such a rearrangement of $p$ are equivalent in the $\preceq_d$ sense, and we can go from one to the other using permutations on the first $n-1$ components, which belong to the elementary thermal operations for $d=(d_0,\dots,d_0,d_1)$.)

To conclude the proof,
we take $m= \Pi^d_p(n)$ (in case of ambiguity, take the highest value possible)
and
we show that, for any permutation $\Pi \in S_n$, $p^\Pi$ is attainable from $p$ by elementary thermal operations. Fix, hence, such a permutation $\Pi_0 \in S_n$,
define $s = \Pi_0(n)$, and note that the relation between $s$ and $m$ determines the entries of $p^{\Pi_0}$. We distinguish, in particular, three cases:
   \begin{enumerate}[label=(\Alph*)]
    \item $m=s$. In this case, we have that 
    \begin{equation*}
    y_i^{\Pi_0} = \begin{cases}
    \sum_{j=1}^{i} p_j &\text {if } 1 \leq i<m ,\\
    p_n+\sum_{j=1}^{i-1} p_j &\text {if } m \leq i \leq n.
    \end{cases}
\end{equation*}
By definition of $s$, we have that $p^{\Pi_0}_n=y_s^{\Pi_0}-y_{s-1}^{\Pi_0}=p_n$. The rest of the components of $p^{\Pi_0}$ are $p_i$ for some $1 \leq i <n$ by definition. The order in which these components are presented is not important, since we can attain any arrangement from another one by elementary thermal operations. Hence, we can assume w.l.o.g. that $p^{\Pi_0}=p$ and, thus, $p^{\Pi_0}$ is attainable from $p$ via elementary thermal operations.

    \item $m<s$. In this case, we have that 
    \begin{equation*}
    y_i^{\Pi_0} = \begin{cases}
    \sum_{j=1}^{i} p_j &\text {if } 1 \leq i<m ,\\
    p_n+ (1-\gamma) p_i+\sum_{j=1}^{i-1} p_j &\text {if } m \leq i < s,\\
    p_n+\sum_{j=1}^{i-1} p_j &\text {if } s \leq i \leq n.
    \end{cases}
\end{equation*}
By definition of $s$, we have that $p^{\Pi_0}_n=y_s^{\Pi_0}-y_{s-1}^{\Pi_0}=\gamma p_{s-1}$. Following the argument for the case $m=s$, we can assume w.l.o.g. that
\begin{equation*}
\begin{split}
    p^{\Pi_0} = ( &p_1,\dots,p_{m-1},(1-\gamma)p_{m}+ p_n,(1-\gamma)p_{m+1}+ \gamma p_m, \dots,\\
    &(1-\gamma)p_{s-1}+ \gamma p_{s-2},p_s,\dots,p_{n-1}, \gamma p_{s-1}).
    \end{split}
\end{equation*}
We conclude the proof of this case noticing that 
\begin{equation*}
    p^{\Pi_0} = \left( \prod_{k=m}^{s-1} P^d(k,n) \right) p.
\end{equation*}

    \item $s<m$. In this case, we have that 
    \begin{equation*}
    y_i^{\Pi_0} = \begin{cases}
    \sum_{j=1}^{i} p_j &\text {if } 1 \leq i<s ,\\
    \gamma p_i+\sum_{j=1}^{i-1} p_j &\text {if } s \leq i < m,\\
    p_n+\sum_{j=1}^{i-1} p_j &\text {if } m \leq i \leq n.
    \end{cases}
    \end{equation*}
    By definition of $s$, we have that $p^{\Pi_0}_n=y_s^{\Pi_0}-y_{s-1}^{\Pi_0}= \gamma p_s$. Following the argument for the case $m=s$, we can assume w.l.o.g. that
\begin{equation*}
\begin{split}
    p^{\Pi_0} = ( &p_1,\dots,p_{s-1},(1-\gamma)p_{s}+ \gamma p_{s+1}, \dots,(1-\gamma)p_{m-2}+ \gamma p_{m-1},\\
    &p_n + (1-\gamma) p_{m-1},p_m,\dots,p_{n-1}, \gamma p_s).
    \end{split}
\end{equation*}
We conclude the proof of this case noticing that 
\begin{equation*}
    p^{\Pi_0} = \left( \prod_{k=1}^{m-s} P^d(m-k,n) \right) p.
\end{equation*}
    \end{enumerate}
    This concludes the proof.
\end{proof}

\begin{rem}
As a result of Theorem \ref{polytopes relation}, experimental protocols conditioned by a random variable are weakly universal if and only if $d^\downarrow = (d_0,\dots,d_0,d_1)$.
Furthermore, if $p,q \in \mathbb P_\Omega$ and $q \preceq_d p$, then Theorem \ref{polytopes relation} together with linear programming (which allows us to find the weights decomposing a point that belongs to the convex hull of some finite set \cite{boyd2004convex}) gives us an algorithm to achieve $q$ via (strong) elementary thermal operations on $p$ provided $d^\downarrow =(d_0,\dots,d_0,d_1)$. Up to permutations in $d$, the algorithm can be summarized as follows:
    \begin{enumerate}[label=(\alph*)]
        \item Input $0<d =(d_0,\dots,d_0,d_1) \in \mathbb P_\Omega$ and $p,q \in \mathbb P_\Omega$ such that $q \preceq_d p$.
        \item Calculate the decomposition of $q$ in terms of the extremes of $\mathcal C_d (p)$ using linear programming.
        \item Calculate $m$ and $s$ as in the proof of Theorem \ref{polytopes relation} for each extreme obtained in (b) and use them, following again Theorem \ref{polytopes relation}, to obtain a sequence of $d$-swaps that (up to permutations) yield the extreme when applied to $p$.    
        \item Find the sequence of swaps that are lacking in (c) in order to achieve each extreme via $d$-swaps. 
        \item Output a convex combination of products of $d$-swaps that yield $q$ when applied to $p$. 
    \end{enumerate}
\end{rem} 

 Theorem \ref{polytopes relation} illustrates that the general bounds on the number of $d$-swaps required to reach the extreme distributions in the ETO resource theory can be substantially improved whenever $d^\downarrow=(d_0,\dots,d_0,d_1)$. In particular, we have shown that any extreme of $\mathcal C_d(p)$ can be obtained using $|\Omega|$ $d$-swaps. This contrasts with the known general bounds, which escalate like $|\Omega|!$ \cite[Theorem 6]{lostaglio2018elementary}. Moreover, it should be noted that the decompositions in Theorem \ref{polytopes relation} coincide with those in Proposition \ref{eto vs to} when $|\Omega|=3$. This is the case since we have that $P^d(1,3) P^d(1,2) = P^d(1,2) P^d(2,3)$ and, hence, we can locate the sequence of swaps at the start.

\subsection{Weak universality of weak elementary thermal operations}
\label{weak uni WETO}

In order to study the relation between the the thermal and weak elementary thermal resource theories, we start again by characterizing the low-dimensional case, where we can calculate everything explicitly. 

\begin{prop}
    [Equivalence TO and WETO resource theories for $|\Omega| = 3$]
\label{Td transfrom low dim}
    If $0<d \in \mathbb P_\Omega$ and $|\Omega| = 3$, then the following statements are equivalent:
 \begin{enumerate}[label=(\alph*)]
     \item The thermal operations resource theory is equal to the weak elementary thermal operations 
resource theory
\begin{equation*}
    \text{RT}_{\text{TO}}(d) = \text{RT}_{\text{WETO}}(d).
\end{equation*}
    \item $d^\downarrow = (d_0,d_0,d_1)$.
 \end{enumerate}
\end{prop}

\begin{proof}
By Lemma \ref{ordered d}, it suffices to assume that $d=d^\downarrow$ throughout the proof.

The fact that $(a)$ implies $(b)$ holds as a direct consequence of Proposition \ref{difference to eto} for any finite $\Omega$. We show this by contrapositive: If the implication was false, then we could reach any distribution achievable by TO using finite sequences of $T^d$-transforms for some $d \neq (d_0,\dots,d_0,d_1)$. However, this would imply that any process achievable by TO is also achievable by ETO for such $d$, which contradicts Proposition \ref{difference to eto}.

   We show now that $(b)$ implies $(a)$. We take $\gamma = d_1/d_0$ and $p=(a,b,c)$, noting that we can assume $a \geq b$ w.l.o.g. (otherwise, we simply apply $P^d(1,2)$ to $p$ first and then follow the argument below), and $q \in \mathbb P_\Omega$ such that $q \preceq_d p$. Since $q \preceq_d p$ and $d = (d_0,d_0,d_1)$ with $0<d_1 \leq d_0$, we know from Proposition \ref{eto vs to} that $q$ is contained in the convex hull of $\{A_0 p, \dots, A_9 p\}$,
   \begin{equation*}
       q \in \mathcal C_d (p) = \text{conv} \{A_0 p, \dots, A_9 p\},
   \end{equation*}
    where $A_0 = \mathbb I$ and the others are defined as in Proposition \ref{eto vs to}. To conclude the proof, we will give a sequence of $T^d$-transforms that, when applied to $p$, yield $q$. We consider three cases (the cases where some equality holds follow easily from these):
   \begin{enumerate}[label=(\Alph*)]
\item $\gamma a > \gamma b > c$. In this case, $\mathcal C_d (p)$ is (roughly) given by Figure \ref{covex hull 1},
with $q$ being achievable by a sequence $T^d(1,2)T^d(2,3)$ if it lies below the dashed line and by $T^d(1,2)T^d(1,3)P^d(2,3)$ if it lies above.
\end{enumerate}

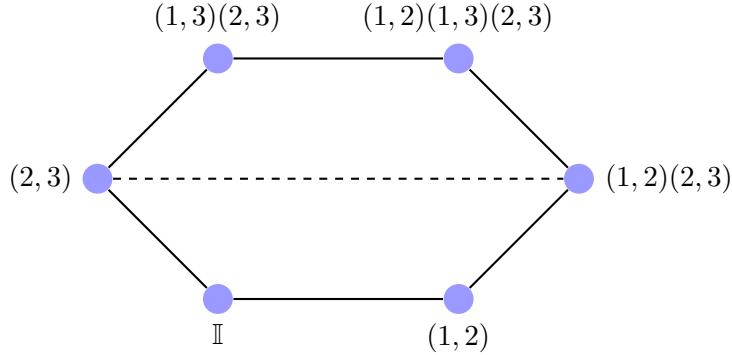
\begin{figure}[ht]
\centering
\begin{tikzpicture}[scale=0.8]
\node[small dot new,label={[yshift=0.01cm, thick, font=\fontsize{11}{11}\selectfont, thick]270:$\mathbb I$}] at (-2,0) (1) {};
\node[small dot new,label={[yshift=0.01cm, thick, font=\fontsize{11}{11}\selectfont, thick]270:$(1,2)$}] at (2,0) (2) {};
\node[small dot new,label={[xshift=0.01cm, thick, font=\fontsize{11}{11}\selectfont, thick]180:$(2,3)$}] at (-4,2) (3) {};
\node[small dot new,label={[xshift=0.01cm, thick, font=\fontsize{11}{11}\selectfont, thick]0:$(1,2) (2,3)$}] at (4,2) (4) {};
\node[small dot new,label={[yshift=0.01cm, thick, font=\fontsize{11}{11}\selectfont, thick]90:$(1,3) (2,3)$}] at (-2,4) (5) {};
\node[small dot new,label={[yshift=0.01cm, thick, font=\fontsize{11}{11}\selectfont, thick]90:$(1,2) (1,3) (2,3)$}] at (2,4) (6) {};
\path[draw,thick,-]
    (1) edge node {} (2)
    (1) edge node {} (3)
    (2) edge node {} (4)
    (3) edge node {} (5)
    (4) edge node {} (6)
    (5) edge node {} (6)
    ;
    \path[draw,dashed,thick,-]
    (3) edge node {} (4)
    ;
\end{tikzpicture}
\caption{Rough representation of the set $\mathcal C_d(p)$
for $d=(d_0,d_0,d_1)$ with $0<d_1 \leq d_0$ and $p=(a,b,c)$ with $\gamma a \geq \gamma b \geq c$. Note that we label the vertices by the (ordered) sequence of elementary thermal operations that we apply to $p$ to achieve them and that we substitute $P^d(i,j)$ by $(i,j)$.}
\label{covex hull 1}
\end{figure}

\begin{enumerate}[label=(\Alph*)]
\setcounter{enumi}{1}
\item $\gamma a > c >  \gamma b$. In this case, $\mathcal C_d (p)$ is (roughly) given by Figure \ref{covex hull 2}, with $q$ being achievable by a sequence $T^d(1,2)T^d(1,3)$ if it lies above the dashed line and by $T^d(1,2)T^d(2,3)$ if it lies below.
\end{enumerate}

\begin{figure}[ht]
\centering
\begin{tikzpicture}[scale=0.8]
\node[small dot new,label={[yshift=0.01cm, thick, font=\fontsize{11}{11}\selectfont, thick]270:$(2,3)$}] at (-2,0) (1) {};
\node[small dot new,label={[yshift=0.01cm, thick, font=\fontsize{11}{11}\selectfont, thick]270:$(1,2) (2,3)$}] at (2,0) (2) {};
\node[small dot new,label={[xshift=0.01cm, thick, font=\fontsize{11}{11}\selectfont, thick]180:$\mathbb I$}] at (-4,2) (3) {};
\node[small dot new,label={[xshift=0.01cm, thick, font=\fontsize{11}{11}\selectfont, thick]0:$(1,2)$}] at (4,2) (4) {};
\node[small dot new,label={[yshift=0.01cm, thick, font=\fontsize{11}{11}\selectfont, thick]90:$(1,3)$}] at (-2,4) (5) {};
\node[small dot new,label={[yshift=0.01cm, thick, font=\fontsize{11}{11}\selectfont, thick]90:$(1,2) (1,3)$}] at (2,4) (6) {};
\path[draw,thick,-]
    (1) edge node {} (2)
    (1) edge node {} (3)
    (2) edge node {} (4)
    (3) edge node {} (5)
    (4) edge node {} (6)
    (5) edge node {} (6)
    ;
     \path[draw,dashed,thick,-]
    (3) edge node {} (4)
    ;
\end{tikzpicture}
\caption{Rough representation of the set $\mathcal C_d(p)$
for $d=(d_0,d_0,d_1)$ with $0<d_1 \leq d_0$ and $p=(a,b,c)$ with $\gamma a \geq c \geq  \gamma b$. Note that we use the notation in Figure \ref{covex hull 1}.}
\label{covex hull 2}
\end{figure}
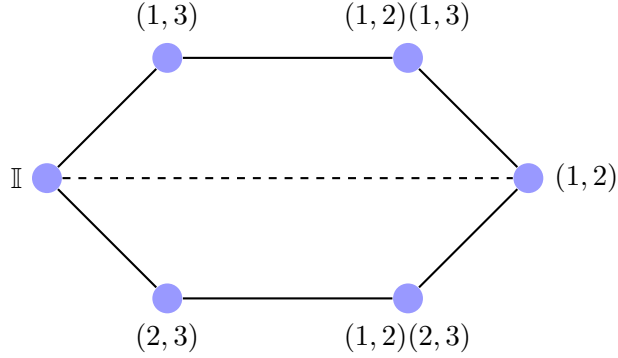

\begin{enumerate}[label=(\Alph*)]
\setcounter{enumi}{2}
\item $c > \gamma a > \gamma b$. In this case, $\mathcal C_d (p)$ is (roughly) given by Figure \ref{covex hull 3},
with $q$ being achievable by a sequence $T^d(1,2)T^d(1,3)$ if it lies above the dashed line and by $T^d(1,2)T^d(2,3)P^d(1,3)$ if it lies below.
   \end{enumerate}

   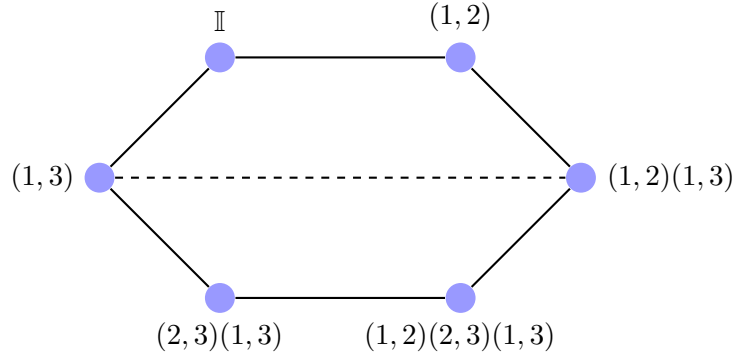
\begin{figure}[ht]
\centering
\begin{tikzpicture}[scale=0.8]
\node[small dot new,label={[yshift=0.01cm, thick, font=\fontsize{11}{11}\selectfont, thick]270:$(2,3) (1,3)$}] at (-2,0) (1) {};
\node[small dot new,label={[yshift=0.01cm, thick, font=\fontsize{11}{11}\selectfont, thick]270:$(1,2) (2,3) (1,3)$}] at (2,0) (2) {};
\node[small dot new,label={[xshift=0.01cm, thick, font=\fontsize{11}{11}\selectfont, thick]180:$(1,3)$}] at (-4,2) (3) {};
\node[small dot new,label={[xshift=0.01cm, thick, font=\fontsize{11}{11}\selectfont, thick]0:$(1,2) (1,3)$}] at (4,2) (4) {};
\node[small dot new,label={[yshift=0.01cm, thick, font=\fontsize{11}{11}\selectfont, thick]90:$\mathbb I$}] at (-2,4) (5) {};
\node[small dot new,label={[yshift=0.01cm, thick, font=\fontsize{11}{11}\selectfont, thick]90:$(1,2)$}] at (2,4) (6) {};
\path[draw,thick,-]
    (1) edge node {} (2)
    (1) edge node {} (3)
    (2) edge node {} (4)
    (3) edge node {} (5)
    (4) edge node {} (6)
    (5) edge node {} (6)
    ;
     \path[draw,dashed,thick,-]
    (3) edge node {} (4)
    ;
\end{tikzpicture}
\caption{Rough representation of the set $\mathcal C_d(p)$
for $d=(d_0,d_0,d_1)$ with $0<d_1 \leq d_0$ and $p=(a,b,c)$ with $c \geq \gamma a \geq \gamma b$. Note that we use the notation in Figure \ref{covex hull 1}.}
\label{covex hull 3}
\end{figure}
 This concludes the proof.  
\end{proof}

\begin{rem}
    As a result of Proposition \ref{Td transfrom low dim}, if $|\Omega|=3$, elementary thermal operations with deterministic protocols are weakly universal if and only if $d^\downarrow = (d_0,d_0,d_1)$.
\end{rem}

Note that we deal with resource theories in Proposition \ref{Td transfrom low dim} instead of polytopes as in Proposition \ref{eto vs to}. The reason for this will become clear later on.

The characterization in Proposition \ref{Td transfrom low dim} holds actually in general, as the following theorem shows.

\begin{theo}[Equivalence TO and WETO resource theories]
\label{wETO equiv}
    If $0<d \in \mathbb P_\Omega$,
then the following statements are equivalent:
 \begin{enumerate}[label=(\alph*)]
     \item The thermal operations resource theory is equal to the weak elementary thermal operations 
resource theory
\begin{equation*}
    \text{RT}_{\text{TO}}(d) = \text{RT}_{\text{WETO}}(d).
\end{equation*}
    \item $d^\downarrow = (d_0,\dots,d_0,d_1)$.
    \end{enumerate}
\end{theo}

\begin{proof}
By Lemma \ref{ordered d}, it suffices to assume that $d=d^\downarrow$ throughout the proof.

    The fact that $(a)$ implies $(b)$ holds analogously to its counterpart in Proposition \ref{Td transfrom low dim}, i.e., as a direct consequence of Proposition \ref{difference to eto}.

    To show that $(b)$ implies $(a)$, we begin taking $p,q \in \mathbb P_\Omega$ such that $q \preceq_d p$. We also take $n=|\Omega|$ and $\gamma=d_1/d_0$ for simplicity and note that we can assume w.l.o.g. that
    \begin{equation}
    \label{wlog}
        p_1 \geq p_2 \geq \cdots \geq p_{n-1} \text{ and } q_1 \geq q_2 \geq \cdots \geq q_{n-1}.
    \end{equation}
   (If that were not the case, we can first apply a sequence of $P^d(i,j)$ with $1 \leq i,j < n$ to reach the desired order for $p$, follow the argument below to reach $q$ with the desired order, and finally apply another sequence of $P^d(i,j)$ with $1 \leq i,j < n$ until we reach $q$.)
    
   Assuming, hence, the desired order for the components of $p$ and $q$, we will prove that $q$ can be achieved from $p$ via $T^d$-transforms by induction on
   \begin{alignat*}{3}
    h_0: \text{ } \mathbb P_\Omega \times& \mathbb P_\Omega &&\rightarrow &&
    \text{ }\{0,1,\dots,|\Omega|-1\}\\
    (p ,&q) &&\mapsto &&\text{ }|\Pi^d_p(n)-\Pi^d_q(n)|.
\end{alignat*}
   (For simplicity, whenever there is uncertainty in either $\Pi^d_p(n)$ or $\Pi^d_q(n)$, we assume they take the values that minimize $h_0(p,q)$.) 
We deal with the case $h_0(p,q)=0$ in the following lemma. (It should be noted that a more general version of Lemma \ref{same d-order} was obtained in \cite[Theorem 12]{perry2018sufficient} for partial level thermalizations, a subset of the weak elementary thermal operations.)

\begin{lemma}
\label{same d-order}
If $d \in \mathbb P_\Omega$ with $d^\downarrow=(d_0,\dots,d_0,d_1)$,
$p,q \in \mathbb P_\Omega$ with $q \preceq_d p$ and  $h_0(p,q)=0$, then there exists a finite sequence of $T^d$-transforms $(T^d_{\lambda_k}(i_k,j_k))_{k=1}^{k_0}$ such that 
    \begin{equation*}
        q = \left(\prod_{k=1}^{k_0} T^d_{\lambda_k}(i_k,j_k) \right)p,
    \end{equation*}
    where $0 \leq \lambda_k \leq 1$ and $1 \leq i_k < j_k \leq |\Omega|$ for $1 \leq k \leq k_0$.
\end{lemma}
\begin{proof}
The result follows by extending the argument in Lemma 2 of \cite[p. 47]{hardy1952inequalities}. More specifically, it follows by induction on
\begin{alignat*}{3}
h_1: \text{ }\mathbb P_\Omega \times& \mathbb P_\Omega &&\rightarrow &&
    \text{ }\{0,1,\dots,|\Omega|\} \\ (p,&q) &&\mapsto &&
    \text{ }|\{1 \leq i \leq n|p_i \neq q_i\}|.
\end{alignat*}

If $h_1(p,q)=0$, then $p=q$ and we have finished.

If $h_1(p,q)=s+1$ for some $s \geq 0$, then, since $\sum_{i=1}^n (p^d_i-q^d_i)=0$ and $\sum_{i}^\ell q^d_i \leq \sum_{i}^\ell p^d_i$ for $1 \leq \ell \leq n$, there exist some indexes $1 \leq k<l \leq n$ such that 
\begin{equation*}
    p^d_k>q^d_k,\text{ } p^d_{k+1}=q^d_{k+1},\text{ } \dots,\text{ } p^d_{l-1}=q^d_{l-1},\text{ } p^d_l<q^d_l.
\end{equation*}
We distinguish three cases:
\begin{enumerate}[label=(A.\arabic*)]
\item $p^d_k,q^d_k,p^d_l$ and $q^d_l$ are associated with $d_0$. In this case, we can follow the proof of Lemma 2 in \cite[p. 47]{hardy1952inequalities} and obtain by $T^d$-transforms on $p$ some $p'$ such that $h_1(p',q) \leq s$ and $q \preceq_d p'$.
\item $p^d_k$ and $q^d_k$ are associated with $d_0$ and $p^d_l$ and $q^d_l$ are associated with $d_1$. In this case, we have that $\gamma p^d_k > \gamma q^d_k \geq q^d_l > p^d_l$. Hence, there exists some $0 \leq \lambda_0 \leq 1$ such that $\lambda_0 \gamma p^d_k + (1- \lambda_0) p^d_l = q^d_l$. If $(1-\lambda_0 \gamma) p^d_k + \lambda_0 p^d_l \geq q^d_k$, then we take $p'= T^d_{\lambda_0}(k,l) p$. Otherwise, we consider some $0 \leq \lambda_1 < \lambda_0$ such that $(1-\lambda_1 \gamma) p^d_k + \lambda_1 p^d_l = q^d_k$ and we take $p'= T^d_{\lambda_1}(k,l) p$. In any case, $p'$ is obtained by applying $T^d$-transforms on $p$ and $h_1(p',q) \leq s$. Moreover, we have
\begin{equation*}
\begin{split}
    \sum_{i=1}^\ell (p')^d_i =  \sum_{i=1}^\ell p^d_i \geq \sum_{i=1}^\ell q^d_i &\text{ for } \ell=1,\dots,k-1,\\
    (p')^d_k \geq q^d_k, \text{ } (p')^d_\ell = p^d_\ell = q^d_\ell  &\text{ for } \ell=k+1,\dots,l-1,\\
    \sum_{i=1}^\ell (p')^d_i =  \sum_{i=1}^\ell p^d_i \geq \sum_{i=1}^\ell q^d_i &\text{ for } \ell=l,\dots,n.
\end{split}
\end{equation*}
Hence, since $p'$ and $q$ have the same $d$-order by construction, $q \preceq_d p'$.

\item $p^d_l$ and $q^d_l$ are associated with $d_0$ and $p^d_k$ and $q^d_k$ are associated with $d_1$. In this case, we have that $p^d_k > q^d_k \geq \gamma q^d_l > \gamma p^d_l$ and we can follow an argument analogous to that in $(A.2)$.
\end{enumerate}
This concludes the proof.
\end{proof}

As a result of Lemma \ref{same d-order}, 
to conclude, we only ought to show that, if $h_0(p,q)>0$, then, applying some $T^d$-transforms to $p$, we can obtain some $p'$ such that $q \preceq_d p'$ and $h_0(p',q)<h_0(p,q)$. We take, hence, $h_0(p,q)=m+1$ for some $m \geq 0$ and consider two cases:
    \begin{enumerate}[label=(B.\arabic*)]
\item $\Pi^d_p(n) > \Pi^d_q(n)$. 
In this scenario, we take $i_0$ the component for which $i_0+1 =\Pi^d_p(i_0)+1 = \Pi^d_p(n)$, note that $\gamma p_{i_0} > p_n$,
and define $p'=T^d_\lambda(i_0,n) p$ for some $0 \leq \lambda \leq 1$. We distinguish two cases:
\begin{enumerate}[label=(B.1.\arabic*)]
\item $h_0(p,q)>1$. In this case, we ought to see that there exists such a $\lambda$ fulfilling 
\begin{equation}
\label{case 1}
    \begin{split}
        &p_n' \geq \gamma p_{i_0}' \text{ and} \\
        &p_n' + \Delta \geq \gamma q_{i_0-1}, 
    \end{split}
\end{equation}
where we take $\Delta= \sum_{j=1}^{i_0-1} p_j - (q_n + (1-\gamma) q_{i_0-1}+ \sum_{j=1}^{i_0-2} q_j)$ and the first equation assures that $h(p',q) = m$ while the second one assures that $q \preceq_d p'$. (This is the case since the Lorenz $d$-curve of $p'$ coincides with that of $p$ except for the components that we are modifying, given that $\gamma p_{i_0} > p'_n$ and $\gamma p'_{i_0}>p_n$. Moreover, the conditions in \eqref{case 1} suffice to assure that, in the region where it differs from that of $p$, the Lorenz $d$-curve of $p'$ is not below that $q$.)

Isolating $\lambda$, the equations in \eqref{case 1} are equivalent to
\begin{equation}
\label{case 2}
    \begin{split}
        \lambda &\geq \frac{\gamma p_{i_0} - p_n}{\gamma (1+ \gamma) p_{i_0} - (1+ \gamma) p_n} \text{ and} \\
        \lambda &\geq \frac{\gamma q_{i_0-1} - (\Delta+p_n)}{\gamma p_{i_0} - p_n}, 
    \end{split}
\end{equation}
respectively. It is not difficult to see that the right hand side of both inequalities is bounded by $1$: For the first one we use that $\gamma p_{i_0} >p_n$ by assumption and for the second that $\gamma q_{i_0-1} \leq \gamma p_{i_0}+ \Delta$ since $q \preceq_d p$. Hence, we can find some $0 \leq \lambda \leq 1$ fulfilling \eqref{case 1}.
\item $h_0(p,q)=1$. This case is analogous to $(B.1.1)$, substituting \eqref{case 1} by 
\begin{equation}
\label{case 1b}
    \begin{split}
        &p_n' \geq \gamma p_{i_0}' \text{ and} \\
        &p_n' + \Delta \geq q_n, 
    \end{split}
\end{equation}
with $\Delta= \sum_{j=1}^{i_0-1} (p_j - q_j)$, and \eqref{case 2} by
\begin{equation*}
    \begin{split}
        \lambda &\geq \frac{\gamma p_{i_0} - p_n}{\gamma (1+ \gamma) p_{i_0} - (1+ \gamma) p_n} \text{ and} \\
        \lambda &\geq \frac{q_n - (\Delta+p_n)}{\gamma p_{i_0} - p_n}. 
    \end{split}
\end{equation*}
It is not difficult to see that the right hand side of both inequalities is bounded by $1$: The first follows like \eqref{case 2} and the second since $q \preceq_d p$ and, hence, $q_n \leq \gamma p_{i_0}+ \Delta$. Thus, there exists some $0 \leq \lambda \leq 1$ fulfilling \eqref{case 1b}.
\end{enumerate} 
\item $\Pi^d_p(n) < \Pi^d_q(n)$.
In this scenario, we take $i_0$ the component for which $\Pi^d_p(i_0)-1 = \Pi^d_p(n)$, note that $p_n>\gamma p_{i_0}$,
and define $p'=T^d_\lambda(i_0,n) p$ for some $0 \leq \lambda \leq 1$. We can conclude, analogously to $(B.1.1)$, by finding such a $\lambda$ fulfilling 
\begin{equation}
    \begin{split}
        &\gamma p_{i_0}' \geq p_n' \text{ and} \\
        &p_{i_0}' + \Delta \geq q_{i_0}, 
    \end{split}
\end{equation}
where we take $\Delta= \sum_{j=1}^{i_0-1} (p_j - q_j)$ and the first equation assures that $h(p',q) = m$ while the second one assures that $q \preceq_d p'$.
These equations are equivalent to
\begin{equation}
    \begin{split}
        \lambda &\geq \frac{p_n-\gamma p_{i_0}}{(1+ \gamma) p_n-\gamma (1+ \gamma) p_{i_0}} \text{ and} \\
        \lambda &\geq \frac{q_{i_0} - (\Delta+p_{i_0})}{p_n-\gamma p_{i_0}}, 
    \end{split}
\end{equation}
respectively. It is not difficult to see that the right hand side of both inequalities is bounded by $1$: For the first one we use that $\gamma p_{i_0} <p_n$ by assumption and for the second that $q_{i_0} \leq \Delta + p_n+ (1-\gamma) p_{i_0}$
since $q \preceq_d p$.
\end{enumerate}
   By induction, this concludes the proof.
\end{proof}

\begin{rem}
As a result of Theorem \ref{wETO equiv}, elementary thermal operations with deterministic protocols are weakly universal if and only if $d^\downarrow = (d_0,\dots,d_0,d_1)$.
     Hence, given the result in Theorem \ref{polytopes relation}, Theorem \ref{wETO equiv} shows that the conditioning our experimental protocols via random variables does not augment the cases where elementary thermal operations are weakly universal. Moreover, if $p,q \in \mathbb P_\Omega$ and $q \preceq_d p$, then Theorem \ref{wETO equiv} provides an algorithm to achieve $q$ via weak elementary thermal operations on $p$ provided $d^\downarrow =(d_0,\dots,d_0,d_1)$. Up to permutations in $d$, the algorithm can be summarized as follows:
    \begin{enumerate}[label=(\alph*)]
        \item Input $0<d =(d_0,\dots,d_0,d_1) \in \mathbb P_\Omega$ and $p,q \in \mathbb P_\Omega$ such that $q \preceq_d p$.
        \item While $h_0(p,q)>0$ and following Theorem \ref{wETO equiv}, calculate a recursive sequence of $T^d$-transforms that sequentially reduce $h_0$ and update $p$ by applying the sequence to it. 
        \item While $h_1(p,q)>0$ and following Theorem \ref{wETO equiv}, calculate a recursive sequence of $T^d$-transforms that sequentially reduce $h_1$ and update $p$ by applying the sequence to it.
        \item Output (in order) the sequence of $T^d$-transforms generated in (b) and (c). This sequence yields $q$ when applied to $p$. 
    \end{enumerate}
\end{rem}

Note that Theorem \ref{polytopes relation} follows as a direct corollary of Theorem \ref{wETO equiv}. As we will see in Section \ref{poly relas}, the tight relation between TO and WETO resource theories for $d^\downarrow= (d_0,\dots,d_0,d_1)$ breaks down at the polytope level. In fact, Theorem \ref{wETO equiv} cannot be extended to polytopes even for uniform $d$ by Theorem \ref{uniform weto = to}. This contrasts with the extension of Theorem \ref{polytopes relation}, which we will prove in Section \ref{poly relas}.

\section{Universality of elementary thermal operations}
\label{poly relas}

As a follow up to the previous section, we turn our attention to thermal polytopes, that is, to the universality of elementary thermal operations. Hence, in this section, we answer (Q2) by considering first non-determinisitic protocols in Section \ref{strong uni ETO} and then determinisitc ones in Section \ref{strong uni WETO}.

\subsection{Universality of strong elementary thermal operations}
\label{strong uni ETO}

 The first question we wish to answer is whether
Theorem \ref{polytopes relation} can be extended to polytopes, that is, we would like to know how does Proposition \ref{eto vs to} look when we consider $|\Omega| > 3$. In order to do so, we use the work by Jurkat and Ryser \cite{jurkat1967term} and Hartfiel \cite{hartfiel1974study}. We begin recalling an algorithm provided in \cite{jurkat1967term}.

\begin{defi}[Jurkat-Ryser $d$-algorithm and $d$-matrix {\cite{jurkat1967term}}]
\label{def: jr algo}
If $d \in \mathbb P_\Omega$, the Jurkat-Ryser $d$-algorithm is a procedure to construct a matrix of dimension $|\Omega| \times |\Omega|$ that begins with an empty $|\Omega| \times |\Omega|$ matrix $A_1$ and a couple of vectors $r^1=s^1=d$ and, for each step $m \geq 2$, does the following:
\begin{enumerate}[label=(\alph*)]
\item Select a position $(i_m,j_m)$ in $A_{m-1}$ that has not been assigned a value yet. If such a position does not exist, return $A_{m-1}$.
    \item Define $A_m$ as the matrix equivalent to $A_{m-1}$ with the addition of the $(i_m,j_m)$ entry, which equals $\min (r^m_i,s^m_j)$, and fill the rest of row $i$ (column $j$) with zeros provided $r^m_i = \min (r^m_i,s^m_j)$ ($s^m_j = \min (r^m_i,s^m_j)$).
    \item Define
    \begin{equation*}
        \begin{split}
        r^{m+1} &=(r^m_1,\dots,r^m_{i-1},r^m_i-\min (r^m_i,s^m_j),r^m_{i+1},\dots,r^m_{|\Omega|}), \\
        s^{m+1} &=(s^m_1,\dots,s^m_{j-1},s^m_j-\min (r^m_i,s^m_j),s^m_{j+1},\dots,r^m_{|\Omega|}).
        \end{split}
    \end{equation*}
   \item Return to step (a).
\end{enumerate}
We say a matrix $A$ is a Jurkat-Ryser $d$-matrix if it can be constructed following the Jurkat-Ryser algorithm with $r^1=s^1=d$. Moreover, for any $m \geq 1,$ we call $r^m$ and $s^m$ the $m$-th Jurkat-Ryser row and column $d$-vectors, respectively.
\end{defi}

The relevance of Definition \ref{def: jr algo} for our work here is encapsulated in the following theorem, where we denote by $\text{diag}(x_1,\dots,x_n)$ a diagonal matrix with entries $x_1,\dots,x_n$. 

\begin{theo}[Extremes TO polytope {\cite[Theorem 4.1]{jurkat1967term}, \cite[Lemma 1.1]{hartfiel1974study}}]
\label{extremes to}
    If $d \in \mathbb P_\Omega$, then the extreme points of the polytope of thermal operations take the form
    $A D^{-1}$, where $A$ is a Jurkat-Ryser $d$-matrix and $D=\text{diag}(d_1,d_2,\dots,d_{|\Omega|})$.
\end{theo}

\begin{proof}
    Most of the theorem is due to Jurkat and Ryser \cite[Theorem 4.1]{jurkat1967term} and its final form due to Hartfiel \cite[Lemma 1.1]{hartfiel1974study}. The only thing we ought to notice is that we have to take the transpose of the result in \cite{hartfiel1974study} (since it uses a different convention for stochastic matrices) and that, if a matrix $A$ can be constructed via the Jurkat-Ryser algorithm, then its transpose $A^T$ can as well. (We simply follow the steps in the construction of $A$ transposing the indexes we use in each of them.)
\end{proof}

Note that Theorem \ref{extremes to} proves that the TO polytope is in fact a polytope in the sense of \cite{brondsted2012introduction}, since it is clearly convex.


If $A$ is a Jurkat-Ryser $d$-matrix such that $M= A D^{-1}$ for some extreme of the TO polytope, then we say $A$ is \emph{associated} to $M$.
In order to determine the relation between the TO and ETO polytopes, a couple more definitions will prove to be useful. The aim of these definitions is to keep track of the previous choices in the Jurkat-Ryser $d$-algorithm, which will allow us to determine the future entries of a Jurkat-Ryser $d$-matrix. Hence, we begin by defining the \emph{history} of a Jurkat-Ryser $d$-matrix.

\begin{defi}[History]
If $d \in \mathbb P_\Omega$ and $A$ is a Jurkat-Ryser $d$-matrix,
we call a sequence of ordered pairs $((i_{k_l}^A,j_{k_l}^A))_{l=1}^{l_0}$ through which $A$ was constructed the history of $A$ and denote it by $H(A)$. Lastly, if $1 \leq m \leq l_0$, we call the subsequence $((i_{k_l}^A,j_{k_l}^A))_{l=1}^{m} \subseteq H(A)$ the history of $A$ until step $m$ and denote it by $H(A,m)$.
\end{defi}

Although there are several \emph{histories} for some Jurkat-Ryser $d$-matrix $A$, we fix here w.l.o.g. one instance of the Jurkat-Ryser $d$-algorithm generating $A$ and, hence, one history associated to $A$.

By Proposition \ref{difference to eto}, the equilibrium distribution we are interested in has a simple structure, $d=(d_0,\dots,d_0,d_1)$. Hence, the most important property to determine the entries in the Jurkat-Ryser $d$-matrix are what components in the Jurkat-Ryser $d$-vectors have been affected by $d_1$. We formalize this in the following definition.   

\begin{defi}[Row and column connection to $d_1$]
\label{connection def}
    If $d = (d_0,\dots,d_0,d_1) \in \mathbb P_\Omega$ and $A$ and $r^m$ are the Jurkat-Ryser $d$-matrix and $m$-th row $d$-vector, respectively, then we say $r^m_a$ is row-connected to $d_1$ if either $m=0$ and $a=|\Omega|$ or $m>0$, $r^m_a>0$ and there exists a subsequence of the history of $A$ until step $m-1$, $((i_{k_\ell}^A,j_{k_\ell}^A))_{\ell=1}^{\ell_0} \subseteq H(A,m-1)$, such that 
    \begin{alignat*}{2}
            &i^A_{k_{\ell_0}}=a &&\text{ and } i^A_{k_1}=|\Omega|,\\ &j^A_{k_{2\ell}}=j^A_{k_{2\ell-1}} &&\text{ and } i^A_{k_{2\ell+1}}=i^A_{k_{2\ell}} \text{ for } 1 \leq \ell, \text{ and} \\
            &a \neq i^A_{t} &&\text{ for } k_{\ell_0}<t \leq m-1.
            \end{alignat*}
    We call $\ell_0$ the row-connection length to $d_1$ and 
    we define the row-connection between $s^m_a$ and $d_1$, the column-connection to $d_1$ of both $r^m_a$ and $s^m_a$, and their respective lengths in the same vein. Furthermore, we say $r^m_a$ or $s^m_a$ is connected to $d_1$ if it is either row-connected or column-connected and refer to its connection length to $d_1$ in an analogous way. Lastly, we say a component $A_{i,j}$ is connected to $d_1$ if it was generated, for some $m \geq 1$, using $r^m_i$ and $s^m_j$ with at least one of them connected to $d_1$, and we naturally extend this definition to $M=AD^{-1}$.
\end{defi}

 Note that, as one can easily check by contradiction, we do not need to add to Definition \ref{connection def} constraints like
\begin{alignat*}{2}
&j_t \neq j_{k_{2\ell-1}} &&\text{ for } k_{2\ell-1}<t<k_{2\ell} ,\text{ or}\\
&i_t \neq i_{k_{2\ell}} && \text{ for } k_{2\ell}<t<k_{2\ell+1}.
\end{alignat*}

We are now ready to characterize the equivalence between the TO and ETO polytopes, which we address in the following theorem.

\begin{theo}[Equivalence TO and ETO polytopes]
\label{equi to eto poly}
If $0<d \in \mathbb P_\Omega$, then the following statements are equivalent:
\begin{enumerate}[label=(\alph*)]
     \item The thermal operations polytope is equal to the elementary thermal operations 
polytope
\begin{equation*}
    \mathcal P_{\text{TO}}(d) = \mathcal P_{\text{ETO}}(d).
\end{equation*}
    \item $d^\downarrow = (d_0,\dots,d_0,d_1)$.
 \end{enumerate}
\end{theo}

\begin{proof}
By Lemma \ref{ordered d}, it suffices to assume that $d=d^\downarrow$ throughout the proof.

    Sufficiency follows as in Proposition \ref{eto vs to}, i.e., as a direct consequence of Proposition \ref{difference to eto}.

    To prove necessity, we will construct a decomposition as a product of $d$-swaps of an arbitrary extreme point of the thermal operations polytope $M$ (whose form we know from Theorem \ref{extremes to}).
    We take $\gamma=d_1/d_0$ and $n=|\Omega|$ for simplicity, and assume throughout that $d_1<d_0$. (The case with equality is well known by Theorem \ref{birkhoff}.) 

As a first step, we determine the possible entries of $M$. In particular, we note that, for $1 \leq i,j \leq n$, $M_{i,j} \in \{1,\gamma,1-\gamma,0\}$.
We show this in the following lemma, where we start proving and analogous result for Jurkat-Ryser $d$-matrices.

\begin{lemma}
\label{lemma J-R A}
    If $d=(d_0,\dots,d_0,d_1) \in \mathbb P_\Omega$
    and $\gamma=d_1/d_0$, then the following statements hold:
    \begin{enumerate}[label=(\alph*)]
\item If $A$ is a Jurkat-Ryser $d$-matrix, then 
    \begin{equation}
    \label{JR d0}
        A_{i,j} \in \left\{d_0,d_1,d_0-d_1,0\right\}
    \end{equation}
    for $1 \leq i,j \leq |\Omega|$.
    \item If $M$ is an extreme point of the thermal operations polytope, then 
    \begin{equation}
        M_{i,j} \in \left\{1,\gamma,1-\gamma,0\right\}
    \end{equation}
    for $1 \leq i,j \leq |\Omega|$.
    \end{enumerate}
\end{lemma}

\begin{proof}
\begin{enumerate}[label=(\alph*)]
\item By definition, the only non-zero components of $A$ consist of the minimum of $r^m_i$ and $s^m_j$ for some $m \geq 1$ and $1 \leq i,j \leq |\Omega|$. Moreover, aside from the zeros which are generated by the Jurkat-Ryser algorithm (which are not used again for any comparison later on) we have that $r^m_i=s^m_j=d_0$ unless $r^m_i$ or $s^m_j$ is connected to $d_1$.
In particular, one can see that
$r^m_i=d_1$ if it is row-connected to $d_1$ and $r^m_i=d_0-d_1$ if it is column-connected to $d_1$. Similarly, $s^m_j=d_1$ if it is column-connected to $d_1$ and $s^m_j=d_0-d_1$ if it is row-connected to $d_1$. 
To conclude, note that, if $r^m_i$ ($s^m_j$) is row-connected to $d_1$, then $s^m_j$ ($r^m_i$) is either column-connected to $d_1$ or not connected to $d_1$ at all, and vice versa. Hence, in case $r^m_i$ and $s^m_j$ are both connected to $d_1$, there are only two possible scenarios:
\begin{enumerate}[label=(a.\arabic*)]
    \item $r^m_i$ is row-connected to $d_1$ and $s^m_j$ is column-connected to $d_1$. In this case,
    we have $r^m_i=s^m_j=d_1$.
    \item $r^m_i$ is column-connected to $d_1$ and $s^m_j$ is row-connected to $d_1$. In this case, 
    we have $r^m_i=s^m_j=d_0-d_1$.
    \end{enumerate}
As a result, we have that 
 \begin{equation*}
    (r^m_i,s^m_j) = \begin{cases}
    (a,b), &\text {or}\\
   (b,a),&
    \end{cases}
\end{equation*}
 where 
 \begin{equation*}
(a,b) \in \{(d_0,d_0), (d_0,d_1), (d_0,d_0-d_1),(d_1,d_1), (d_0-d_1,d_0-d_1)\}.
\end{equation*}
Hence, \eqref{JR d0} holds.

\item This follows as a direct consequence of $(a)$ since, by Theorem \ref{extremes to}, $M=A D^{-1}$, where $A$ is a Jurkat-Ryser $d$-matrix and $D=diag(d_0,\dots,d_0,d_1)$. To conclude, it suffices to notice that $A_{i,|\Omega|} \in \{0,d_1\}$ for all $1 \leq i \leq |\Omega|$, since, provided it is not zero, we have for all $m \geq 1$ that $s^m_{|\Omega|} = d_1$ and hence, as we argued in $(a)$, $r^m_{i} \in \left\{ d_0,d_1\right\}$ for $1 \leq i \leq |\Omega|$.
\end{enumerate}
This concludes the proof.
\end{proof}

Now that we know what entries $M$ may have, we show, in the following lemma, how they may be distributed along its lines.\footnote{If $M \in \mathcal M_{|\Omega|,|\Omega|}(\mathbb R)$, then a \emph{line} of $M$ is either a row or a column \cite{jurkat1967term}.}

\begin{lemma}
\label{lines of M}
If $d=(d_0,\dots,d_0,d_1) \in \mathbb P_\Omega$, $d_1 <d_0$, $\gamma=d_1/d_0$ and $M$ is an extreme point of the thermal operations polytope, then the lines of $M$ fulfill the following properties:
\begin{enumerate}[label=(\alph*)]
\item The last row has either a $1$ in the last column or a $\gamma$ in another column and the rest are zeros.
\item If a row which is not the last one has a $1$ in the last column, then it also has a $1-\gamma$ and the rest are zeros.
\item If a row which is not the last one has a $\gamma$, then it also has a $1-\gamma$ (both not in the last column) and the rest are zeros.
\item The last column has a single one and the rest are zeros.
\item If a column has a $1-\gamma$, then it also has a $\gamma$ and the rest are zeros.
\end{enumerate}
\end{lemma}
\begin{proof}
\begin{enumerate}[label=(\alph*)]
\item This follows from Lemma \ref{lemma J-R A} and the fact that $Md=d$.
\item Take $R_i$ the $i$-th row for $1 \leq i < |\Omega|$ and assume it has a one in the last column. Since $Md=d$, there must be another entry in $R_i$ that is neither $0$ nor $1$. Moreover, there can be no $\gamma$ in $R_i$. To show this, we consider two cases:
\begin{enumerate}[label=(b.\arabic*)]
\item There is one $\gamma$ in $R_i$ that was generated before the $1$ at step $m$. Then, arguing as in Lemma \ref{lemma J-R A}, $\gamma$ must come from the comparison of $r^m_i=d_0$ and $s_j^m=d_1$ to avoid filling the row with zeros. As a result, we have $r^{m+1}_i=d_0-d_1$ and, by the argument in Lemma \ref{lemma J-R A}, the next non-zero entry in $R_i$ will be $1-\gamma$. However, this scenario is impossible since the existence of the $1$ contradicts the fact that $Md=d$.
\item There is one $\gamma$ in $R_i$ that was generated after the $1$. If we assume the $1$ was generated at step $m$,
then, arguing as in Lemma \ref{lemma J-R A},
it must come from the comparison of $r^m_i=d_0$ and $s_j^m=d_1$ to avoid filling the row with zeros and, furthermore, the next non-zero entry in the row will be $1-\gamma$. However, this scenario is impossible since the existence of the $\gamma$ contradicts the fact that $Md=d$.
\end{enumerate}
In summary, by Lemma \ref{lemma J-R A} $(b)$, a row with the $1$ in the last column has a $1-\gamma$ somewhere and the rest are zeros.

\item Take $R_i$ the $i$-th row for $1 \leq i < |\Omega|$ and assume it has a $\gamma$, which cannot be in the last column by Lemma \ref{lemma J-R A} $(b)$. Since $R_i$ is not the last row and $Md=d$, then there must be another non-zero entry in $R_i$. Moreover, $R_i$ cannot have another $\gamma$. If there were another $\gamma$, consider the $\gamma$ that appeared first at step $m$. Arguing as in Lemma \ref{lemma J-R A}, it must come from the comparison of $r^m_i=d_0$ and $s_j^m=d_1$ to avoid filling the row with zeros. Thus, the next non-zero entry in the row will be a $1-\gamma$ and the existence of a second $\gamma$ contradicts the fact that $Md=d$. Lastly, $R_i$ cannot have a $1$. This is the case since, if the one was not on the last column, then this would contradict the fact that $Md=d$. Furthermore, if it were on the last column, we can follow the proof of $(b)$.
In summary, since $Md=d$, a row which is not the last one and has $\gamma$ will also have another a $1-\gamma$ and the rest zeros.

\item This follows from Lemma \ref{lemma J-R A} $(b)$ plus the fact that $M$ is a stochastic matrix.

\item Take $C_i$ the $i$-th column for $1 \leq i < |\Omega|$ a column with a $1-\gamma$. (It cannot be the last column by $(b)$.)  Since $M$ is stochastic, there must be another non-zero entry which cannot be a $1$. Moreover, there can be no other $1-\gamma$ in $C_i$. To show this, let us consider
the $1-\gamma$ that appeared first at step $m$. Arguing as in Lemma \ref{lemma J-R A}, such $1-\gamma$ must come from the comparison of $r^m_i=d_0-d_1$ and $s_j^m=d_0$ to avoid filling the column with zeros. Thus, the next non-zero entry in the row will be a $\gamma$. However, the existence of a second $1-\gamma$ in $C_i$ contradicts the fact that $M$ is stochastic. In summary, by Lemma \ref{lemma J-R A} $(b)$, a column with a $1-\gamma$ entry will have another entry with a $\gamma$ and the rest zeros.
\end{enumerate}
This concludes the proof.
\end{proof}

Now that we have identified the possible lines $M$ may have, we will establish how they are positioned in $M$ relative to each other in the following lemma.

\begin{lemma}
\label{distribution lines M}
    If $d=(d_0,\dots,d_0,d_1) \in \mathbb P_\Omega$, $d_1<d_0$, $\gamma=d_1/d_0$ and $M$ is an extreme point of the thermal operations polytope, then one of the following holds:
    \begin{enumerate}[label=(\alph*)]
\item $M$ is a permutation matrix with $M_{|\Omega|,|\Omega|}=1$.
\item There exists a sequence of pairs $Q=((i_k,j_k))_{k=0}^{t_0}$ such that 
    \begin{equation}
    \label{recursion M}
  \begin{alignedat}{3}
&j_0=|\Omega|, \text{ }&&M_{i_0,j_0}=1, &&\\
 &M_{i_{k-1},j_{k}}= 1-\gamma , \text{ }&&M_{i_k,j_k}=\gamma, &&\text{ if } 1 \leq k < t_0,\\
 &M_{i_{t_0-1},j_{t_0}} = 1-\gamma, \text{ }&& M_{i_{t_0},j_{t_0}} = \gamma, &&  
  \end{alignedat}
\end{equation}
where $i_{t_0}=|\Omega|$, $1 \leq j_{t_0} < |\Omega|$, and $0<t_0<\infty$. Moreover, the submatrix $M_0 = M \setminus M[i_{0},\dots,i_{t_0}; j_0,\dots,j_{t_0}]$ is a permutation matrix.\footnote{If $M \in \mathcal M_{|\Omega|,|\Omega|}(\mathbb R)$, then $M_0$ is a \emph{submatrix} of $M$, denoted 
\begin{equation*}
M_0=M \setminus M[a_1,\dots,a_n;b_1,\dots,b_m],    
\end{equation*}
if it is equal to $M$ after eliminating rows $\{a_1,\dots,a_n\}$ and columns $\{b_1,\dots,b_m\}$ with $1 \leq a_i,b_j \leq |\Omega|$ for $1 \leq i \leq n$ and $1 \leq j \leq m$
\cite{jurkat1967term}.}
\end{enumerate}
\end{lemma}

\begin{proof}
By Lemma \ref{lines of M} $(d)$, $M$ must have a $1$ and the rest zeros in the last column. In case the $1$ is in the last row, then, following Lemma \ref{lemma J-R A}, it was introduced at some step $m_0$ comparing $r^{m_0}_n=d_1$ with $s^{m_0}_n=d_1$ and, for all $m \neq m_0$, both $r^m_i$ and $s^m_j$ are not connected to $d_1$. Hence, $M$ is a permutation matrix with $M_{n,n}=1$, as stated in $(a)$.

Assume now that the $1$ in the last column of $M$ is in some row $R_{i_0}$ with $1 \leq i_0 < n$. We will show that $(b)$ holds. In particular, in this scenario, we define the sequence of pairs $Q=((i_k,j_k))_{k=0}^{t_0}$, where we take $i_0$ as in the previous line and $j_0=n$. Moreover, we define, for all $k \geq 1$, $j_k$ such that the column $C_{j_k}$ has a $1-\gamma$ in row $i_{k-1}$, and $i_k$ such that the row $R_{i_k}$ has a $\gamma$ in column $C_{j_k}$. We follow this procedure until we reach some $t_0>0$ for which $j_{t_0+1}$ it is no longer defined. To conclude that \eqref{recursion M} holds, we check the following properties:
\begin{enumerate}[label=(b.\arabic*)]
\item There exists some $t_0>0$ for which $Q$ is well-defined. To show this, we rely on Lemma \ref{lines of M}. In particular, there exists a single $1-\gamma$ in row $R_{i_0}$ by Lemma \ref{lines of M} $(b)$. Moreover, for all $k \geq 1$, there exists a single $\gamma$ in column $C_{j_k}$ by Lemma \ref{lines of M} $(e)$. Furthermore, while $i_{k} < n$, there exists a single $1-\gamma$ in row $R_{i_k}$ by Lemma \ref{lines of M} $(c)$.
\item $t_0< \infty$. To show this, we first note that we never repeat a pair $(i_{k_0},j_{k_0})=(i_{k_1},j_{k_1})$ for $0 \leq k_0 < k_1$. If that were the case, then we would have $(i_{k_0-v},j_{k_0-v})=(i_{k_1-v},j_{k_1-v})$ for all $0 \leq v \leq k_0$ by the uniqueness properties in Lemma \ref{lines of M}. However, $i_0 \neq i_k$ for all $k>0$, since there is no $\gamma$ in the $i_0$ row by Lemma \ref{lines of M} $(b)$. Hence, since there are no repeated pairs in $Q$ and the number of $\gamma$ in $M$ is finite, there exists some $t_0 > 0$ such that $i_{t_0}=n$. (This is the case since we are not in scenario $(a)$ and Lemma \ref{lines of M} $(a)$ holds.) Thus, since there is no $1-\gamma$ in row $R_n$ by Lemma \ref{lines of M} $(a)$, $j_{t_0+1}$ is not defined.
\end{enumerate}   
To conclude the proof, we only ought to show that the submatrix $M_0$ is a permutation matrix. To show this, we first notice that $M_{i,j}$ is connected to $d_1$ if and only if $(i,j) \in Q$. Hence, $(M_0)_{i,j} \in \{0,1\}$. Moreover, whenever we have $i=i_k$ and $j \neq j_k$ (or vice versa) for some $0 \leq k \leq t_0$, then $M_{i,j}=0$. Hence, given that $M$ is a stochastic matrix, we have that $M_0$ is a permutation matrix.
\end{proof}

(We include an extreme of the thermal polytope that fulfills Lemma \ref{distribution lines M} $(b)$ in Figure \ref{fig:matrix alg}.)

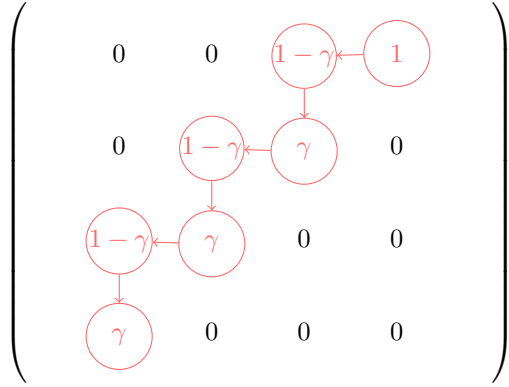
\begin{figure}[t]
    \centering
\begin{tikzpicture}[every node/.style={minimum size=3.5em}]
  \matrix (m) [matrix of math nodes, nodes in empty cells,
  left delimiter={(}, right delimiter={)}] at (0,0)
  {
    0 & 0 & |[draw,red!60,circle,inner sep=0pt,minimum size=2em]| 1-\gamma & |[draw,red!60,circle,inner sep=0pt,minimum size=2.5em]| 1\\
    0 & |[draw,red!60,circle,inner sep=0pt,minimum size=2em]| 1-\gamma & |[draw,red!60,circle,inner sep=0pt,minimum size=2.5em]| \gamma & 0 \\
    |[draw,red!60,circle,inner sep=0pt,minimum size=2.5em]| 1-\gamma & |[draw,red!60,circle,inner sep=0pt,minimum size=2.5em]| \gamma & 0 & 0 \\
    |[draw,red!60,circle,inner sep=0pt,minimum size=2.5em]| \gamma & 0 & 0 & 0 \\
  };

  \draw[red!60,->] (m-1-4) --  (m-1-3);
  \draw[red!60,->] (m-1-3) -- (m-2-3) ;
  \draw[red!60,->] (m-2-3) -- (m-2-2) ;
  \draw[red!60,->] (m-2-2) -- (m-3-2) ;
  \draw[red!60,->] (m-3-2) -- (m-3-1) ;
  \draw[red!60,->] (m-3-1) -- (m-4-1) ;
\end{tikzpicture}
    \caption{Extreme of the thermal polytope that fulfills Lemma \ref{distribution lines M} $(b)$ for $|\Omega|=4$. We highlight in red the components that correspond to the $Q$ associated to $M$ and connect them with arrows that join one step in the recursion \eqref{recursion M} with the following one. Several examples with $|\Omega|=3$ can be found in \eqref{extremes poly dim 3}.}
    \label{fig:matrix alg}
\end{figure}

Now that we know the structure of $M$, we conclude the proof defining a product of etos $N_2$ such that $M=N_2$. We consider two cases:
 \begin{enumerate}[label=(\alph*)]
\item If we are in the scenario of Lemma \ref{distribution lines M} $(a)$, then $M$ is a permutation matrix that acts as the identity on the last column. Hence, by Theorem \ref{birkhoff}, there exists a sequence of pairs $((x_k,y_k))_{k=0}^{t_2}$ with $x_k,y_k<n$ for $0\leq k \leq t_2$, such that
\begin{equation*}
    M=N_2 \coloneqq \prod_{k=0}^{t_2} P^d(x_k,y_k).
\end{equation*}

\item If we are in the scenario of Lemma \ref{distribution lines M} $(b)$, then we start by defining
\begin{equation*}
N_0 \coloneqq \prod_{k=0}^{t_0} P^d(j_k,n),
\end{equation*}
    where $(j_k)_{k=0}^{t_0}$ are the second components of the pairs in $Q$ from Lemma \ref{distribution lines M}.
    
    It is not difficult to see by induction that any product
    \begin{equation*}
        \prod_{q=1}^{q_0} P^d(u_q,n),
    \end{equation*}
     where $u_q < n$ for $1 \leq q \leq q_0$ and $u_q \neq u_{q'}$ whenever $q \neq q'$, fulfills the following relation between rows:
Row $R_{u_1}$ has a $1$ in column $C_n$ and a $1-\gamma$ in column $C_{u_1}$. Moreover, row $R_{u_k}$ has a $\gamma$ in column $C_{u_{k-1}}$ and a $1-\gamma$ in column $C_{u_{k}}$
for $1<k\leq q_0$. Finally, row $R_n$ has a $\gamma$ in column $u_{q_0}$.
(Furthermore, the entries that are not mentioned for each row are zero.)

Hence, although potentially in different rows, $N_0$ possesses one horizontal line equivalent to each $R_{i_k}$ for $1 \leq k \leq t_0$, where $(i_k)_{k=0}^{t_0}$ are the first components of the pairs in $Q$ from Lemma \ref{distribution lines M}.
 Thus, there exists a sequence $(o_k)_{k=0}^{t_0-1}$ with $o_k<n$ for $0\leq k \leq t_0-1$ such that
\begin{equation*}
N_1 \coloneqq \left(\prod_{k=0}^{t_0-1} P^d(j_k,o_k)\right) N_0
\end{equation*}
coincides with $M$ inside $[i_{0},\dots,i_{t_0}; j_0,\dots,j_{t_0}]$.

To conclude, note that, since $N_1 \setminus N_1[i_{0},\dots,i_{t_0}; j_0,\dots,j_{t_0}]$ and $M_0$ are permutation matrices, there exists a sequence of pairs $((a_k,b_k))_{k=0}^{t_1}$ with $a_k,b_k<n$ for $0\leq k \leq t_1$, such that
\begin{equation*}
    M=N_2 \coloneqq \left(\prod_{k=0}^{t_1} P^d(a_k,b_k)\right) N_1.
\end{equation*}
\end{enumerate}
This concludes the proof.
\end{proof}

\begin{rem}
As a result of Theorem \ref{equi to eto poly}, elementary thermal operations with protocols that can be conditioned using random variables are universal if and only if $d^\downarrow = (d_0,\dots,d_0,d_1)$. Moreover, given some thermal operation $M \in \mathcal M_{|\Omega|, |\Omega|}(\mathbb R)$, then Theorem \ref{equi to eto poly} together with linear programming (see the comment below Theorem \ref{polytopes relation}) gives us an algorithm to realize $M$ via strong elementary thermal operation provided $d^\downarrow =(d_0,\dots,d_0,d_1)$. Up to permutations in $d$, the algorithm can be summarized as follows:
    \begin{enumerate}[label=(\alph*)]
        \item Input $0<d =(d_0,\dots,d_0,d_1) \in \mathbb P_\Omega$ and $M \in \mathcal P_{\text{ETO}}(d)$.
        \item Calculate the decomposition of $M$ in terms of the extremes of $\mathcal P_{\text{ETO}}(d)$ using linear programming.
        \item Calculate $Q$ as in the proof of Theorem \ref{equi to eto poly} for each extreme obtained in (b) and use it, following again Theorem \ref{equi to eto poly}, to obtain a sequence of $d$-swaps that (up to permutations) yield the extreme.    
        \item Find the sequence of swaps that are lacking in (c) in order to achieve each extreme via $d$-swaps. 
        \item Output a convex combination of products of $d$-swaps that yield $M$. 
    \end{enumerate}
\end{rem}

Note that Theorem \ref{equi to eto poly} can be used to obtain Theorem \ref{polytopes relation} directly, although it requires using the tools developed in \cite{jurkat1967term,hartfiel1974study}. Also, by using a permutation matrix $Q \in \mathcal M_{|\Omega|, |\Omega|}(\mathbb R)$ such that $d^\downarrow=Qd$, we can get results analogous to Lemmas \ref{lemma J-R A}, \ref{lines of M} and \ref{distribution lines M} provided $d^\downarrow = (d_0,\dots,d_0,d_1)$.

\subsection{Universality of weak elementary thermal operations}
\label{strong uni WETO}

Although, as we have shown, the relation between the ETO and TO polytopes is analogous to that of their polytopes, the situation is quite different when WETO enters the picture. This was already noticed for the uniform case in Theorem \ref{uniform weto = to}, which we generalize in the following theorem.  

\begin{theo}[Equivalence TO and WETO polytopes]
\label{equi to weto poly}
If $0<d \in \mathbb P_\Omega$, then the following statements are equivalent:
\begin{enumerate}[label=(\alph*)]
     \item The thermal operations polytope is equal to the weak elementary thermal operations 
polytope
\begin{equation*}
    \mathcal P_{\text{TO}}(d) = \mathcal P_{\text{WETO}}(d).
\end{equation*}
    \item $|\Omega| = 2$.
 \end{enumerate}
\end{theo}

\begin{proof}
By Lemma \ref{ordered d}, it suffices to assume that $d=d^\downarrow$ throughout the proof.

Necessity follows from Lemma \ref{equiv dim 2}. We prove sufficiency by contrapositive.
In particular, if $d \neq (d_0,\dots,d_0,d_1)$ up to permutations, then the result holds already at the resource theory level by Proposition \ref{difference to eto}. We take, hence, $d = (d_0,\dots,d_0,d_1)$ with $d_1<d_0$. (The case where $d_1=d_0$ is known by Theorem \ref{uniform weto = to}.) To deal with this situation, we distinguish two cases
and provide, for each of them, a $d$-stochastic matrix $M \in \mathcal M_{|\Omega|, |\Omega|}(\mathbb R)$ that cannot be decomposed as a product of $T^d$-transforms:
\begin{enumerate}[label=(\Alph*)]
\item $|\Omega| \geq 4$. In this scenario, the result follows as a consequence of the proof of Theorem \ref{uniform weto = to} (see \cite{marcus1984products}). In particular, we can take
\begin{equation*}
    M \coloneqq M_0 \bigoplus \mathbb{I}_{\setminus (1,\dots,|\Omega|-1|)},
\end{equation*}
with $M_0 \in \mathcal M_{|\Omega|-1, |\Omega|-1}(\mathbb R)$ 
a doubly stochastic matrix 
such that $(M_0)_{i,j}>0$ if $i \neq j$ and $(M_0)_{i,j}=0$ if $i = j$. (For instance, we could define $M_0$ as in \eqref{simple RW}.) It is clear that $M$ is a $d$-stochastic matrix. To conclude, assume that $M$ can be decomposed as a sequence of $T^d$-transforms
    \begin{equation}
    \label{M Td deco}
        M= \prod_{k=1}^{N} T^d_{\lambda_k}(i_k,j_k)
    \end{equation}
    and consider $N_0$ the first index such that $\lambda_{N_0}>0$ and $j_{N_0}=|\Omega|$. In this scenario, we have that 
    \begin{equation*}
      \left(\prod_{k=1}^{N_0} T^d_{\lambda_k}(i_k,j_k)\right)_{|\Omega|,|\Omega|}<1.  
    \end{equation*}
     It is then easy to see that this implies $(\prod_{k=1}^{N_1} T^d_{\lambda_k}(i_k,j_k))_{|\Omega|,|\Omega|}<1$ for $N_0 \leq N_1 \leq N$, which contradicts the definition of $M$. Hence, for $1 \leq k \leq N$, $j_k = |\Omega|$ implies $\lambda_k=0$, thus, $T^d_{\lambda_k}(i_k,j_k)= \mathbb I$.
    As a result, $M$ can be decomposed as a product of $T^d$ transforms if and only if $M_0$ can be decomposed as a product of $T$-transforms. The latter is, however, false due to the proof of Theorem \ref{uniform weto = to}.
\item $|\Omega|=3$. In this scenario, the argument above does not hold since Theorem \ref{uniform weto = to} requires $M_0$ to be at least a $3 \times 3$ matrix. In the spirit of the proof of Theorem \ref{uniform weto = to}, however, we take 
    \begin{equation}
    \label{dim 3 counterex}
        M \coloneqq \begin{pmatrix}
 1-\phi & \phi & 0 \\
\phi & 0 & \frac{1-\phi}{\gamma} \\
0 & 1-\phi & 1- \frac{1-\phi}{\gamma}
\end{pmatrix}
    \end{equation}
    with $1-\gamma < \phi < 1$. Note that $M$ is clearly a $d$-stochastic matrix.
    
    Since there exists some $0 \leq \lambda' \leq 1$ such that
\begin{equation}
\label{prod inte}
    \prod_{k=1}^{k_0} T^d_{\lambda_k}(i,j)= T^d_{\lambda'}(i,j),
\end{equation}
where $1 \leq i,j \leq 3$ and $0 \leq \lambda_k \leq 1$ for $1 \leq k \leq k_0$,
and, moreover, 
\begin{equation}
\label{commutation}
    T^d_\lambda (1,3)=P^d(1,2) T^d_\lambda (2,3) P^d(1,2)
\end{equation}
for any $0 \leq \lambda \leq 1$,
then $M$ can be decomposed as product of $T^d$-transforms if and only if there exist some $\ell,m \in \{0,1\}$ and $N \geq 0$ such that
\begin{equation}
\label{product}
    M= \left( T^d_{\lambda_A} (2,3) \right)^\ell \left(\prod_{k=1}^N T^d_{\lambda_k} (1,2) T^d_{\beta_k} (2,3) \right) \left( T^d_{\lambda_B} (1,2) \right)^m,
\end{equation}
where $0 \leq \lambda_A,\lambda_B,\lambda_k,\beta_k\leq 1$ for $1 \leq k \leq N$.
To conclude the argument, it is straightforward to check that this is not the case. (For instance, one can argue by direct calculation using the zeros in $M$.)

As a last remark, note that, in case $\phi=1$, then $M$ equals $P^d(1,2)$ and, hence, it belongs to the WETO polytope. Moreover, if $\phi=1-\gamma$, then $M$ also belongs to the WETO polytope. In particular, $M=P^d(1,2) P^d(2,3) P^d(1,3)$.
\end{enumerate}
This concludes the proof.
\end{proof}

\begin{rem}
As a result of Theorem \ref{equi to eto poly}, elementary thermal operations with deterministic protocols are universal if and only if $|\Omega|=2$. Hence, by Theorem \ref{equi to eto poly}, conditioning the protocols via random variables augments the instances where elementary thermal operations are universal.
\end{rem}

Note that, for each $1-\gamma< \phi <1$, $M$ in \eqref{dim 3 counterex} exemplifies that the WETO polytope is not convex in general since, aside from its extremes, the segment joining $P^d(1,2)$ and 
$P^d(1,2) P^d(2,3) P^d(1,3)$ does not belong to it by Theorem \ref{equi to weto poly}.

\section{Advantage of strong elementary thermal operations over their weak counterpart}
\label{advantage}

In this section, we make substantial progress regarding (Q3) and weak universality in Section \ref{weak uni adv} and fully answer it for the non-weak case in Section \ref{uni adv}.

\subsection{Weak universality advantage}
\label{weak uni adv}


As a result of Theorems \ref{polytopes relation} and \ref{wETO equiv}, we have the following relation between strong and weak elementary resource theories.

\begin{coro}
\label{equi rt eto weto}
    If $0<d \in \mathbb P_\Omega$ and $d^\downarrow = (d_0,\dots,d_0,d_1)$, then the resource theories of strong and weak elementary thermal operations are equal
    \begin{equation*}
    \text{RT}_{\text{ETO}}(d) = \text{RT}_{\text{WETO}}(d).
\end{equation*}
\end{coro} 

To deal with the relation between these two theories more in general, we define the set of quasi-uniform distributions, which includes the hypothesis in Corollary \ref{equi rt eto weto}.

\begin{defi}[Quasi-uniform distribution]
\label{quasi-uniform def}
    $0<d \in \mathbb P_\Omega$ is quasi-uniform provided it has, at most, two different entries $d_i \in \{x,y\}$ for all $i \in \Omega$.
\end{defi}

Note that the distributions $d \in \mathbb P_\Omega$ fulfilling Definition \ref{quasi-uniform def} are those for which either $d^\downarrow = (d_0,\dots,d_0,d_1)$ with $d_0 \geq d_1$ or 
$d^\downarrow = (d_0,d_1,\dots,d_1)$ with $d_0 > d_1$.

The relevance of Definition \ref{quasi-uniform def} lies in the fact that, as we show in the following proposition, ETO and WETO resource theories do not coincide whenever $d$ is not quasi-uniform.

\begin{prop}[Difference ETO and WETO resource theories]
\label{difference weto eto}
If $0<d \in \mathbb P_\Omega$, then 
the elementary thermal operations resource theory is equal to the weak elementary thermal operations 
resource theory
\begin{equation*}
    \text{RT}_{\text{ETO}}(d) = \text{RT}_{\text{WETO}}(d)
\end{equation*}
only if 
$d$ is quasi-uniform.
\end{prop}

\begin{proof}
    By Lemma \ref{ordered d}, it suffices to assume that $d=d^\downarrow$ throughout the proof.

    We argue by contrapositive, that is, we assume we have some $d \in \mathbb P_\Omega$ with three different entries and will show that, in this scenario, there exists some pair $p,q \in \mathbb P_\Omega$ such that $q$ is achievable from $q$ via (strong) elementary thermal operations but not by its weaker counterpart.

    We first note that we can restrict ourselves to the case where $|\Omega|=3$, as we show in the following lemma.

    \begin{lemma}
    \label{dim 3 to larger}
Take $0<d \in \mathbb P_\Omega$ and $0<d' \in \mathbb P_{\Omega'}$ such that $\Omega \subseteq \Omega'$ and $d'_i=d_i$ for all $i \in \Omega$. If $p,q \in \mathbb P_\Omega$, then the following statements hold:
\begin{enumerate}[label=(\alph*)]
\item If $q \in C^{ETO}_d(p)$, then $(q,0) \in C^{ETO}_{d'}((p,0))$, where $(r,0) \in \mathbb P_{\Omega'}$ and
    \begin{equation*}
         (r,0)_i \coloneqq \begin{cases}
    r_i, &\text {if } i \in \Omega\\
   0,&\text {if } i \in \Omega' \setminus \Omega.
    \end{cases}
    \end{equation*}
\item If $\text{supp} (p)=\text{supp} (q)$ and $q \not \in C^{WETO}_d(p)$, then $(q,0) \not \in C^{WETO}_{d'}((p,0))$.
\end{enumerate}
\end{lemma}

    \begin{proof}
        \begin{enumerate}[label=(\alph*)]
\item Straightforward.
\item For simplicity, we assume throughout the proof that $\text{supp} (p)=\text{supp} (q)=\Omega$. The same method can be followed whenever $\text{supp} (p)=\text{supp} (q) < \Omega$. Moreover,  by Lemma \ref{ordered d}, we assume $d'=(d')^\downarrow$ and $d=d^\downarrow$ .

We begin assuming $(q,0) \in C^{WETO}_{d'}((p,0))$ and argue by contradiction. By assumption, we have that
    \begin{equation}
    \label{path}
    (q,0) = \left( \prod_{k=1}^\ell T_{\lambda_k}^d(i_k,j_k)\right) (p,0) 
\end{equation}
for some $\ell \geq 0$ and $1 \leq i_k<j_k \leq |\Omega'|$ for $1 \leq k \leq \ell$. Note that, in the remainder of the proof, we use the notation
\begin{equation}
    t^m \coloneqq
    \begin{cases}
    (p,0), &\text {if } m=0\\
     \left( \prod_{k=1}^m T_{\lambda_k}^d(i_k,j_k)\right) (p,0),&\text{if } 1 \leq m \leq \ell.
    \end{cases}
\end{equation}

Since $q \not \in C^{WETO}_d(p)$, there must be some $1 \leq k_0 \leq \ell$ such that $i_{k_0} \not \in \Omega$ (or $j_{k_0} \not \in \Omega$). We fix $k_0$ the first index with this property such that $d'_{i_{k_0}} \neq d_m$ ($d'_{j_{k_0}} \neq d_m$) for all $m \in \Omega$. (We deal with the remaining cases later on.) Note we can assume w.l.o.g. that
\begin{equation}
\label{non-zero mass}
  t^{k_0}_s >0 \text{ for } s = i_{k_0} \text{ }(s = j_{k_0}).
\end{equation}
In particular, since $k_0$ is minimal by construction, we can assume that $d'_{j_{k_0}} = d_m$ ($d'_{i_{k_0}} = d_m$) for some $m \in \Omega$.

The existence of such a $k_0$ leads to contradiction. To show this, we consider the following two cases:
\begin{enumerate}[label=(\Alph*)]
\item $j_{k_0} \not \in \Omega$. By assumption, we have that $d'_{i_{k_0}} > d'_{j_{k_0}}$. In particular, by \eqref{non-zero mass}, we obtain that
\begin{equation*}
|\text{supp} (t^{k_0})|> |\text{supp} ((p,0))| = |\text{supp} ((q,0))|.    
\end{equation*}
This, together with \eqref{path}, contradicts Lemma \ref{monot eto}.
\item $i_{k_0} \not \in \Omega$. In this case, we have that $T_{\lambda_k}^d(i_k,j_k)$ is a $d$-swap to avoid contradicting Lemma \ref{monot eto}. Moreover, to avoid contradicting \eqref{path}, there must be some $k_0 < m_0 \leq \ell$ such that
\begin{equation}
\label{0 outside omega}
    t^m_{i_{k_0}}=0 \text{ for all } m \geq m_0.
\end{equation}
Hence, there must be some $i_{k_0}>w_0 \in \Omega$ and two sequences $(n_p)_{p=1}^N$, with $k_0 <n_p<n_{p+1}\leq \ell$ for all $p$, and $(v_p)_{p=1}^N \subseteq \Omega'$ such that $t^{n_{p+1}}_{v_{p+1}}>0$ and $t^{n_{p+1}}_{v_p}=0$, where $v_0=k_0$, $d'_{v_{p+1}} \geq d'_{v_p}$ by (A) and $v_N = w_0$. However, to avoid contradicting \eqref{0 outside omega}, we must have $t^{n_{N-1}}_{v_N}=0$. Despite this, $t^0_{v_N}>0$ since $\text{supp} (p)=\Omega$. Thus, we can argue analogously that there must be some $w_0>w_1 \in \Omega$ with similar properties to those of $w_0$ and $d_{w_1} > d_{w_0}$. (Note that we could have $d_{w_1} = d_{w_0}$. However, we could argue in a similar fashion that this would imply the existence of some $w_1>w_1' \in \Omega$ with equivalent properties and such that $d_{w'_1} > d_{w_1}=d_{w_0}$.) Following the argument recursively, we obtain an unbounded strictly decreasing sequence $(w_n)_{n \geq 0} \subseteq \Omega$. This contradicts the finiteness of $\Omega$.
\end{enumerate}
By (A) and (B), if \eqref{path} holds, then, for all $1 \leq k \leq \ell$,
\begin{equation*}
    d'_{i_k},d'_{j_k} \in \{d_i| i \in \Omega\}.
\end{equation*}
We conclude noting two properties. First, whenever $d_{i_k} = d_{j_k}$ with $t^{k-1}_{i_k}>0$ and $t^{k-1}_{j_k}=0$
or vice versa, then, for all $1 \leq k \leq \ell$, $T_{\lambda_k}^d(i_k,j_k)$ in \eqref{path} must be a swap to avoid contradicting Lemma \ref{monot eto}. Second, whenever $d_{i_k} > d_{j_k}$, we must have $t^{k-1}_{i_k},t^{k-1}_{j_k}>0$  by the same reason. (If $t^{k-1}_{j_k}=0$ we can argue as in (A) and, if $t^{k-1}_{i_k}=0$, as in (B).) Using these properties, it is easy to see that \eqref{path} implies $q \in C^{WETO}_d(p)$. (In particular, for all $0 \leq m \leq \ell$, there exists some $r^m \in \mathbb P_\Omega$ such that $t^m=(r^m,0)$ up to permutations and $r^m \in \mathcal C^{WETO}_d(p)$.) This yields the desired contradiction and concludes the proof of (b).
\end{enumerate}
This concludes the proof.
    \end{proof}

(Note that one may not need the assumption on the support in Lemma \ref{dim 3 to larger}. However, this is not important for our purposes here.)
    
As a result of Lemma \ref{dim 3 to larger}, it suffices to fix $|\Omega|=3$ and find a pair $p,q \in \mathbb P_\Omega$ with $\text{supp} (p)=\text{supp} (q)=\Omega$ such that $q \in C^{ETO}_d(p) \setminus C^{WETO}_d(p)$. We conclude the proof showing such a pair exists.
    
In the following, we fix $|\Omega|=3$ and take  $d = (d_0,d_1,d_2)$ with $d_0>d_1>d_2$ (we can do so by assumption). Moreover,
we fix $\alpha= d_1/d_0$, $\beta=d_2/d_0$ and $\gamma=d_2/d_1$, and take some $p = (a,b,c) \in \mathbb P_\Omega$ such that
   \begin{equation}
   \label{eto vs weto}
   \begin{split}
   &\gamma b < c < \beta a, \text{ and } \\
   &\beta ((1-\alpha )a+ (1-\gamma) b +c) < c.
\end{split}
 \end{equation}
 (Note that both conditions in \eqref{eto vs weto} can be simultaneously satisfied. This is easy to see by assuming $b=0$ and noticing that, in this scenario, the conditions are fulfilled provided $1/(1+\beta)<a<1/(1+\tau \beta)$, where $\tau= (1-\alpha)/(1-\beta)$. It is then easy to get some $p$ fulfilling \eqref{eto vs weto} with $0<b$, that is, with $\text{supp}(p)=\Omega$.)

To conclude, we argue by contradiction that there exists some $q \in C^{ETO}_d(p)$ that cannot be achieved via weak elementary thermal operations. In order to do so, we first note that $C^{ETO}_d(p)$ is given by Figure \ref{covex hull last}. (Note that this can be achieved via direct calculation using \cite[Theorem 6]{lostaglio2018elementary}, or \cite[Theorem 5]{son2022catalysis} for simplicity,  and eliminating the non-extreme points later on.)

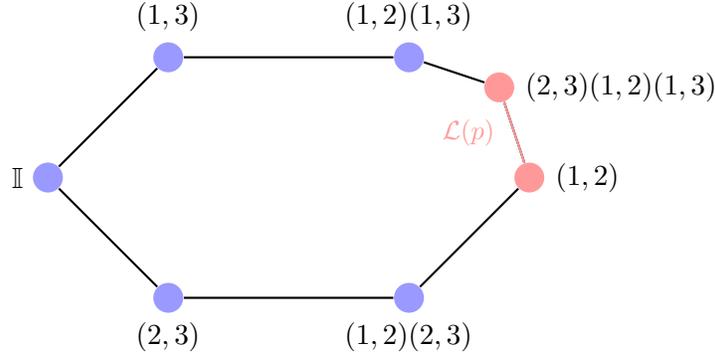
\begin{figure}[ht]
\centering
\begin{tikzpicture}[scale=0.8]
\node[small dot new,label={[yshift=0.01cm, thick, font=\fontsize{11}{11}\selectfont, thick]270:$(2,3)$}] at (-2,0) (1) {};
\node[small dot new,label={[yshift=0.01cm, thick, font=\fontsize{11}{11}\selectfont, thick]270:$(1,2) (2,3)$}] at (2,0) (2) {};
\node[small dot new,label={[xshift=0.01cm, thick, font=\fontsize{11}{11}\selectfont, thick]180:$\mathbb I$}] at (-4,2) (3) {};
\node[small dot new 3,label={[xshift=0.01cm, thick, font=\fontsize{11}{11}\selectfont, thick]0:$(1,2)$}] at (4,2) (4) {};
\node[small dot new,label={[yshift=0.01cm, thick, font=\fontsize{11}{11}\selectfont, thick]90:$(1,3)$}] at (-2,4) (5) {};
\node[small dot new,label={[yshift=0.01cm, thick, font=\fontsize{11}{11}\selectfont, thick]90:$(1,2) (1,3)$}] at (2,4) (6) {};
\node[small dot new 3,label={[xshift=0.01cm, thick, font=\fontsize{11}{11}\selectfont, thick]0:$(2,3)(1,2)(1,3)$}] at (3.5,3.5) (7) {};
\path[draw,thick,-]
    (1) edge node {} (2)
    (1) edge node {} (3)
    (2) edge node {} (4)
    (3) edge node {} (5)
    (4) edge node {} (7)
    (5) edge node {} (6)
    (7) edge node {} (6)
    ;

    \path[draw,thick,-,color=red!40]
    (4) edge node [midway, label=left:$\mathcal L(p)$] {} (7);
\end{tikzpicture}
\caption{Rough representation of the set $C^{ETO}_d(p)$ for $d=(d_0,d_1,d_2)$ with $d_1>d_1>d_2$ and $p=(a,b,c)$ fulfilling \eqref{eto vs weto}. We include in red $\mathcal L(p)$ as defined in \eqref{def L(p)}. Note that we use the notation in Figure \ref{covex hull 1}.}
\label{covex hull last}
\end{figure}

We will show there exists some 
\begin{equation}
\label{def L(p)}
    q \in \mathcal L(p) \coloneqq \left[P^d(2,3)P^d(1,2)P^d(1,3)p,P^d(1,2)p\right] \subseteq C^{ETO}_d(p)
\end{equation}
such that $q$ cannot be achieved from $p$ via weak elementary thermal operations, where $[A,B]$ stands for the segment joining points $A$ and $B$. (Note that, for any $q \in \mathcal L (p)$, we have $\text{supp}(q)=\Omega=\text{supp}(p)$. Moreover, $\mathcal L (p)\setminus \text{ext}(\mathcal L (p)) \neq \emptyset$, where $\text{ext}(\cdot)$ denotes the set of extreme points.)

We begin fixing some $q \in \mathcal L(p)$ such that
\begin{equation}
\label{extremes weto eto}
    q \neq \left( \prod_{k=1}^m P^d(i_k,j_k)\right) p 
\end{equation}
for all $m \geq 1$ and $1 \leq i_k<j_k \leq 3$ for $1 \leq k \leq m$. (This can be done since there is a continuum number of points points in $\mathcal L(p)$.) Hence, if we assume that $q \in C^{WETO}_d(p)$, then there exists at least one $T^d$-transform that is not a $d$-swap in the sequence of $T^d$-transforms that yield $q$ when applied to $p$. That is, we have 
\begin{equation}
\label{contra weto eto}
 \begin{split}
   &q = \left(\prod_{k=m+2}^\ell P^d(i_k,j_k)\right) \left((1-\lambda) \mathbb I + \lambda P^d(i_{m+1},j_{m+1})\right) p_0, \text{ with} \\
   &p_0 \coloneqq \left(\prod_{k=1}^m T_{\lambda_k}^d(i_k,j_k)\right)  p,
\end{split}
\end{equation}
for some $\ell \geq m+1$, $m \geq 0$, $0<\lambda<1$ and $1 \leq i_k<j_k \leq 3$ for $1 \leq k \leq \ell$.

We will show \eqref{contra weto eto} leads to a contradiction, implying $q \not \in C^{WETO}_d(p)$. In order to do so, we prove a stronger result in the following lemma.
\begin{lemma}
\label{lemma weto vs eto}
If we define, for all $r,s \in \mathbb P_\Omega$ and $M \in \mathcal M_{|\Omega|, |\Omega|}(\mathbb R)$, $M[r,s] \coloneqq [Mr,Ms]$, and, for any segment segment $L \subseteq \mathbb R^3$, $E(L) \coloneqq L' \bigcap \mathcal C^{ETO}_d(p)$ with $L' \subseteq \mathbb R^3$ the line such that $L \subseteq L'$, then
    \begin{equation}
    \label{L not belongs}
      \mathcal  L(p) \not \in \bigcup_{\ell \geq 0} B_\ell,
    \end{equation}
    where
    \begin{equation*}
    \begin{split}
        B_\ell \coloneqq & \Biggl\{ E\left(\left(\prod_{k=1}^\ell P^d(i_k,j_k)\right) \left[p_0, P^d(i_{0},j_{0}) p_0\right]\right) : 1 \leq i_k <j_k \leq 3 \\
        &\text{for } 1 \leq k \leq \ell \text{ and } p_0 \text{ is given by \eqref{contra weto eto}} \Biggl\}.
        \end{split}
    \end{equation*}
\end{lemma}

\begin{proof}
    We will prove that $\mathcal L (p) \not \in B_\ell$ for all $\ell \geq 0$ by induction.

If $\ell=0$, then each element in $B_\ell$ has one constant component by definition. However, it is not difficult to see that the extremes of $\mathcal L(p)$ differ in all their components. Hence, $\mathcal L(p) \not \in B_0$.

Assume now that $\mathcal L(p) \not \in B_\ell$ for some $\ell \geq 0$. To show this implies 
$L(p) \not \in B_{\ell+1}$, we note that 
\begin{equation}
\label{recursion B}
B_{\ell+1}=\left\{E\left(P^d(i_{\ell+1},j_{\ell+1})L\right) : L \in B_\ell \text{ and } 1 \leq i_{\ell +1} <j_{\ell +1} \leq 3\right\},  
\end{equation}
and prove separately the following three claims:
\begin{enumerate}[label=(\Alph*)]
\item $\mathcal L(p) \neq E(P^d(1,2) L)$ for all $L \in B_{\ell}$. Since $P^d (1,2)$ leaves the third component of every probability distribution unchanged, it is clear that we only need to consider the cases where $L$ belongs to the area with horizontal blue lines in Figure \ref{areas covex hull last}. Given that $\mathcal L(p) \neq L$ for all $L \in B_{\ell}$ by assumption, it is then easy to show that $\mathcal L(p) \neq E(P^d(1,2) L)$ using that $P^d(1,2)P^d(1,3)p,p \not \in \mathcal L(p)$ and that $P^d(1,2)P^d(1,2)= \alpha \mathbb I + (1-\alpha) P^d(1,2)$. 

\item $\mathcal L(p) \neq E(P^d(1,3) L)$ for all $L \in B_{\ell}$. Since $P^d (1,3)$ leaves the second component of every probability distribution unchanged, it is clear that we only need to consider the cases where $L$ belongs to the shaded blue area in Figure \ref{areas covex hull last}. Given that $\mathcal L(p) \neq L$ for all $L \in B_{\ell}$ by assumption, it is then easy to show that $\mathcal L(p) \neq E(P^d(1,2) L)$ using that
\begin{equation*}
\begin{split}
     &P^d(1,3)p, P^d(1,3) P^d(2,3)p, P^d(1,3) P^d(1,2) P^d(2,3)p, \\
     & P^d(1,3) P^d(1,2)p \not \in \mathcal L(p)
\end{split}
\end{equation*}
and that $P^d(1,3)P^d(1,3)= \beta \mathbb I + (1-\beta) P^d(1,2)$.

\item $\mathcal L(p) \neq E(P^d(2,3) L)$ for all $L \in B_{\ell}$. Since $P^d (2,3)$ leaves the first component of every probability distribution unchanged, it is clear that we only need to consider the cases where $L$ belongs to the area with vertical red lines in Figure \ref{areas covex hull last}. Given that $\mathcal L(p) \neq L$ for all $L \in B_{\ell}$ by assumption, it is then easy to show that $\mathcal L(p) \neq E(P^d(2,3) L)$ using that $P^d(2,3)P^d(1,2)p, P^d(1,2)P^d(1,3)p \not \in \mathcal L(p)$ and that $P^d(2,3)P^d(2,3)= \gamma \mathbb I + (1-\gamma) P^d(2,3)$. 
\end{enumerate}

\begin{figure}[ht]
\centering
\begin{tikzpicture}[scale=0.8]
\fill [color=blue!60] (-3.5,2.5)--(-4,2)--(-2,0)--(2,0)--(4,2)--(3.5,3.5);
\fill [white, pattern=horizontal lines, pattern color=blue] (-2.5,3.5)--(-4,2)--(4,2)--(3.5,3.5);
\fill [white, pattern=vertical lines, pattern color=red] (2,4)--(3.5,3.5)--(4,2)--(1,4);
\node[small dot new] at (-2,0) (1) {};
\node[small dot new] at (2,0) (2) {};
\node[small dot new] at (-4,2) (3) {};
\node[small dot new 3] at (4,2) (4) {};
\node[small dot new] at (-2,4) (5) {};
\node[small dot new] at (2,4) (6) {};
\node[small dot new 3] at (3.5,3.5) (7) {};

\node[small dot new 2,label={[xshift=.25cm, yshift=0.01cm, thick, font=\fontsize{11}{11}\selectfont, thick]90:$(2,3)(1,2)$}] at (1,4) (8) {};
\node[small dot new 2,label={[xshift=-.25cm, yshift=0.01cm, thick, font=\fontsize{11}{11}\selectfont, thick]90:$(1,3)(2,3)$}] at (-1,4) (9) {};
\node[small dot new 2,label={[xshift=.025cm, thick, font=\fontsize{11}{11}\selectfont, thick]0:$(1,3)(1,2)$}] at (2.7,0.7) (10) {};
\node[small dot new 2,label={[xshift=.025cm, thick, font=\fontsize{11}{11}\selectfont, thick]0:$(1,3)(1,2)(2,3)$}] at (3.3,1.3) (11) {};

\path[draw,thick,-]
    (1) edge node {} (2)
    (1) edge node {} (3)
    (2) edge node {} (4)
    (3) edge node {} (5)
    (4) edge node {} (7)
    (5) edge node {} (6)
    (7) edge node {} (6)
    ;
\end{tikzpicture}
\caption{Rough representation of the set $C^{ETO}_d(p)$ for $d=(d_0,d_1,d_2)$ with $d_1>d_1>d_2$ and $p=(a,b,c)$ fulfilling \eqref{eto vs weto}. We include labels for the black points which, although not plotted in Figure \ref{covex hull last} since they are not extremal, are useful when proving Lemma \ref{lemma weto vs eto}. Moreover, we include the areas which are mentioned in (A)-(C) in that lemma. Note that we use the notation in Figure \ref{covex hull 1}.}
\label{areas covex hull last}
\end{figure}
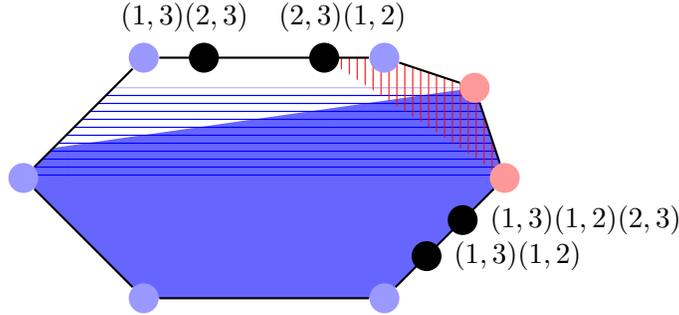

By the recursive relation in \eqref{recursion B}, we can conclude from (A)-(C) that $L(p) \not \in B_{\ell+1}$. By induction, \eqref{L not belongs} holds. This concludes the proof.
\end{proof}

By Lemma \ref{lemma weto vs eto}, we have that,
if $q \in E(L)$ for some $E(L) \in B_\ell$ and $\ell \geq 0$, then $q \in \text{ext} (E(L))$, where $\text{ext}(\cdot)$ denotes the set of extreme points. This contradicts the fact that $0< \lambda <1$ in \eqref{contra weto eto}. In particular, \eqref{contra weto eto} implies that
$q \in L_0 \setminus \text{ext} (L_0)$, where $L_0 \in B_\ell$ for some $\ell \geq 0$.
This yields the desired contradiction and concludes the proof.
\end{proof}

\begin{rem}
    As a result of Proposition \ref{difference weto eto}, conditioning our experimental protocols on the realization of some random variable results in the elementary thermal operations being closer to weak universality provided the system has at least three different energy levels.
\end{rem}

Note that Proposition \ref{difference weto eto} agrees with Lemma \ref{equiv dim 2} since any distribution in quasi-uniform for $|\Omega|=2$.

Since it seems we cannot extend the discrepancy between WETO and ETO resource theories in  Proposition \ref{difference weto eto} beyond quasi-uniform distributions, we consider the low-dimensional case $|\Omega|=3$ in the following proposition. As it turns out, these distributions characterize the equivalence between WETO and ETO resource theories whenever $|\Omega|=3$.

\begin{prop}
     [Equivalence ETO and WETO resource theories for $|\Omega| = 3$]
\label{weto eto low dim}
    If $0<d \in \mathbb P_\Omega$ and $|\Omega| = 3$, then the following statements are equivalent:
 \begin{enumerate}[label=(\alph*)]
     \item The strong elementary thermal operations resource theory is equal to the weak elementary thermal operations 
resource theory
\begin{equation*}
    \text{RT}_{\text{ETO}}(d) = \text{RT}_{\text{WETO}}(d).
\end{equation*}
    \item $d$ is quasi-uniform.
 \end{enumerate}
\end{prop}

\begin{proof}
By Lemma \ref{ordered d}, it suffices to assume that $d=d^\downarrow$ throughout the proof.

We only prove necessity since sufficiency follows from Proposition \ref{difference weto eto}. Since $d$ is quasi-uniform, then we either have $d = (d_0,d_0,d_1)$ with $d_0 \geq d_1$ or $d = (d_0,d_1,d_1)$ with $d_0>d_1$. Since the statement is true for the first instance (by putting together Propositions \ref{eto vs to} and \ref{Td transfrom low dim}), we fix $d = (d_0,d_1,d_1)$ with $d_0>d_1$ and $\gamma = d_1/d_0$ in the following. We will show that, for any $p=(a,b,c) \in \mathbb P_\Omega$ (which we can assume fulfills $b \geq c$ w.l.o.g.), we have $\mathcal C^{ETO}_d(p) \subseteq \mathcal C^{WETO}_d(p)$. In order to do so, we distinguish the following cases (the instances where some equality holds follow easily from these):

 \begin{enumerate}[label=(\Alph*)]
\item $\gamma a > b > c$. In this case, $\mathcal C^{ETO}_d (p)$ is (roughly) given by Figure \ref{covex hull weto 1},
with $q \in C^{ETO}_d (p)$ being achievable by a sequence $T^d(2,3)T^d(1,2)$ if it lies to the left of the dashed line and by $T^d(2,3)T^d(1,3)P^d(1,2)$ if it lies to the right.
\end{enumerate}

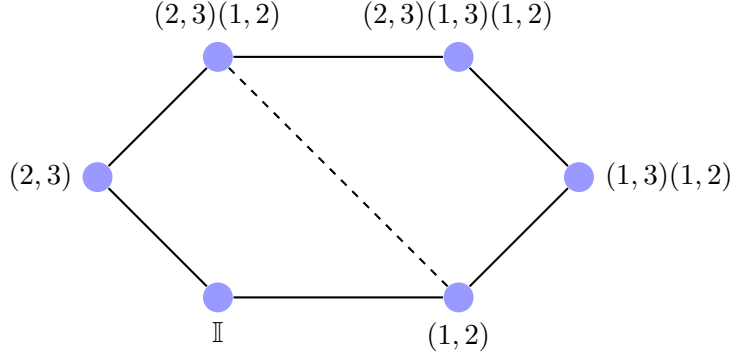
\begin{figure}[ht]
\centering
\begin{tikzpicture}[scale=0.8]
\node[small dot new,label={[yshift=0.01cm, thick, font=\fontsize{11}{11}\selectfont, thick]270:$\mathbb I$}] at (-2,0) (1) {};
\node[small dot new,label={[yshift=0.01cm, thick, font=\fontsize{11}{11}\selectfont, thick]270:$(1,2)$}] at (2,0) (2) {};
\node[small dot new,label={[xshift=0.01cm, thick, font=\fontsize{11}{11}\selectfont, thick]180:$(2,3)$}] at (-4,2) (3) {};
\node[small dot new,label={[xshift=0.01cm, thick, font=\fontsize{11}{11}\selectfont, thick]0:$(1,3)(1,2)$}] at (4,2) (4) {};
\node[small dot new,label={[yshift=0.01cm, thick, font=\fontsize{11}{11}\selectfont, thick]90:$(2,3)(1,2)$}] at (-2,4) (5) {};
\node[small dot new,label={[yshift=0.01cm, thick, font=\fontsize{11}{11}\selectfont, thick]90:$(2,3)(1,3) (1,2)$}] at (2,4) (6) {};
\path[draw,thick,-]
    (1) edge node {} (2)
    (1) edge node {} (3)
    (2) edge node {} (4)
    (3) edge node {} (5)
    (4) edge node {} (6)
    (5) edge node {} (6)
    ;
     \path[draw,dashed,thick,-]
    (2) edge node {} (5)
    ;
\end{tikzpicture}
\caption{Rough representation of $\mathcal C^{ETO}_d(p)$ for $d=(d_0,d_1,d_1)$ with $0<d_1 < d_0$ and $p=(a,b,c)$ with $\gamma a \geq b \geq c$. Note that we use the notation in Figure \ref{covex hull 1}.}
\label{covex hull weto 1}
\end{figure}

\begin{enumerate}[label=(\Alph*)]
\setcounter{enumi}{1}
\item $b>\gamma a > c$. In this case, $\mathcal C^{ETO}_d (p)$ is (roughly) given by Figure \ref{covex hull weto 2}, with $q \in C^{ETO}_d (p)$ being achievable by a sequence $T^d(2,3)T^d(1,3)$ if it lies to the right of the dashed line and by $T^d(2,3)T^d(1,2)$ if it lies to the left.
\end{enumerate}

\begin{figure}[ht]
\centering
\begin{tikzpicture}[scale=0.8]
\node[small dot new,label={[yshift=0.01cm, thick, font=\fontsize{11}{11}\selectfont, thick]270:$(1,2)$}] at (-2,0) (1) {};
\node[small dot new,label={[yshift=0.01cm, thick, font=\fontsize{11}{11}\selectfont, thick]270:$\mathbb I$}] at (2,0) (2) {};
\node[small dot new,label={[xshift=0.01cm, thick, font=\fontsize{11}{11}\selectfont, thick]180:$(2,3)(1,2)$}] at (-4,2) (3) {};
\node[small dot new,label={[xshift=0.01cm, thick, font=\fontsize{11}{11}\selectfont, thick]0:$(1,3)$}] at (4,2) (4) {};
\node[small dot new,label={[yshift=0.01cm, thick, font=\fontsize{11}{11}\selectfont, thick]90:$(2,3)$}] at (-2,4) (5) {};
\node[small dot new,label={[yshift=0.01cm, thick, font=\fontsize{11}{11}\selectfont, thick]90:$(2,3)(1,3)$}] at (2,4) (6) {};
\path[draw,thick,-]
    (1) edge node {} (2)
    (1) edge node {} (3)
    (2) edge node {} (4)
    (3) edge node {} (5)
    (4) edge node {} (6)
    (5) edge node {} (6)
    ;
     \path[draw,dashed,thick,-]
    (2) edge node {} (5)
    ;
\end{tikzpicture}
\caption{Rough representation of $\mathcal C^{ETO}_d(p)$ for $d=(d_0,d_1,d_1)$ with $0<d_1 < d_0$ and $p=(a,b,c)$ with $b \geq \gamma a \geq c$. Note that we use the notation in Figure \ref{covex hull 1}.}
\label{covex hull weto 2}
\end{figure}
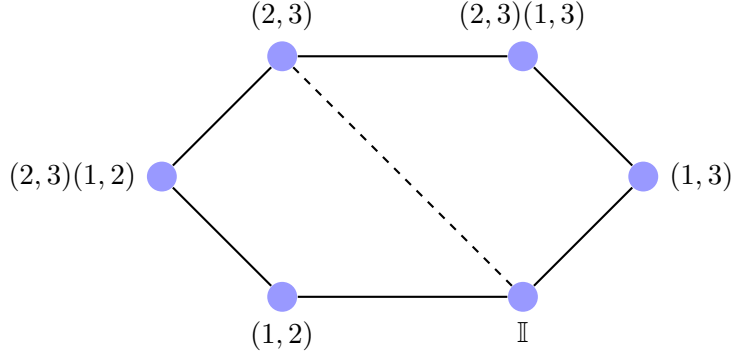

\begin{enumerate}[label=(\Alph*)]
\setcounter{enumi}{2}
\item $b > c >  \gamma a$. In this case, $\mathcal C^{ETO}_d (p)$ is (roughly) given by Figure \ref{covex hull weto 3}, with $q \in C^{ETO}_d (p)$ being achievable by a sequence $T^d(2,3)T^d(1,3)$ if it lies to the right of the dashed line and by $T^d(2,3)T^d(1,2)P^d(1,3)$ if it lies to the left.
\end{enumerate}

\begin{figure}[ht]
\centering
\begin{tikzpicture}[scale=0.8]
\node[small dot new,label={[yshift=0.01cm, thick, font=\fontsize{11}{11}\selectfont, thick]270:$(1,2)(1,3)$}] at (-2,0) (1) {};
\node[small dot new,label={[yshift=0.01cm, thick, font=\fontsize{11}{11}\selectfont, thick]270:$(1,3)$}] at (2,0) (2) {};
\node[small dot new,label={[xshift=0.01cm, thick, font=\fontsize{11}{11}\selectfont, thick]180:$(2,3)(1,2)(1,3)$}] at (-4,2) (3) {};
\node[small dot new,label={[xshift=0.01cm, thick, font=\fontsize{11}{11}\selectfont, thick]0:$\mathbb I$}] at (4,2) (4) {};
\node[small dot new,label={[yshift=0.01cm, thick, font=\fontsize{11}{11}\selectfont, thick]90:$(2,3)(1,3)$}] at (-2,4) (5) {};
\node[small dot new,label={[yshift=0.01cm, thick, font=\fontsize{11}{11}\selectfont, thick]90:$(2,3)$}] at (2,4) (6) {};
\path[draw,thick,-]
    (1) edge node {} (2)
    (1) edge node {} (3)
    (2) edge node {} (4)
    (3) edge node {} (5)
    (4) edge node {} (6)
    (5) edge node {} (6)
    ;
     \path[draw,dashed,thick,-]
    (2) edge node {} (5)
    ;
\end{tikzpicture}
\caption{Rough representation of $\mathcal C^{ETO}_d(p)$ for $d=(d_0,d_1,d_1)$ with $0<d_1 < d_0$ and $p=(a,b,c)$ with $b \geq c \geq \gamma a$. Note that we use the notation in Figure \ref{covex hull 1}.}
\label{covex hull weto 3}
\end{figure}
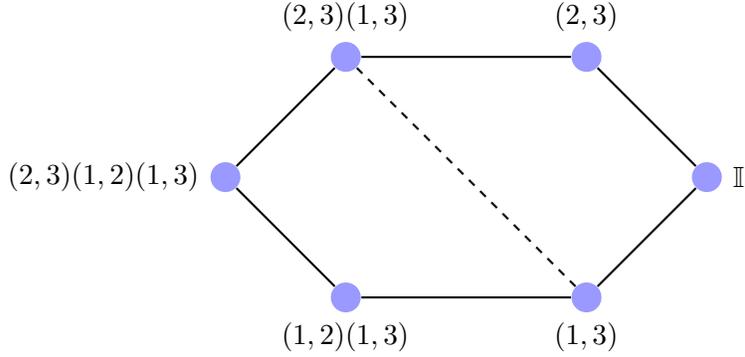

This concludes the proof.
\end{proof}

\begin{rem}
As a result of Proposition \ref{weto eto low dim}, when $|\Omega|=3$, conditioning our experimental protocols on the realization of some random variable results in the elementary thermal operations being closer to weak universality provided the system has exactly three different energy levels. 
\end{rem}

As a consequence of the results in this section, we pose the following conjecture regarding the general equivalence between ETO and WETO resource theories.

\begin{conj}
    [Equivalence ETO and WETO resource theories]
    \label{conjecture}
    If $0<d \in \mathbb P_\Omega$, then the following statements are equivalent:
 \begin{enumerate}[label=(\alph*)]
     \item The strong elementary thermal operations resource theory is equal to the weak elementary thermal operations 
resource theory
\begin{equation*}
    \text{RT}_{\text{ETO}}(d) = \text{RT}_{\text{WETO}}(d).
\end{equation*}
    \item $d$ is quasi-uniform.
 \end{enumerate}
\end{conj}

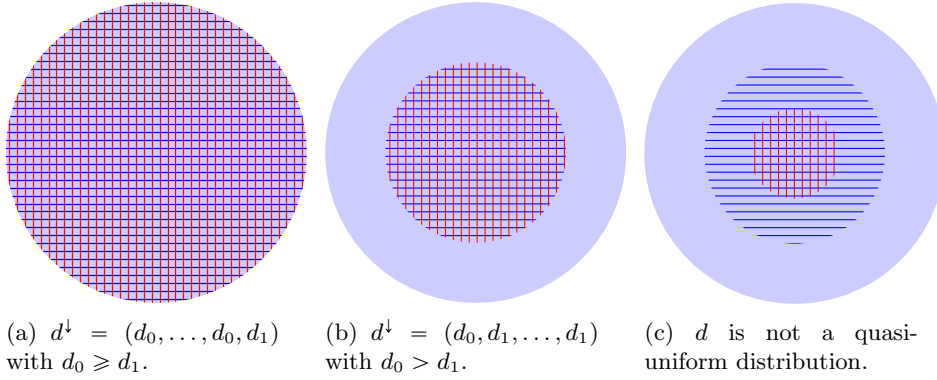
\begin{figure}[tb!]
\centering

\begin{subfigure}{0.3\textwidth}
  \centering
  \begin{tikzpicture}
    \fill [blue!20] (0,0) circle [radius=2cm];
    \fill [white, pattern=horizontal lines, pattern color=blue] (0,0) circle [radius=2cm];
    \fill [white, pattern=vertical lines, pattern color=red] (0,0) circle [radius=2cm];
  \end{tikzpicture}
  \caption{$d^\downarrow = (d_0,\dots,d_0,d_1)$ with $d_0 \geq d_1$.}
\end{subfigure}
\hfill
\begin{subfigure}{0.3\textwidth}
  \centering
  \begin{tikzpicture}
    \fill [blue!20] (0,0) circle [radius=2cm];
    \fill [white, pattern=horizontal lines, pattern color=blue] (0,0) circle [radius=1.2cm];
    \fill [white, pattern=vertical lines, pattern color=red] (0,0) circle [radius=1.2cm];
  \end{tikzpicture}
  \caption{$d^\downarrow=(d_0,d_1,\dots,d_1)$ with $d_0>d_1$.}
\end{subfigure}
\hfill
\begin{subfigure}{0.3\textwidth}
  \centering
  \begin{tikzpicture}
    \fill [blue!20] (0,0) circle [radius=2cm];
    \fill [white, pattern=horizontal lines, pattern color=blue] (0,0) circle [radius=1.2cm];
    \fill [white, pattern=vertical lines, pattern color=red] (0,0) circle [radius=0.6cm];
  \end{tikzpicture}
  \caption{$d$ is not a quasi-uniform distribution.}
\end{subfigure}
\caption{Venn diagram representation (for $|\Omega|=3$) of the relation among the thermal operations resource theories we have considered in this work depending on $0<d \in \mathbb P_\Omega$. In particular, the thermal polytope is filled in light blue and the strong (weak) elementary thermal operations has horizontal (vertical) blue (red) lines. The diagram looks the same for any $|\Omega| >3$ except for the fact we do not know whether weak elementary thermal operations are equivalent to strong elementary thermal operations in (b) in that scenario.
This figure summarizes Propositions \ref{difference weto eto} and \ref{weto eto low dim}, and Theorems \ref{polytopes relation} and \ref{wETO equiv}.}
\end{figure}

\subsection{Universality advantage}
\label{uni adv}

To conclude this section, we would like to determine the relation between the ETO and WETO polytopes. As we state in the following corollary, we know partly how they are related as a consequence of Theorems \ref{equi to eto poly} and \ref{equi to weto poly}.

\begin{coro}
\label{weak equi eto weto poly}
     If $0<d \in \mathbb P_\Omega$ with $d^\downarrow = (d_0,\dots,d_0,d_1)$,
     then the 
 following statements are equivalent:
 \begin{enumerate}[label=(\alph*)]
     \item The (strong) elementary thermal operations polytope is equal to the weak elementary thermal operations 
polytope
\begin{equation*}
    \mathcal P_{\text{ETO}}(d) = \mathcal P_{\text{WETO}}(d).
\end{equation*}
    \item $|\Omega| = 2$.
 \end{enumerate}
\end{coro}

In the following theorem, we characterize the equivalence between the ETO and WETO polytopes, improving thus on Corollary \ref{weak equi eto weto poly}.

\begin{theo}[Equivalence ETO and WETO polytopes]
\label{equi eto weto poly}
    If $0<d \in \mathbb P_\Omega$, then the 
 following statements are equivalent:
\begin{enumerate}[label=(\alph*)]
     \item The (strong) elementary thermal operations polytope is equal to the weak elementary thermal operations 
polytope
\begin{equation*}
    \mathcal P_{\text{ETO}}(d) = \mathcal P_{\text{WETO}}(d).
\end{equation*}
    \item $|\Omega| = 2$.
 \end{enumerate}
\end{theo}

\begin{proof}
By Lemma \ref{ordered d}, it suffices to assume that $d=d^\downarrow$ throughout the proof. Furthermore,
   we take $d_0 \coloneqq 1$ and $d_{|\Omega|+1} \coloneqq 0$.
   
   Necessity is straightforward by Lemma \ref{equiv dim 2}. To prove sufficiency, we argue by contrapositive. Hence, we take $|\Omega| \geq 3$ and argue on the number of jumps in $d$. In particular, we consider the following cases:
     \begin{enumerate}[label=(\Alph*)]
\item There exist some $1 \leq \alpha < |\Omega|$ and $2 \leq k \leq |\Omega|-\alpha$ such that $d_{\alpha-1} > d_\alpha= d_{\alpha+1} = \dots = d_{\alpha+k} > d_{\alpha+(k+1)}$. In this scenario, following the proof of Theorem \ref{equi to weto poly} for $|\Omega| \geq 4$, we have that
\begin{equation*}
    M \coloneqq M_0 \bigoplus \mathbb{I}_{\setminus (\alpha,\dots,\alpha+k)},
\end{equation*}
with $M_0 \in \mathcal M_{k+1, k+1}(\mathbb R)$ a doubly stochastic matrix such that $(M_0)_{i,j}>0$ if $i \neq j$ and $(M_0)_{i,j}=0$ if $i = j$, does not belong to the WETO polytope. (To show this, aside from Theorem \ref{equi to weto poly}, one can consider the number of ones along the main diagonal to conclude that, whenever we have some $M \in \mathcal M_{|\Omega|, |\Omega|}(\mathbb R)$ acting as the identity on some subset $\Omega' \subseteq \Omega$ and the entries of $d$ in $\Omega'$ differ from those in $\Omega \setminus \Omega'$, then we can restrict ourselves to the weak elementary thermal operations on $\Omega \setminus \Omega'$ to prove this result. This is also useful for proving the other cases.) However, by Theorem \ref{birkhoff}, $M$ does belong to the ETO polytope. 

\item There exist some $1 \leq \alpha \leq |\Omega|-2$ such that $d_{\alpha-1}> d_{\alpha} = d_{\alpha+1} > d_{\alpha+2} > d_{\alpha+3}$ or $d_{\alpha-1}> d_{\alpha} > d_{\alpha+1} = d_{\alpha+2} > d_{\alpha+3}$.
In this scenario, the first instance follows by Theorems \ref{equi to weto poly} (in the case $|\Omega|=3$) and \ref{difference to eto}. For the second instance, we take $0<\phi_0<1$ (sometimes a more restrictive choice of $\phi_0$ can make the argument easier) and consider 
\begin{equation}
\label{matrix eto weto}
    M \coloneqq \left(1-\phi_0 \right) P^d(\alpha,\alpha+1) \bigoplus \mathbb{I}_{\setminus (\alpha,\alpha+1)} + \phi_0 P^d(\alpha,\alpha+2) \bigoplus \mathbb{I}_{\setminus (\alpha,\alpha+2)}.
\end{equation} 
It is clear, by definition, that $M$ belongs to the ETO polytope.
To conclude this case, we can profit from \eqref{prod inte} and, analogously to \eqref{commutation}, from the fact that, whenever $d_{\alpha+1}=d_{\alpha+2}$ and  $0 \leq \lambda \leq 1$, then 
\begin{equation*}
    T^d_\lambda (\alpha,\alpha+2)=P^d(\alpha+1,\alpha+2) T^d_\lambda (\alpha,\alpha+1) P^d(\alpha+1,\alpha+2).
\end{equation*}
This allows us to conclude that $M$ is part of the WETO polytope if and only if there exist some $\ell,m \in \{0,1\}$ and $N \geq 0$ such that
\begin{equation*}
\begin{split}
    M=& \left( T^d_{\lambda_A} (\alpha+1,\alpha +2) \right)^\ell \left(\prod_{k=1}^N T^d_{\lambda_k} (\alpha,\alpha +1) T^d_{\beta_k} (\alpha+1,\alpha +2) \right) \\ &\times \left( T^d_{\lambda_B} (\alpha,\alpha +1) \right)^m,
    \end{split}
\end{equation*}
where $0 \leq \lambda_A,\lambda_B,\lambda_k,\beta_k\leq 1$ for $1 \leq k \leq N$.
It is then straightforward to check that the last equation is never satisfied. (For instance, one can argue by direct calculation using the zeros in $M$.)

\item There exist some $1 \leq \alpha \leq |\Omega|-2$ such that $d_{\alpha-1}> d_{\alpha} > d_{\alpha+1} > d_{\alpha+2} > d_{\alpha+3}$.
In this case, we can also consider $M$ as in \eqref{matrix eto weto} and argue similarly as in $(b)$ for most of the instances. (This is the case since in $(b)$ we usually argue on the number of zeros and, if certain components of some product of $T^d$-transforms are non-zero when $d_{\alpha+1}=d_{\alpha+2}$, then they will also be non-zero whenever we consider the analogous product $d_{\alpha+1}>d_{\alpha+2}$.) For the rest of the instances, one can directly check that the result holds.

\item $|\Omega|=2m$ for some $m \geq 2$ and $d_{2p-2} > d_{2p-1} = d_{2p} > d_{2p+1}$ for $1 \leq p\leq m$.
This case can be reduced to case $(b)$ since it only differs in the introduction of copies of a permutation matrix to \eqref{product}. In particular, we take $1 \leq \alpha < m$ and
\begin{equation*}
    M \coloneqq M_0 \bigoplus \mathbb{I}_{\setminus (2\alpha-1,2\alpha,2\alpha+1)},
\end{equation*}
with $M_0 \in \mathcal M_{3, 3}(\mathbb R)$ defined as in \eqref{dim 3 counterex}. Note that $M$ belongs to the ETO polytope by Theorem \ref{equi to eto poly}. To conclude, assume it belongs to the WETO polytope. In this case, it is easy to see that $M$ must have a decomposition like \eqref{product} with the possible inclusion of permutations $P \coloneqq P^d(2\alpha+1,2\alpha+2)$ along the sequence. We conclude showing that any such sequence reduces to \eqref{product} and, hence, $M$ is not in the WETO polytope by Theorem \ref{equi to weto poly}. Consider, thus, the first $P$ appearing in the sequence. Since $P$ commutes with $T^d_\lambda(2 \alpha-1,2\alpha)$, then it is followed either by some $T^d_\lambda(2 \alpha,2\alpha+1)$ or it is the leftmost matrix in the product. In any case, we can consider separately the case where the matrices to the right of $P$ acted trivially and non-trivially on the $2\alpha+1$ component and, using that $M_{2 \alpha+2,2 \alpha+2}=1$, conclude that either the sequence does not yield $M$ or it is equivalent to a sequence with a fewer number of $P$ matrices. Following this argument recursively we reach a sequence like \eqref{product} and obtain the desired conclusion.
\end{enumerate}
This concludes the proof.
\end{proof}

\begin{rem}
As a result of Theorem \ref{equi eto weto poly}, conditioning our experimental protocols on the realization of some random variable results in the elementary thermal operations being closer to weak universality provided $|\Omega| \geq 3$. 
\end{rem}

Note that, although some instances of the proof of Theorem \ref{equi eto weto poly} could be simplified via Proposition \ref{difference weto eto}, we take a simpler approach similar to the one we used when proving Theorem \ref{equi to weto poly}.

\begin{figure}[tb!]
\centering

\begin{subfigure}{0.3\textwidth}
  \centering
  \begin{tikzpicture}
    \fill [blue!20] (0,0) circle [radius=2cm];
    \fill [white, pattern=horizontal lines, pattern color=blue] (0,0) circle [radius=2cm];
    \fill [white, pattern=vertical lines, pattern color=red] (0,0) circle [radius=2cm];
  \end{tikzpicture}
  \caption{$|\Omega|=2$.}
\end{subfigure}
\hfill
\begin{subfigure}{0.3\textwidth}
  \centering
  \begin{tikzpicture}
    \fill [blue!20] (0,0) circle [radius=2cm];
    \fill [white, pattern=horizontal lines, pattern color=blue] (0,0) circle [radius=2cm];
    \fill [white, pattern=vertical lines, pattern color=red] (0,0) circle [radius=1.2cm];
  \end{tikzpicture}
  \caption{$d^\downarrow=(d_0,\dots,d_0,d_1)$.}
\end{subfigure}
\hfill
\begin{subfigure}{0.3\textwidth}
  \centering
  \begin{tikzpicture}
    \fill [blue!20] (0,0) circle [radius=2cm];
    \fill [white, pattern=horizontal lines, pattern color=blue] (0,0) circle [radius=1.2cm];
    \fill [white, pattern=vertical lines, pattern color=red] (0,0) circle [radius=0.6cm];
  \end{tikzpicture}
  \caption{$d^\downarrow \neq (d_0,\dots,d_0,d_1)$.}
\end{subfigure}
\caption{Venn diagram representation of the relation among the thermal polytopes we have considered in this work depending on $0<d \in \mathbb P_\Omega$. In particular, the thermal polytope is filled in light blue and the strong (weak) elementary thermal polytope has horizontal (vertical) blue (red) lines. This figure summarizes Theorems \ref{equi to eto poly}, \ref{equi to weto poly} and \ref{equi eto weto poly}.}
\end{figure}
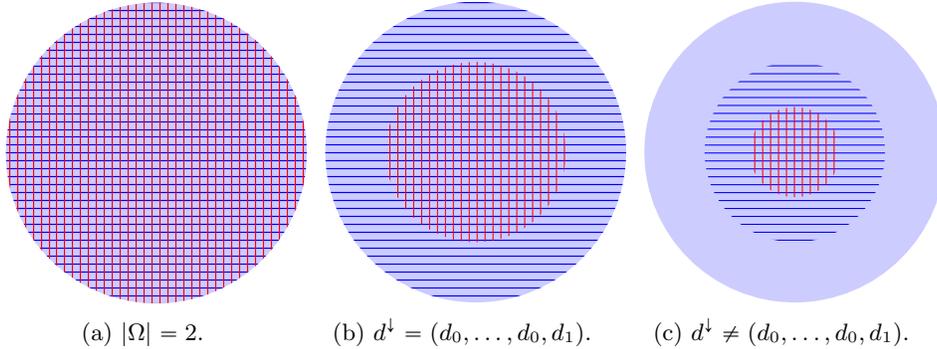

\section{Convexity of weak elementary thermal operations}
\label{sec: convex weto}

To conclude, we use the results in the previous section to address the convexity of weak elementary thermal operations.

As a consequence of our work on the relation between weak and strong elementary thermal operations resource theories in Propositions \ref{difference weto eto} and \ref{weto eto low dim}, we obtain the following straightforward result regarding convexity of weak elementary thermal operations.

\begin{coro}
[WETO resource theory convexity]
\label{convex WETO rt}
    If $0<d \in \mathbb P_\Omega$, then the 
 following statements hold:
\begin{enumerate}[label=(\alph*)]
     \item $\mathcal C_d^{WETO}(p)$ is convex for all $p \in \mathbb P_\Omega$ only if $d$ is quasi-uniform.
    \item If $|\Omega| = 3$, then $\mathcal C_d^{WETO}(p)$ is convex for all $p \in \mathbb P_\Omega$ if and only if $d$ is quasi-uniform.
 \end{enumerate}
\end{coro}

\begin{proof}
    \begin{enumerate}[label=(\alph*)]
        \item We can argue by contrapositive, assuming we have a non quasi-uniform distribution $d$ such that $d_1>d_2>d_3$ w.l.o.g. and noting that, if we take $(p,0), (q,0) \in \mathbb P_\Omega$ as in Proposition \ref{difference weto eto}, then
        \begin{equation*}
    (q,0) \in (\mathcal L(p),0) \coloneqq \left[\left(P^d(2,3)P^d(1,2)P^d(1,3)p,0\right),\left(P^d(1,2)p,0\right)\right]
\end{equation*}
with $\text{ext}((\mathcal L(p),0)) \subseteq \mathcal C_d^{WETO}((p,0))$ and $(q,0) \not \in \mathcal C_d^{WETO}((p,0))$.
        \item Necessity follows by (b), while sufficiency follows by Proposition \ref{weto eto low dim} since $\mathcal C_d^{WETO}(p)=\mathcal C_d^{ETO}(p)$ and the latter is convex by definition.
    \end{enumerate}
\end{proof}

Regarding polytopes, we  can use
Theorem \ref{equi eto weto poly} to  obtain a full characterization of the convexity of the WETO polytope.

\begin{coro}
[Characterization WETO polytope convexity]
\label{convex WETO}
    If $0<d \in \mathbb P_\Omega$, then the 
 following statements are equivalent:
\begin{enumerate}[label=(\alph*)]
     \item The polytope of weak elementary thermal operations is convex.
    \item $|\Omega| = 2$.
 \end{enumerate}
\end{coro}

\begin{proof}
By Lemma \ref{ordered d}, it suffices to assume that $d=d^\downarrow$.
 Necessity follows from Lemma \ref{equiv dim 2} since the TO polytope is convex and, in this instance, equal to the WETO polytope. To prove sufficiency, note that, provided $|\Omega|>2$, there exists some $M \in \mathcal M_{|\Omega|, |\Omega|}(\mathbb R)$ that belongs to the ETO polytope and not to the WETO polytope by Theorem \ref{equi eto weto poly}. Moreover, by definition,
    \begin{equation*}
        M = \sum_{k=1}^{k_0} \lambda_k \prod_{\ell=1}^{\ell_0} P^d(i_{k,\ell},j_{k,\ell}),
    \end{equation*}
    where $\sum_{k=1}^{k_0} \lambda_k =1$ and $\lambda_k \geq 0$ and $1 \leq i_{k,\ell} < j_{k,\ell} \leq |\Omega|$  for $1 \leq k \leq k_0$ and  $1 \leq \ell \leq \ell_0$. Lastly, since $\prod_{\ell=1}^{\ell_0} P^d(i_{k,\ell},j_{k,\ell})$ belongs to it for $1 \leq k \leq k_0$, the WETO polytope does not contain a convex combination of points in it. Hence, it is not convex. 
\end{proof}

As a direct consequence of Corollary \ref{convex WETO}, the WETO polytope is a polytope according to the definition in \cite{brondsted2012introduction} if and only if $|\Omega|=2$. Note that we can show, provided $d$ is not quasi-uniform, that $\mathcal P_{\text{WETO}}(d)$ is not convex using Corollary \ref{convex WETO rt}. This is the case since, for any such $d$, we can use Corollary \ref{convex WETO rt} to find some $p\in \mathbb P_\Omega$ such that $\mathcal C^{WETO}_d(p)$ is not convex, which is impossible provided $\mathcal P_{\text{ETO}}(d)$ is convex. (If $r,q \in \mathcal C^{WETO}_d(p)$ and $\mathcal P_{\text{ETO}}(d)$ is convex, then $(1-\lambda) q + \lambda r = (1-\lambda) M_0 p + \lambda M_1 p = M_2 p \in \mathcal C^{WETO}_d(p)$ for all $0 \leq \lambda \leq 1$, where $M_0,M_1,M_2 \in \mathcal P_{\text{ETO}}(d)$.)

\section{Conclusion}

The answers to our main questions are the following:

\begin{enumerate}[label=(Q\arabic*)]
    \item When are elementary thermal operations \emph{weakly} universal?

Both strong and weak elementary thermal operations are weakly universal if and only if $d^\downarrow=(d_0,\dots,d_0,d_1)$. Interestingly, the addition of random variables to our experimental protocols does not offer any advantage regarding the instances where elementary thermal operations are weakly universal.
    
    \item When are elementary thermal operations universal?

Strong elementary thermal operations are universal if and only if $d^\downarrow=(d_0,\dots,d_0,d_1)$. However, weak elementary thermal operations are universal if and only if $|\Omega|=2$. This illustrates the advantage of adding random variables to our experimental protocols, in contrast with the situation when dealing with weak universality.
    
    \item When do the incorporation of  random variables to the elementary thermal operations offer an advantage regarding weak or non-weak universality?

    Our answers to (Q1) and (Q2) show that the incorporation of random variables is not advantageous regarding the instances where elementary thermal operations are weakly universal, although it does augment the instances where they are universal. Moreover, even if they are not universal, protocols with random variables can realize more thermal operations than those without them in most cases. In particular, they cover more thermal operations if and only if $|\Omega|\geq 3$. Lastly, even if they are not weakly universal, protocols with random variables can reproduce more input-output pairs that are connected via thermal operations provided there are at least three different energy levels. In fact, these are the only situations where they offer such advantage provided $|\Omega| \leq 3$. However, it is still unknown whether there are instances where the advantage still holds with $|\Omega| > 3$ and only two different energy levels.  
\end{enumerate}

To conclude, let us point out a few directions of future research. The most immediate question is whether Conjecture \ref{conjecture} holds. This would conclude the classification of the thermal resource theories we have considered here and, from the technical side, potentially provide an easier proof of the classical result by Hardy et al. in Theorem \ref{equi to weto uniform}.

The study of thermal operations \eqref{quan to} for general $\rho$ continues to constitute a main challenge in the field, with only a few known results \cite{gour2018quantum}. Closer to our approach, the study of the evolution of coherence under elementary thermal operations also constitutes an area where little is known \cite{lostaglio2018elementary}. Moreover, the extremes of the ETO polytope and resource theory are still not well-understood, with only some (potentially bad) upper bounds being known \cite{lostaglio2018elementary,son2022catalysis}. Algorithmic cooling \cite{lostaglio2018elementary,alhambra2019heat} may benefit from out work, since some known algorithms \cite[Theorem 1]{alhambra2019heat} rely on thermal operations. (In particular, on \emph{$\beta$-permutations}, which are extremes of the TO polytope from Theorem \ref{extremes to}.)

We have characterized the main case of interest where we only allow two levels to act non-trivially. An obvious open question is how this changes whenever we progressively increase the number of levels where we are allowed to simultaneously act non-trivially. A first result showing that one cannot achieve all thermal operations on a system with $|\Omega|$ levels using only operations that act non-trivially on $|\Omega|-1$ of them was reported in \cite[Section III]{mazurek2018decomposability}, where a setup analogous to that in Proposition \ref{difference to eto} (a) (which we inherit from \cite[Corollary 5]{lostaglio2018elementary}) was used. The experimental relevance of doing so is, however, still unclear, since the experimental realization of such models becomes harder as the number of levels where non-trivial action is permitted increases.   

An application of our work here could be the establishment of the function characterization of $\preceq_d$ for $d^\downarrow = (d_0,\dots,d_0,d_1)$, that is, the characterization of all the functions that are allowed to be involved in a second law for $\preceq_d$. The key technical tool we have developed that could contribute in doing so is the so-called \emph{path} result in Theorem \ref{wETO equiv} (see \cite[p. 586 and p. 81]{marshall_inequalities:_2011} and also \cite[p. 45]{bhatia2013matrix}). This becomes more interesting when put together with the question in the previous paragraph, since the establishment of results analogous to the path one but involving larger proper subsets of $\Omega$ may lead us to the characterization of any function involved in a second law for arbitrary $d$. Lastly, the approach in \cite[p. 7]{de2022geometric} may also be useful to study the second law for $\preceq_d$ in general.

\paragraph{Acknowledgments}
This research is part of the Munich Quantum Valley, which is supported by the Bavarian state government with funds from the Hightech Agenda Bayern Plus.

\newpage
\bibliographystyle{unsrt}
\bibliography{main.bib}

\end{document}